\documentclass[a4paper,12pt]{amsart}

\usepackage{latexstyle}

\begin{document}

\title{The spatial evolution of economic activities and the emergence of cities}

\author{Davide Fiaschi$^a$}
\thanks{$^{a}$ Universit\'a degli Studi di Pisa, Dipartimento di Economia e Management.}
\address{D. Fiaschi, Universit\`a degli Studi di Pisa, Dipartimento di Economia e Management. Via Ridolfi 10, 56124 Pisa (PI), Italy.}
\email{davide.fiaschi@unipi.it}

\author{Cristiano Ricci$^b$}
\thanks{$^b$ Corresponding author. Universit\'a degli Studi di Pisa, Dipartimento di Economia e Management}
\address{C. Ricci, Universit\`a degli Studi di Pisa, Dipartimento di Economia e Management. Via Ridolfi 10, 56124 Pisa (PI), Italy.}
\email{cristiano.ricci@unipi.it}

\begin{abstract}
This paper examines the spatial agglomeration of workers and income in a continuous space-time framework. Local markets feature spatial spillovers and both exogenous and endogenous amenities. Workers relocate to maximise their instantaneous utility, constrained by mobility costs. In the limit of infinite workers, short-run equilibria are described by a partial differential equation (PDE). The PDE reveals spatial dynamics influenced by initial conditions, path dependence, and metastability (persistence), where prolonged stability is disrupted by sharp transitions to new distributions. We characterise conditions for spatial agglomeration in stationary equilibria and demonstrate that social utility consistently increases over time, suggesting efficient spatial allocations. Numerical results replicate key patterns, such as city formation, dependence on historical spatial patterns, and nonlinear out-of-equilibrium dynamics.
\end{abstract}

\maketitle



\newpage

\newpage

\section{Introduction \label{sec:introduction}}
This paper proposes a micro-foundation for the emergence of spatial agglomerations (cities) in an economy where individual mobility is driven by differential utility over a continuous space. The model incorporates spatial spillovers in both production and consumption. Space is heterogeneous due to the presence of exogenous and endogenous amenities—such as potential markets, infrastructures, and facilities—which are primarily shaped by population and economic density.

To illustrate the aim of our analysis, consider Figure \ref{fig:CountryMapsNightlights}, which shows the spatial distribution of nightlight intensity for three of the most important European countries.\footnote{Nightlights are taken from the VIIRS 2.1 database (\url{https://eogdata.mines.edu/products/vnl/}), which provides the average nightlight intensity for cells of 500 x 500 metres.}
\begin{figure}[!htbp]
\centering
\hspace{-2cm}
\begin{subfigure}{0.32\linewidth}
\includegraphics[width=1.7\linewidth]{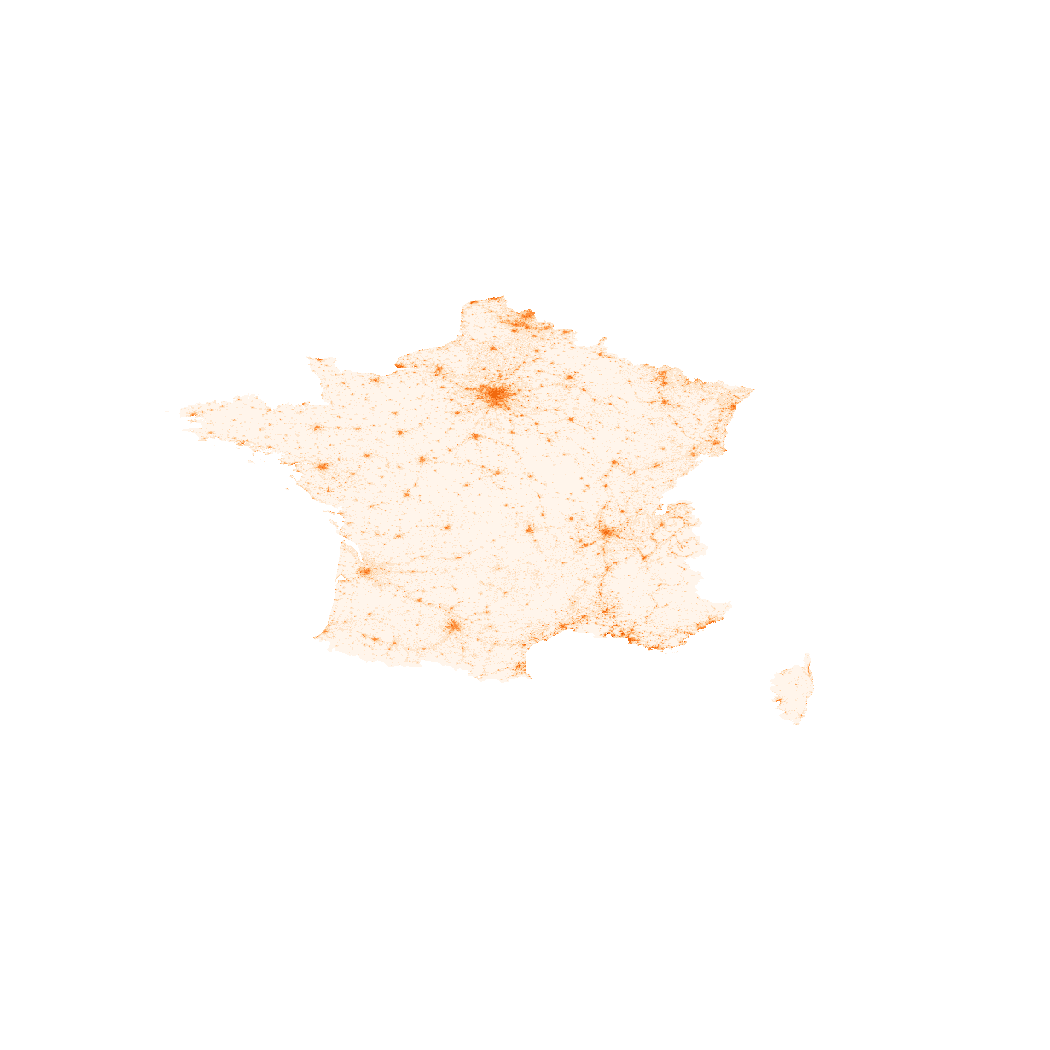}
\end{subfigure}
\hfill 
\begin{subfigure}{0.32\linewidth}
\includegraphics[width=1.5\linewidth]{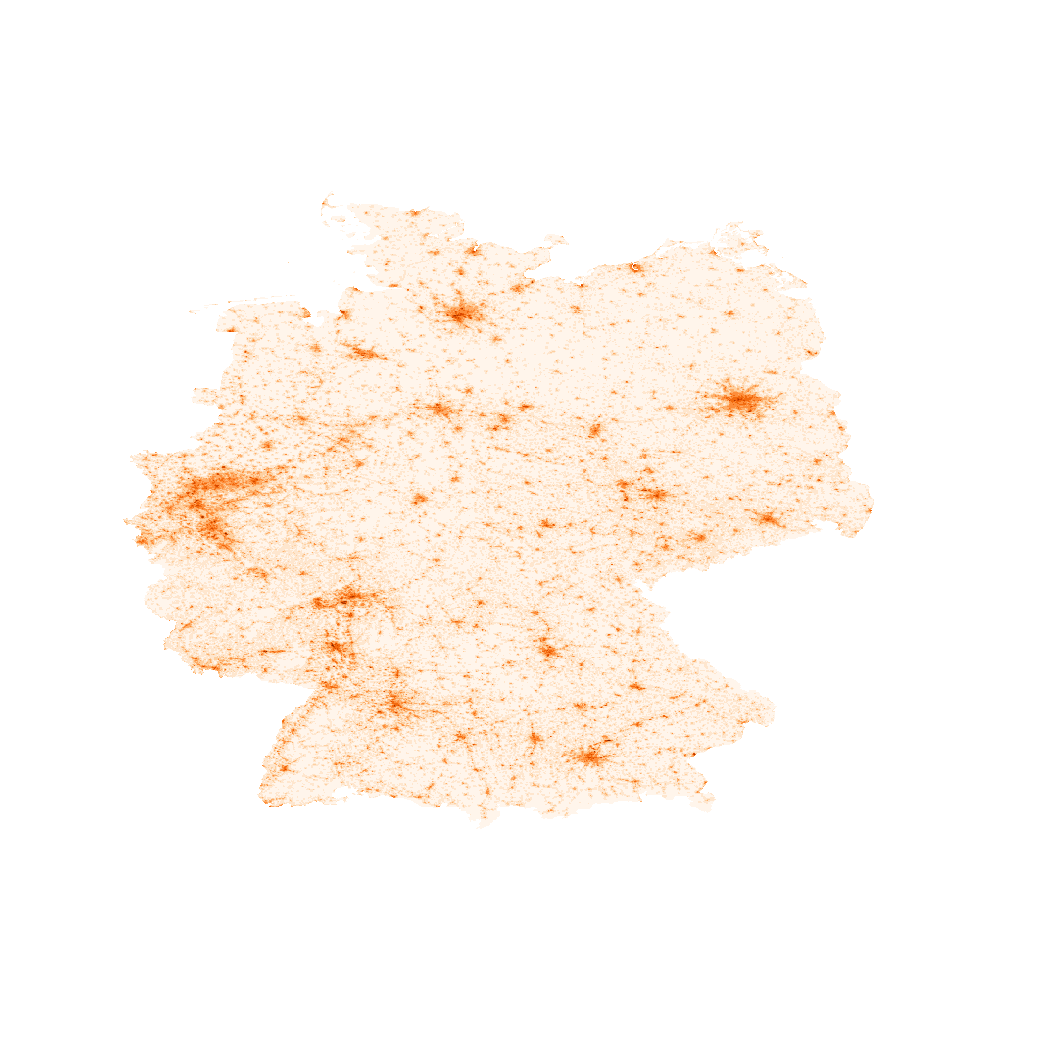}
\end{subfigure}
\hfill 
\hspace{-1cm}
\begin{subfigure}{0.32\linewidth}
\includegraphics[width=1.5\linewidth]{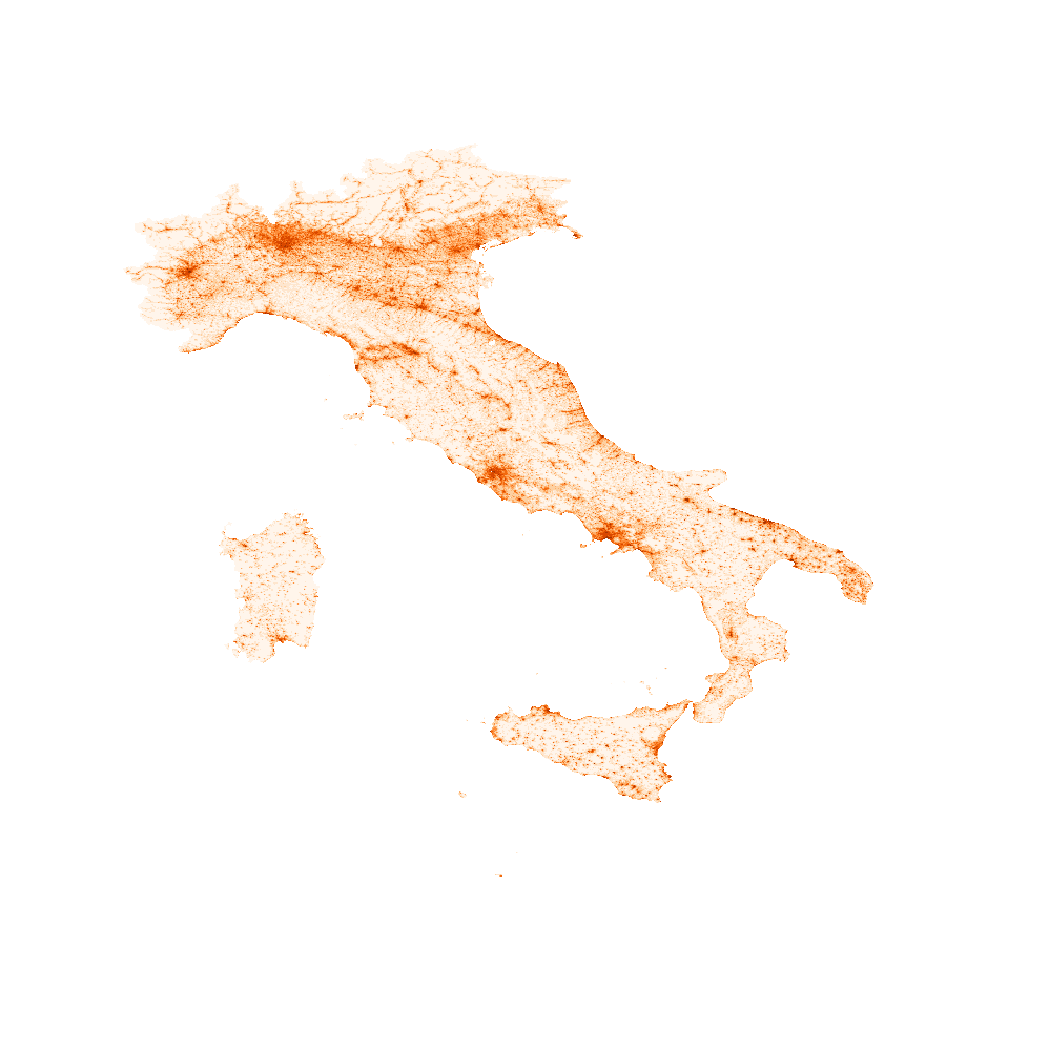}
\end{subfigure}
\caption{The maps of night lights in 2021 for France, Germany and Italy (from left to right), based on cells of about 500m x 500m. Source: VIIRS 2.1 database.}
\label{fig:CountryMapsNightlights}
\end{figure}
Using nightlight intensity as a proxy for local population/economic activity \citep{michalopoulos2018spatial}, there is a clear pattern of spatial agglomeration, reflecting the localisation of the main cities in the three countries. However, the spatial patterns are very heterogeneous: France, with a considerable homogeneity of its territory, shows an extreme degree of spatial agglomeration around Paris and only a few other clusters. On the other hand, Germany, which also has a territory with no particular geographical barriers, has a less concentrated distribution, with several medium-sized agglomerations spread evenly over the country. Finally, Italy, which has a very heterogeneous territory, shows a concentration of population in the coastal areas and in the Pianura Padana. In the latter, some medium-sized towns are located along the ancient Roman \textit{Via Emilia}, a narrow road that runs from Rimini along the Apennines to Piacenza. Several cities along this road were founded by the Romans more than two millennia ago (e.g. Rimini, Boulogne, Modena, Reggio Emilia, Parma, Piacenza) \citep{de2022historical}. Moreover, as the Italian population has grown over the last 30 years, several medium-sized Italian cities have experimented with a sharp reduction in population density in favour of a few larger cities.
To sum up, France, Germany and Italy have spatial agglomerations of different sizes and shapes, some agglomerations have been observed since the Roman Empire, and the topography of the territories has a strong influence on the observed spatial patterns. Finally, we add that over the last two millennia several European cities have experienced strong fluctuations in their size \citep{mumford1961city,pirenne2014medieval}.

To account for these different spatial patterns within a spatial general equilibrium model, we develop a continuous two-dimensional space model (locations are identified by their longitude and latitude), in line with the urban economics literature (see, for example, \citealp{Fujita1989}).
Local labour productivity is influenced by spatial spillovers, in particular by the \textit{local} density of labour. Locations are also characterised by the presence of exogenous and endogenous amenities, the latter depending on the \textit{local} density of workers and income \citep{GABAIX20042341,redding2024quantitative}. In their decisions, workers are myopic with respect to the future, but not with respect to the current location of other workers and the existence of spatial spillovers. In particular, they move across locations by maximising their instantaneous utility in the presence of movement costs and an idiosyncratic random component. We show that the sequence of these short-run workers allocations, in the \textit{Mean-Field} limit as the number of workers becomes infinite, can be expressed by a \textit{partial differential equation} (PDE) belonging to the class of \textit{Aggregation-Diffusion Equations}.
The stationary spatial equilibrium (if it exists) shows spatial agglomerations when aggregative forces sufficiently dominate diffusion. In the sequence of short-run worker allocations, social utility, as defined by the sum of individual utilities, is non-decreasing and hence spatial allocation is efficient in the stationary equilibrium distribution.
We conduct numerical experiments to investigate the properties of the model. In particular: i) the non-linear out-of-equilibrium dynamics; ii) the emergence of a single megacity versus the existence of several different cities \citep{GABAIX20042341}; iii) the importance of initial conditions, i.e. \textit{history}, and large shocks, i.e. \textit{path dependence}, for the stationary equilibrium distribution \citep{allen2020persistence}; and finally, iv) the phenomenon of \textit{metastability} i.e. \textit{persistence}, which implies long periods of apparent stability in the spatial distribution followed by a sharp transition to a new (meta)stable equilibrium \citep{allen2020persistence}.

We make three main contributions.
First, we develop a new framework with a sound micro-foundation to study some observed geographical patterns of particular interest, including population agglomeration, i.e. the emergence of cities and their evolution in continuous space, and the dynamics of spatial wage inequality in continuous time. The framework is sufficiently flexible to allow for inhomogeneous spaces due to the presence of mountains, rivers, coasts, roads, etc. Our choice of starting with individual agents and taking the limit of an infinite agent economy allows us to make transparent the interactions between the agent and its local environment, e.g. how neighbouring agents affect its productivity and how local endogenous amenities depend crucially on the local population density. A key implication of our approach is the ability to directly relate the parameters of the PDE describing the aggregate dynamics to the spatial extent of agents' interactions, the magnitude of agents' movement costs, and the characteristics of the random component in agents' choices.

Second, although there are no closed solutions to our PDE, we find in the general setting that there is dependence on initial conditions and path dependence in the localisation of spatial economic activity, but the shape of the stationary equilibrium distribution is unimodal (mono-centric) and radially symmetric in the absence of significant exogenous amenities. However, our results highlight how transitional dynamics contain particularly important information due to the presence of metastable equilibria and non-linear out-of-equilibrium dynamics, in contrast to the study of a linearised dynamical system around a long-run equilibrium.

Third, our approach allows to define a measure of social utility for the economy, inspired by the literature on \textit{Gradient Flow}. Social utility is non-decreasing over time, and in the stationary equilibrium distribution, the spatial distribution of workers reaches maximum efficiency, suggesting that only the expected utility of nearby workers provides sufficient information to steer an economy towards an efficient spatial allocation.

Our approach is very close to the literature on the disequilibrium foundations of general equilibrium models, where ``agents are aware of their disequilibrium state and act on arbitrage opportunities'' and ``production and consumption as being in real time'' \citep{fisher1989disequilibrium,howitt2006coordination} as opposed to the general spatial equilibrium model as in \cite{allen2014trade}.  For example, \cite{sonnenschein1982price} characterise the dynamics of convergence to the long-run general equilibrium as a sequence of short-run equilibria in which myopic firms change their output according to the differential rate of profit of markets close to their actual market. In particular, \cite{novshek1986quantity} discuss the case of a finite number of firms and as a sequence of short-run Cournot-Nash equilibria converging to a long-run (efficient) equilibrium, while \cite{artzner1986convergence} consider the case of an infinite number of firms. 

\cite{krugman1994complex,krugman1996self} represent the pioneering contributions to the study of the observed complex landscapes in the spatial distribution of population and economic activity (see in particular his ``edge city'' model). \cite{EconomyAsAComplexEvolvingSystemII1997} contains several key contributions in this respect, including those by Steven Durlauf on social-spatial interactions and by Paul Krugman on the self-organisation of economies in space. \cite{durlauf1994spillovers} examines the importance of local spillovers for the emergence of spatial inequality, while \cite{durlauf2005} review efforts to introduce complex systems methods into economics and to understand empirical phenomena.

Another important strand of literature is related to the formation of cities (spatial agglomeration), which points to spatial externalities in production, i.e. the condition for the existence of \textit{Zipf's law}, in particular on labour productivity \citep{Fujita1989,krugman1994complex,glaeser2001cities} and on the spatial externalities arising from urban density, which ``spreads knowledge, which either makes workers more skilled or entrepreneurs more productive'' \citep{glaeser2010complementarity}. 
\cite{davis2002bones} test three alternative theories of the determinants of the density of economic activity, while \cite{duranton2020economics} discuss the benefits and costs of urban density.
A spatial general equilibrium model without dynamics but very rich for the presence of spatial externalities, trade and labour mobility is presented by \cite{allen2014trade}. \cite{gabaix1999zipf} is also very close to our approach, but in a finite state space.
As the most recent quantitative urban economics, \cite{allen2020persistence} is the closest contribution to ours, being interested in the dependence on initial conditions, path dependence and persistence in the dynamics of spatial economic activities. 

\cite{sznitman1991topics} is the pioneering contribution to the understanding of the limiting behaviour of a system of interacting agents, while \cite{morale2005interacting} contain a possible micro-foundation of a model with aggregation. \cite{bovier2016metastability} discuss in detail the concept of metastability and its application. 
In the last 20 years, Aggregation-Diffusion Equations have inspired much mathematical work and have been used in several biological applications \citep{carrillo2019aggregation}. \cite{mocenni2010identifying} present a model very similar in spirit to ours, showing Turing patterns over space.
Finally, \cite{arbia2001modelling} models the geography of economic activity in continuous space.

Overall, we highlight three main departures of our work from most of the recent urban economics literature (see e.g. \citealp{redding2024quantitative}).
First, we assume that consumption and production decisions are made in continuous time. As argued by \citet[p.547]{gandolfo1997economic}, although individual decisions are generally made at discrete time intervals, they are not coordinated and are therefore unlikely to be synchronised; instead, they are likely to overlap, and considering an economy in which a large number of decisions are made by a large number of agents, it seems natural to treat economic phenomena as if they were continuous.  Moreover, a continuous-time specification is particularly useful for formulating a dynamic adjustment process when both stocks (e.g. labour stocks) and flows (e.g. local wages) are involved. A side effect, more related to empirical application, is that the parameters' of a continuous time model can be estimated/calibrated independently of the observation interval. The continuous time specification also avoids the risk that a fixed period length assumed for each agent's decision (e.g. workers and firms may have different time scaling in their decisions) may lead to misleading conclusions, including confusion between stock and flow equilibria. Finally, the continuous time formulation is somewhat easier to handle than difference systems and allows the calculation of continuous paths for the endogenous variables, which can have relevant implications for numerical experiments and policy analysis. The cost of considering continuous time is the special care that needs to be taken when running the model on data that are generally available in discrete time \citep{fiaschiParentiRicci2023}.

Second, the space is assumed to be continuous as in \cite{allen2014trade} and \cite{desmet2018geography}, the latter being the more appropriate setting for the study of the spatial allocation process, where administrative barriers should play a negligible role and the interactions of agents in local (labour and goods) markets depend on the location of each agent (consumer, worker, firm) with respect to the locations of the other agents. This is particularly important in the empirical application, where our model can be adapted to any administrative level of observation (e.g. counties, states, countries) without any significant change in the parameterisation of spatial spillovers. Given these advantages, the cost is the difficulty of framing workers' commuting and migration choices in an intertemporal framework, as in \citet{caliendo2019trade} and \citet{kleinman2023dynamic}. 

Third, the focus on transitional dynamics, which shows that two key phenomena cannot be ignored: i) the non-linearity in the distribution dynamics, i.e. the possibility that some locations show a first phase of increasing followed by a next phase of decreasing labour and income density; and ii) metastability, i.e. the possibility that the spatial distribution of economic activity shows a long period of stability but is followed by sharp transitions. In other words, on a short time scale the economy reaches some kind of (apparent) equilibrium, while on a long time scale it moves from one of these (meta) equilibria to the other. A recent contribution to transitional dynamics is \citet{allen2020persistence}, which, however, do not consider metastability, whereas \cite{caliendo2019trade} and \cite{kleinman2023dynamic} do not address the non-linearity in distribution dynamics.

The paper is structured as follows. Section \ref{sec:model} illustrates the micro-foundations of the workers' movement and presents the main result on the dynamics in the limit of infinite workers; Section \ref{sec:localUtilitFunction} describes the specific setting of our economy and derives its aggregate behaviour; Section \ref{sec:propertiesEconomy} contains the main theoretical results on the dynamics; Section \ref{sec:numericalExplorations} contains some numerical explorations; and Section \ref{sec:conclusions} concludes. The Appendix contains the proofs. The Online Appendix contains some extensions. 

\section{Spatial dynamics of workers}\label{sec:model}
In this section we discuss worker dynamics in two-dimensional continuous space.
In particular,  we present a flexible theoretical framework in terms of the specification of individual preferences, technology and spatial spillovers (leaving their specific characterisation to Section \ref{sec:localUtilitFunction}), where workers move according to differential utilities over space and the dynamics of all workers can be modelled in a parsimonous way by a PDE. 

\noindent Suppose we have an economy with $N$ workers, and denote the spatial location of worker $i$ at each time $t$ by $X^{i,N}_{t}$, and the location of all other workers (including itself) by $\mathbf{X}^{N}_{t} = \left(X^{1,N}_{t},X^{2,N}_{t},\dots,X^{N,N}_{t}\right)$. Let $u^{i}$ be the utility function of worker $i$, which depends on the location of all other workers $\mathbf{X}^{N}_{t}$, time $t$, its location in space $X^{i,N}_{t}$, and some unobserved personal preferences. The movement of worker $i$ is given by
\begin{equation}\label{eq:movementAgent}
dX^{i,N}_{t} = \left( \frac{1}{c_{M}} \right) \nabla_{x} u^{i}\left(\mathbf{X}_t^N\right)\left(X^{i,N}_{t},t\right) dt,
\end{equation}
where $\nabla_x u^{i}$ is the \textit{gradient}, i.e. in a two-dimensional space, $\nabla_x u^{i} \equiv \left(\partial_{x_1} u^{i},\partial_{x_2} u^{i}\right)$ ($x_1$ and $x_2$ are the two directions). It indicates the direction in which worker $i$ should move to increase its utility, given its location $X^{i,N}_{t}$ and the location of all $N$ workers $\mathbf{X}_t^N$. 
In the Online Appendix \ref{app:NashEquilibriumAgentsMovement} we show that the movement of the worker in Eq. \eqref{eq:movementAgent} can be derived as a Nash equilibrium. The factor $(1/c_{M})$ represents the inverse cost of movement for the workers and is also discussed in the Online Appendix \ref{app:NashEquilibriumAgentsMovement}. Migration over long distances (e.g. between countries) is not well represented by this type of movement cost, but it is marginal for the phenomena we are interested in this paper.

Eq. \eqref{eq:movementAgent} represents the minimal setup for workers who seek higher utility by moving through space, but are subject to movement costs (frictions) and do not coordinate their choices by taking $\mathbf{X}_t^N$ as given. \cite{howitt2006coordination}, \cite{lebaron2008modeling}, \cite{colander2008beyond} and \cite{aoki2011reconstructing} provide a detailed discussion of the implications of this approach in macroeconomics, which is called the \textit{worker-based model}. The latter is a complement/alternative approach to the framework with optimising forward-looking (in the sense of space and time) workers. In particular, the law of worker migration in our model, given by Eq. \eqref{eq:movementAgent}, can be traced to the one proposed in \citet[p.748]{caliendo2019trade} or \citet[p.12]{allen2020persistence} under the hypothesis of an infinite discount rate and an appropriate choice of movement costs and random component of utility.\footnote{This statement can be made rigorous by using Theorem 2.1 in \citet{bardi2021convergence}.}

According to Eq. \eqref{eq:movementAgent} at time $t$ the worker $i$ does not move if its utility in the neighbourhood of its location $X^{i,N}$ is equal to its actual location, i.e. $\nabla_{x} u^{i}\left(\mathbf{X}_t^N\right)\left(X^{i,N}_{t},t\right)=0$. The equilibrium is therefore defined as the condition when $ \nabla_{x} u^{i}\left(\mathbf{X}_t^N\right)\left(
X^{i,N}_{t},t\right)=0 \; \forall i$, which allows for the possibility of i) \textit{Pareto dominated equilibria} due to lack of worker coordination \citep{howitt2006coordination} and ii) equilibria characterised by \textit{complex geographical landscape} \citep{krugman1994complex}. Also importantly, this framework can generate \textit{non-linear out-of-equilibrium dynamics} as shown in \cite{krugman1994complex,krugman1996self}.
All these features have empirical support and are one of the main advantages of this approach. We will show below that another advantage is the flexibility to take into account in the analysis some key features of the urban/regional economy, such as the presence of spatial spillovers in production (spatial increasing returns) and consumption (endogenous amenities), together with a non-uniform space (exogenous amenities and natural barriers), while maintaining relative analytical tractability.

\subsection{Workers' preferences}
The form of workers' utility is taken by the \textit{Random Utility Model} framework (\citealp{train2009discrete}):
\begin{equation}
u^{i}\left(\mathbf{X}_t^N\right)\left(x,t\right) := v^{N}\left(\mathbf{X}_t^N\right)\left(x,t\right) +\sigma x \cdot \frac{dB^{i}_{t}}{dt} ,
\label{eq:randomUtility}
\end{equation}
i.e. the utility is given by the linear combination of a \textit{systematic} (deterministic) component and a \textit{random} component $\sigma x \cdot dB^{i}_{t}/dt$, where $(B^i_t)_{i \in \NN}$ is a sequence of {independent Brownian motions} and $\sigma$ is a scaling factor for the variance of the random component. 

The former is a utility function common to all workers, defined as follows:
\begin{eqnarray*}
\label{eq:systematicUtilityFun}
v^{N}:&\RR^{2N} & \to \big[ \RR^{2} \times \RR^{+} \to \RR\big]\nonumber \\
&\mathbf{X}&\mapsto \big[(x,t) \mapsto v^{N}(\mathbf{X})(x,t)\big],
\end{eqnarray*}
that is, for every $\mathbf{X} := ({X}_{1},\dots,{X}_{N})\in \RR^{2N}$ representing the location of all workers, there exists a \textit{field of utility} $v^{N}(\mathbf{X}):\RR^{2} \times \RR^{+} \to \RR$ which associates, at each point $x$ and $t \in \RR$, the corresponding utility. 
The systematic utility of worker $i$, given the configuration of all workers $\mathbf{X}^{N}_{t}$, is given by $v^{N}\left(\mathbf{X}^{N}_{t}\right)(X^{i,N}_{t},t)$. The first argument of $v$, $\mathbf{X}^{N}$, is the overall form of utility determined by the spatial distribution of workers; the second, $X^{i,N}_{t}$, assigns to each worker the level of utility that corresponds to its location; finally the third component, $t$, is the overall change in utility due to technology and other time-varying exogenous factors (e.g. technological progress, climate change, etc.).
The systematic utility measures the worker's utility of being in a particular location $x$, which reflects the wage, amenities and any other local characteristics. In other words, $v$ is the \textit{indirect utility} of being in $x$. 

The second component of Eq. \eqref{eq:randomUtility} instead reflects the \textit{idiosyncratic preferences} of worker $i$ for location $x$ and is given by a scalar product between location $x$ and $dB^{i}_{t}/dt$, which (at least formally) represents the instantaneous variation of a stochastic process with Gaussian independent increments over disjoint time intervals with zero mean and variance $dt$. The location $x$ is inserted as a multiplicative factor to make the variance of the random component of the utility dependent on the distance between two locations.
Unbounded noise is commonly used in stochastic modelling (mainly Gaussian); more realistic noise should require boundedness, but to our knowledge this prevents the use of standard stochastic differential equation techniques \citep{domingo2020properties}. 
Our instantaneous utility in Eq. \eqref{eq:randomUtility} is very similar to that used in \citet[p.748]{caliendo2019trade} and/or \citet[p.12]{allen2020persistence}, but with a different shape of the stochastic component.

Finally, from Eqq. \eqref{eq:movementAgent} and \eqref{eq:randomUtility}:
\begin{equation} 
dX^{i,N}_{t} = \left(\frac{1}{c_M}\right)\nabla_{x} v^{N}\left(\mathbf{X}_t^N\right)\left(X^{i,N}_{t},t\right) \, dt +\left(\frac{1}{c_M}\right) \sigma dB^{i}_{t}.
\label{eq:movementAgent_II}
\end{equation}
Eq. \eqref{eq:movementAgent_II} shows that worker $i$ moves to nearby locations where systematic utility is higher, plus a random individual preference for different locations.

\subsection{Limit of infinite workers\label{subsec:spatialDynamicsPopulationInfiniteEconomy}}

Theorem \ref{teo:limitInfiniteAgents} below, states that in the limit of infinite workers the dynamics at the aggregate level can be described by a PDE.
\begin{teo}[Spatial Dynamics of Workers]\label{teo:limitInfiniteAgents}
Assume that Eq. \eqref{eq:movementAgent} describes the movement of the worker. Then, as $N$ tends to infinity, the stock of workers at location $x$ at time $t$ converges in probability to the unique solution of
\begin{equation}\label{eq:dynamicsLabourDistribution} 
\partial_t l(x,t) = \dfrac{\sigma^2}{2c_M^2} \Delta_x l\left(x,t\right) - \left(\frac{1}{c_M}\right) \, \div_x \left( l(x,t) \nabla_x v\left(x,t\right) \right) 
\end{equation}
or, in short
\begin{equation}\label{eq:dynamicsLabourDistributionUtility}
\partial_t l(x,t) = - \left(\frac{1}{c_M}\right) \, \div_x \left( l(x,t) \nabla_x u(x,t)\right),
\end{equation}
where:
\begin{equation}\label{eq:teoTotalUtility}
u(x,t) \equiv v(x,t) - \left(\dfrac{\sigma^2}{2 c_M} \right)\log l(x,t).
\end{equation}
\end{teo}
\begin{proof}
For a sketch of the proof, see Appendix \ref{app:proofMeanFieldLimit}.
\end{proof}
\noindent From a mathematical point of view, Theorem \ref{teo:limitInfiniteAgents} is a generalisation of the more classical theory of the \emph{Fokker-Planck equation}, in which the aggregate behaviour of a set of \textit{independent} processes following an exogenous common law is derived as the solution to a PDE. However, our case is much more complicated, since workers follow \textit{non-independent} processes as they interact through the utility function in each period.

From an economic point of view, Eq. \eqref{eq:dynamicsLabourDistributionUtility} expresses the dynamics of the distribution as a result of the movement of workers to maximise the utility in Eq. \eqref{eq:teoTotalUtility}, analogous to the dynamics of individual workers expressed in Eq. \eqref{eq:movementAgent}. The expected utility of a worker at location $x$ in period $t$, $u(x,t)$, is the sum of a systematic utility $v(x,t)$ plus a random component whose aggregate result is given by $- \left(\sigma^2/\left(2 c_M\right) \right)\log l(x,t)$. In other words, the random component in the movement of workers at the individual level induces the presence of diffusion at the aggregate level, i.e. the first term on the right-hand side of Eq. \eqref{eq:dynamicsLabourDistribution}, or the second (\textit{entropic}) term in the utility function in Eq. \eqref{eq:teoTotalUtility}. 

The generality of Theorem \ref{teo:limitInfiniteAgents} allows to consider several economic features proposed in urban economics (e.g. spatial spillovers in consumption and in the labour market, see \citealp{Fujita1989}) without specific constraints (e.g. on the shape of the utility function, the production function and technological spillovers). At the same time, we can use a corpus of mathematical techniques specifically designed to study PDE's of the same type as Eq. \eqref{eq:dynamicsLabourDistribution}  \citep{morale2005interacting}.

\section{Environment\label{sec:localUtilitFunction}}
In this section, we specify agents' utility, production and amenities according to some key features that characterise urban economies \citep{Fujita1989}. In particular, we consider an economy with no capital accumulation, local technology and labour markets, but characterised by spatial spillovers. We also consider two types of amenities, the first exogenous and the second depending on local income and labour density.

Suppose that there is no accumulation of physical capital in the economy, i.e. no savings. The level of consumption is therefore equal to the level of workers' real wage $w^N$:
\begin{equation}\label{eq:consumption}
c^N\left(\mathbf{X}_t^N,t\right)\left(X^{i,N}_{t}\right) = \dfrac{\tilde{w}^N\left(\mathbf{X}_t^N,t\right)\left(X^{i, N}_{t}\right)}{P^N\left(\mathbf{X}_t^N,t\right)\left(X^{i,N}_{t}\right)}= w^N\left(\mathbf{X}_t^N,t\right)\left(X^{i,N}_{t}\right),
\end{equation}
where $c^N\left(\mathbf{X}_t^N,t\right)\left(X^{i,N}_{t}\right)$ is the consumption of worker $i$, $\tilde{w}^N\left(\mathbf{X}_t^N,t\right)\left(X^{i, N}_{t}\right)$ its nominal wage and $P^N\left(\mathbf{X}_t^N,t\right)\left(X^{i,N}_{t}\right)$ the price level of the place of residence of worker $i$. 
In Eq. \eqref{eq:consumption} we implicitly assume that the \textit{place of residence} is also the \textit{place of work} of the workers, i.e. we set aside the issue of \textit{commuting}, which would require also taking into account local markets for goods and trade costs. 
Moreover, workers' utility also depends on the amenities specific to the location of residence; in particular, these amenities can be both exogenous, such as weather conditions, or endogenous, such as the presence of facilities, schools, hospitals, pollution, traffic congestion, etc., which depend on local production and population density.

The form of workers' utility is assumed to be equal to
\begin{equation}
u^{i}\left(\mathbf{X}_t^N,t\right)\left(x\right) : = w^N\left(\mathbf{X}_t^N,t\right)\left(x\right) + A^N_{EN}\left(\mathbf{X}_t^N,t\right)\left(x\right) + A_{ES}\left(x,t\right) +\sigma x \cdot \frac{dB^{i}_{t}}{dt} ,
\label{eq:randomUtilityEcon}
\end{equation}
i.e. utility is linear in consumption, includes two additional terms for \textit{endogenous and exogenous amenities}, $A^N_{EN}\left(\mathbf{X}_t^N, t\right)\left(X^{i,N}_{t}\right)$ and $A_{ES}\left(X^{i,N}_{t},t\right)$ respectively, and a random component $\sigma x \cdot dB^{i}_{t}/dt$. Eq. \eqref{eq:randomUtilityEcon} is the counterpart of the logarithmic utility assumed for example in \cite{caliendo2019trade}, \cite{allen2020persistence} and \cite{redding2024quantitative}.
In conclusion, from Eqq. \eqref{eq:movementAgent} and \eqref{eq:randomUtilityEcon}, the worker's movement is described by:
\begin{equation} 
dX^{i,N}_{t} \hspace{-0.1cm}= \hspace{-0.1cm}\left(\hspace{-0.1cm}\frac{1}{c_M}\hspace{-0.1cm}\right)\hspace{-0.1cm}\left[ \nabla_{x} w^N\left(\hspace{-0.05cm}\mathbf{X}_t^N,t\hspace{-0.05cm}\right)\left(\hspace{-0.05cm}X^{i,N}_{t}\hspace{-0.05cm}\right) \hspace{-0.1cm}+\hspace{-0.1cm}  \nabla_{x}A^N_{EN}\left(\hspace{-0.05cm}\mathbf{X}_t^N,t\hspace{-0.05cm}\right)\left(\hspace{-0.05cm}X^{i,N}_{t}\hspace{-0.05cm}\right) \hspace{-0.1cm}+\hspace{-0.1cm} \nabla_{x}A_{ES}\left(\hspace{-0.05cm}X^{i,N}_{t},t\hspace{-0.05cm}\right) \right]dt  +\left(\hspace{-0.1cm}\frac{1}{c_M}\hspace{-0.1cm}\right) \sigma dB^{i}_{t}.
\label{eq:movementAgent_IIEcon}
\end{equation} 

\subsection{Local labour markets\label{sec:localWage}}

Production at location $x$ in period $t$ in an economy with $N$ workers, $y^N(x,t)$, is defined by
\begin{equation}
y^N(x,t) = A_l^N(x,t) l^N(x,t)^\beta,
\label{eq:localProductionFunnction}
\end{equation}
with $\beta \in \left(0,1\right]$, where $A_l^N(z,t)$ is a \textit{local-specific technological progress} affecting the\textit{ marginal productivity} of labour, and $l^N(x,t)$ is the \textit{stock of labour} available at location $x$ at period $t$. \cite{allen2014trade} discuss how Eq. \eqref{eq:localProductionFunnction} is compatible with the presence of land as a factor of production (in this case $\beta \in \left(0,1\right)$). Therefore the wage is given by
\begin{eqnarray}
w^N(x,t) &=& A_l^N(x,t) l^{N}(x,t)^{\beta-1}.
\label{eq:wages}
\end{eqnarray}
In the case of $\beta \in \left(0,1\right)$, Eq. \eqref{eq:wages} points to a \textit{congestion effect} in the labour market due to the decreasing marginal productivity of labour, which, ceteris paribus, attracts workers away from locations with higher labour densities. This is different from \citet{allen2020persistence}, who adopt a linear technology, and more similar to \citet{allen2014trade} and \citet{redding2024quantitative}, where this effect is justified by the need for the production of a certain amount of of land.

A technical difficulty in our framework, where space is continuous and the number of workers is finite, is that each worker's participation in local labour markets, i.e. the stock of labour available at location $x$ in period $t$, is modelled as the sum of the differential contributions of all workers within a given distance.  In particular, we assume that the individual contributions are determined by the (geographical) distance of the workers from the labour market locations and are conveniently modelled by the \textit{kernel function} $\theta_N\left(\cdot\right)$. Thus the stock of labour at location $x$ at time $t$ is defined by
\begin{equation}\label{eq:lN}
l^N(x,t) = \sum_{i = 1}^{N} \frac{1}{N}\cdot\theta_N(x-X^{i,N}_t),
\end{equation}
where $\theta_{N}(z) := {h_N}^{-2} \cdot \theta\left({z}{h_N}^{-1}\right)$, and $h_N$ tends to zero as $N$ goes to infinity. In particular, we assume that $h_N := N^{-\lambda}$ and $\lambda \in (0,1/4]$, which is crucial to derive the aggregate dynamics as a PDE as discussed in Appendix \ref{app:proofMeanFieldLimit}. The function $\theta (\cdot)$ is assumed to have a compact support with a unitary radius, to integrate to one, and to be symmetric and non-increasing with respect to the distance from the origin. 
Eq. \eqref{eq:lN} thus appears as a weighted sum of the individual endowments of labour of each worker, where the weights are determined by the distance of each worker from location $x$, with weights equal to zero for a distance greater than $h_N$. 

Given $h_N$, as $N$ becomes large, the interconnectedness between local labour markets increases because the number of workers participating in different markets increases with worker density. The decrease of $h_N$ in $N$, measured by $\lambda$, ``offsets'' this phenomenon. The limiting case of $\lambda =0$, i.e. $h_N=1$, corresponds to the {Mean-Field} limit, which has the very counterfactual implication that a location $x$ with zero inhabitants could have positive production if a populated location is within distance 1 from $x$.  For $\lambda >1$, in the limit of $N$ going to infinity, any interconnection between local labour markets disappears, which would eliminate an important source of spatial spillovers. Our case of $\lambda \in (0,1/4]$ instead implies that as the number of workers $N$ goes to infinity, neighbouring local labour markets are still connected, but local positive production requires the presence of some resident workers at the location. For $\lambda \in (1/4,1]$, we conjecture that the aggregate dynamics can always be modelled by a PDE, but a formal proof cannot be derived.

To summarise, in our economy each worker participates in local labour markets no further away than $h_N$, and the latter also reflects the influence of distance on the participation of workers in local labour markets. The labour market can be seen as composed of many interconnected \textit{local} labour markets (one for each location $x$), where the intensity of the interconnection is directly related to $h_N$. However, as the number of workers $N$ goes to infinity, this intensity goes to zero and therefore $h_N$ does not appear in the PDE. Finally, two neighboured local labour markets still result interconnected also in the limit. The latter is a distinctive feature compared to the case with discrete space and independent local labour markets (ignoring bilateral trade) typically considered in urban economics \citep{allen2014trade, allen2020persistence}.

\subsection{Technological progress\label{sec:localTechProgress}}
Regarding local technological progress $A_l^N(x,t)$ in the presence of positive spatial externalities, we assume, following \cite{papageorgiou1983agglomeration}, \cite{glaeser2010complementarity} and \cite{glaeser2001cities} among others, that: 
\begin{equation}\label{eq:laborProductivity}
A_l^N(x,t) = G(x)\sum_{i=1}^{N} \frac{1}{N}W_h^P(x-X^{i,N}_t),
\end{equation}
where $W_h^{P}\left(z\right):= \left(1/h^2\right) W^P\left(z/h\right)$ is a kernel function assumed to have compact support with radius (bandwidth) $h>0$, integrating to one, symmetric with respect to the origin and non-increasing with distance from the origin, and $G(x)$ is a non-negative function defined on the space $\Omega$. 
Local technological progress is therefore the product of the two terms: i) $G(x)$, the \textit{potential use} of location $x$ for production, which can be traced back to the possibility of using this specific location for productive purposes (e.g. those locations where production is almost impossible, such as in the middle of rivers or the open sea, should have $G(x)\approx 0$); ii) $\sum_{i=1}^{N} \frac{1}{N}W_h^P(x-X^{i,N}_t)$, a measure of the local density of workers, where each worker is considered for its endowment of labour $1/N$ and its proximity to location $x$ is modelled by the kernel function $W_h^P$. The bandwidth $h$ reflects the extension of spatial spillovers, i.e. higher $h$ means wider spatial spillovers. 
Eq. \eqref{eq:laborProductivity} implies that the \textit{positive spatial externalities} in production at location $x$ depend not only on workers at $x$, as in \cite{allen2014trade,allen2020persistence}, but also on all workers within a distance $h$ from $x$, as in \cite{redding2024quantitative}.
Therefore, production and wage at location $x$ are given by
\begin{eqnarray}
y^N(x,t) &=& G(x)\left[\sum_{i=1}^{N} \frac{1}{N}W_h^P(x-X^{i,N}_t)\right] l^{N}(x,t)^\beta, \text{ and} \label{eq:localIncome}\\
w^N(x,t) &=& G(x)\left[\sum_{i=1}^{N} \frac{1}{N}W_h^P(x-X^{i,N}_t) \right]l^{N}(x,t)^{\beta-1}.
\label{eq:ProductionWages}
\end{eqnarray} 

\subsection{Amenities\label{sec:localAmenities}}

\textit{Local amenities} can be both exogenous and endogenous. 
The \textit{exogenous amenities} are related to specific characteristics of the location $x$ independent of the distribution of workers over space. Typical examples are climate, weather, natural landscape, etc. All these characteristics can be considered exogenous to the overall dynamics of population distribution and can be modelled as a function $A_{ES}(x,t)$. The shape of $A_{ES}$ reflects spatial differences in individual utility caused by specific characteristics of each location $x$. The extreme case is the sea, rivers and/or high mountains, where the function $A_{ES}$ should signal an extremely low utility for people living in those places. 

The \textit{endogenous amenities} require more sophisticated modelling because they are affected by local population and income. Typically, endogenous amenities include local services (e.g. schools, hospitals and theatres), local goods (e.g. shopping centres) and any other infrastructure and facilities that are \textit{locally supplied and consumed} and therefore increase with local income. However, endogenous amenities are also subject to \textit{congestion}, i.e., given their stock, individual utility decreases with the number of people who have access to them. Therefore, we assume that the endogenous amenities are given by:\footnote{See Online Appendix \ref{app:endoAmenities} for a simple micro-foundation of Eq. \eqref{eq:endoAmenities}.}
\begin{equation}
A^N_{EN}(x,t) =  A_0 \left\{ y^N(x,t)^\phi - \mu_A l^N(x,t)\right\} =
A_0 \left\{ \left[G(x)\sum_{i=1}^{N} \frac{1}{N}W_h^P(x-X^{i,N}_t) l^N(x,t)^\beta\right]^{ \phi} - \mu_A l^N(x,t)\right\},
\label{eq:endoAmenities}
\end{equation}
where $A_0$ is a scale parameter, $\phi \in \left(0,1\right)$ is a parameter measuring the elasticity of endogenous amenities to local income, while congestion is assumed to be proportional to the number of workers, with $\mu_A >0$ measuring congestion per worker.
Endogenous amenities are therefore a source of agglomeration until the stock of labour is below $l^N(x,t)^{TR} = \left[\phi \beta \left(A_l^N(x,t) \right)^\phi/\mu_A \right]^{1/\left(1-\phi \beta\right)}$, while after this threshold they act as a source of repulsion.
Endogenous amenities of Eq. \eqref{eq:endoAmenities} at location $x$ thus depend not only on workers at $x$, as in \cite{allen2014trade,allen2020persistence}, but also on all workers within a distance $h$ from $x$. In a more general setting, endogenous amenities could have significant spatial spillovers over a distance different from $h$, as in \cite{redding2024quantitative}.

\subsection{Spatial dynamics of the economy}

Theorem \ref{teo:limitInfiniteAgentsII} recasts Theorem \ref{teo:limitInfiniteAgents} for the specific setting of the economy just presented.\footnote{It is also possible to consider an economy in which the number of workers increases at an exogenous local growth rate $n(x,t)$, which would give rise to an additional term in the right-hand side of Eq. \eqref{eq:dynamicsLabourDistributionII} in the form of $+n(x,t)l(x,t)$. A brief description of this economy can be found in the Online Appendix \ref{app:increasingNumberAgents}.}

\begin{teo}[Spatial Dynamics of the Economy]\label{teo:limitInfiniteAgentsII}
Assume that Eq. \eqref{eq:movementAgent_IIEcon} describes the movement of workers, Eq. \eqref{eq:lN} the stock of labour at each location, Eq. \eqref{eq:laborProductivity} the local technological progress, Eq. \eqref{eq:localIncome} the local income, Eq. \eqref{eq:ProductionWages} the local wage, and Eq. \eqref{eq:endoAmenities} the endogenous amenities in location $x$ at time $t$. Then, as $N$ tends to infinity, the stock of workers in location $x$ at time $t$, $l^{N}(x,t)$ converges in probability to the unique solution of:
\begin{multline}\label{eq:dynamicsLabourDistributionII}
\partial_t l(x,t) = \dfrac{\sigma^2}{2c_M^2} \Delta_x l\left(x,t\right) - \left(\frac{1}{c_M}\right) \, \div_x \left( l(x,t) \nabla_x w\left(x,t\right) \right) - \\
- \left(\frac{1}{c_M}\right) \, \div_x \left( l(x,t) \nabla_{x}A_{EN}(x,t) \right) - \left(\frac{1}{c_M}\right) \, \div_x \left( l(x,t) \nabla_{x}A_{ES}(x,t) \right),
\end{multline}
where 
\begin{multline} 
\label{eq:systematicUtility}
A_{l}(x,t) := G(x)\left(W_{h}^{P}*l\right)(x,t); \quad
y(x,t) := G(x)(W_{h}^{P} * l)(x,t) l(x,t)^{\beta}; \\ 	
w(x,t) := G(x)(W_{h}^{P} * l)(x,t) l(x,t)^{\beta-1};\quad 
A_{EN}(x,t) := y(x,t)^{\phi} - \mu_{A}l(x,t); \text{ and}  \\ 
v(x,t) := w(x,t) + A_{EN}(x,t) + A_{ES}(x,t) 
\end{multline}
are the value of local technological progress, income, wages, endogenous amenities in the limit of infinite $N$, respectively.
\end{teo}
\begin{proof}
For a sketch of the proof see Appendix \ref{app:proofMeanFieldLimit}.
\end{proof}
Some of the key definitions in Eq. \eqref{eq:systematicUtility}, in particular wages and endogenous amenities, can be brought directly back to \cite{allen2014trade} and \cite{allen2020persistence} in the limit of $N$ going to infinity and $h$ going to zero (i.e. no spillovers across neighbouring locations but only in the same location).\footnote{For example, since $\lim_{h \rightarrow 0} \lim_{N \rightarrow \infty}\sum_{i=1}^{N} \left(1/N\right) W_h^P(x-X^{i,N}_t)=l(x,t)$, $y(x,t)$ and $A_{EN}(x,t)$ in Theorem \eqref{eq:dynamicsLabourDistributionII} are ``isomorphisms'' of wages, production and endogenous amenities of \citet{allen2020persistence} under the assumptions that $\beta=\alpha_1$, $\alpha_2=0$, $\phi=\beta_1\left(1+\alpha_1\right)$ and $\mu_{A}=0$.} In the same way, in the limit of $N$ going to infinity, an educated qualification of the bilateral travel time in \cite{redding2024quantitative}, i.e. linear in the spatial distance, can reproduce the shape of spatial spillovers in consumption and production of Eq. \eqref{eq:systematicUtility} under some suitable normalisation. 

Eq. \eqref{eq:dynamicsLabourDistributionII} in Theorem \ref{teo:limitInfiniteAgentsII} is similar to Eq. (10) in \cite{papageorgiou1983agglomeration} for an economy where workers and the number of locations are finite, there are positive spatial externalities, moving costs are zero, and the time scaling of the decision is discrete. Also, Eq. (8.3) in \cite{krugman1996self}, which governs the density of firms in the ``edge city'' model, is similar to our Eq. \eqref{eq:dynamicsLabourDistributionII}, but it lacks a micro-foundation based on workers' choices and spatial interactions. \citet{xepapadeas2016spatial}'s model shows some similarities with Eq. \eqref{eq:dynamicsLabourDistributionII}, although their focus is on capital accumulation.
Our next step is to use Theorem \ref{teo:limitInfiniteAgentsII} to characterise some properties of our economy.

\section{Properties of the economy\label{sec:propertiesEconomy}}
We now proceed to illustrate some key theoretical properties of our economy. In particular, Section \ref{sec:techwagesincome} investigates the \textit{inner structure} of a city with a focus on spatial distribution of technology, wages, income and amenities; Section \ref{sec:TheorylongRunEquilibrium} is devoted to the analysis of the \textit{stationary equilibrium distribution} (SED); Section \ref{sec:TheoryMetastability} deals with the phenomenon of \textit{metastability}; and Section \ref{sec:TheorySocialUtility} introduces the concept of \textit{social utility}. 
Online Appendix \ref{app:representativeWorker} studies the economy from the point of view of the \textit{representative} worker.  

\subsection{Spatial distribution of technology, income, wages and amenities\label{sec:techwagesincome}}

Heuristically, the local number of workers positively affects technology and income, while the relationship with wages and endogenous amenities is more complex for the presence of congestion effects.

With respect to technology $A_l(x,t)$, we have from Theorem \ref{teo:limitInfiniteAgentsII}, assuming $G(x)=1$:
\begin{equation}
\nabla_x A_l(x,t) = (W_{h}^{P} * \nabla_x l)(x,t), 
\end{equation}
from which the gradients of technology ($\nabla_x A_l(x,t)$) and the one of the number of workers ($ \nabla_x l$) point to the same direction, i.e. $A_l(x,t)$ is growing with $l(x,t)$.
The same is true for income, given that
\begin{equation}
\nabla_x y(x,t) = \nabla_x A_l(x,t) l(x,t)^{\beta} + \beta A_l(x,t) l(x,t)^{\beta-1} \nabla_x l(x,t).
\end{equation}
On the contrary, for wage we have
\begin{equation}
\nabla_x w(x,t) =w(x,t) \left[ \dfrac{\nabla_x A_l(x,t) }{A_l(x,t) } + \left(\beta-1\right) \dfrac{\nabla_x l(x,t)}{l(x,t)} \right],
\end{equation}
where the first term inside the square brackets grows with $\nabla_x l$, while the second decreases for $\beta \in \left(0,1\right)$, i.e. when the congestion effect operates in the local labour market. A city may therefore have the highest wages in its suburbs rather than in its centre. Figure \ref{fig:distributionOverSpace} illustrates this point for an economy with the same parameters but a different spatial distribution of workers.
The same congestion effect holds for endogenous amenities, given that
\begin{equation}
\nabla_x A_{EN}(x,t) = \phi y(x,t)^{\phi-1} \nabla_x y(x,t) - \mu_A \nabla_x l(x,t).
\end{equation}
\begin{figure}[!htbp]
\begin{subfigure}[b]{0.45\textwidth} 
\centering
\includegraphics[width=\textwidth]{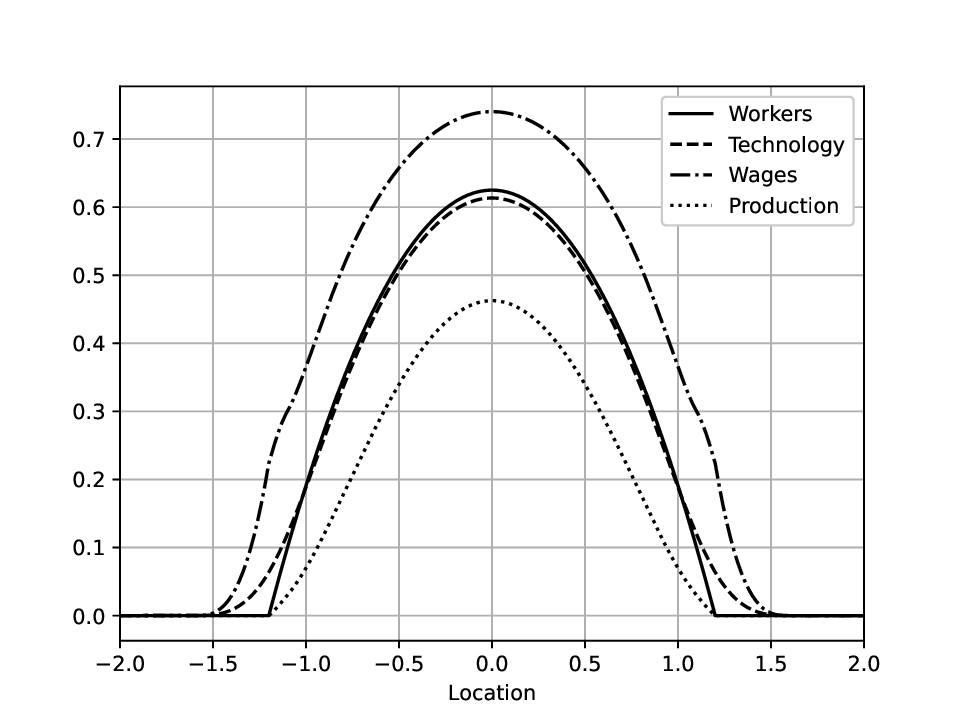}
\caption{Bandwidth of Epanechnikov distribution equal 1.2.}
\label{fig:distributionOverSpace_I}
\end{subfigure}
\hfill
\begin{subfigure}[b]{0.45\textwidth} 
\centering
\includegraphics[width=\textwidth]{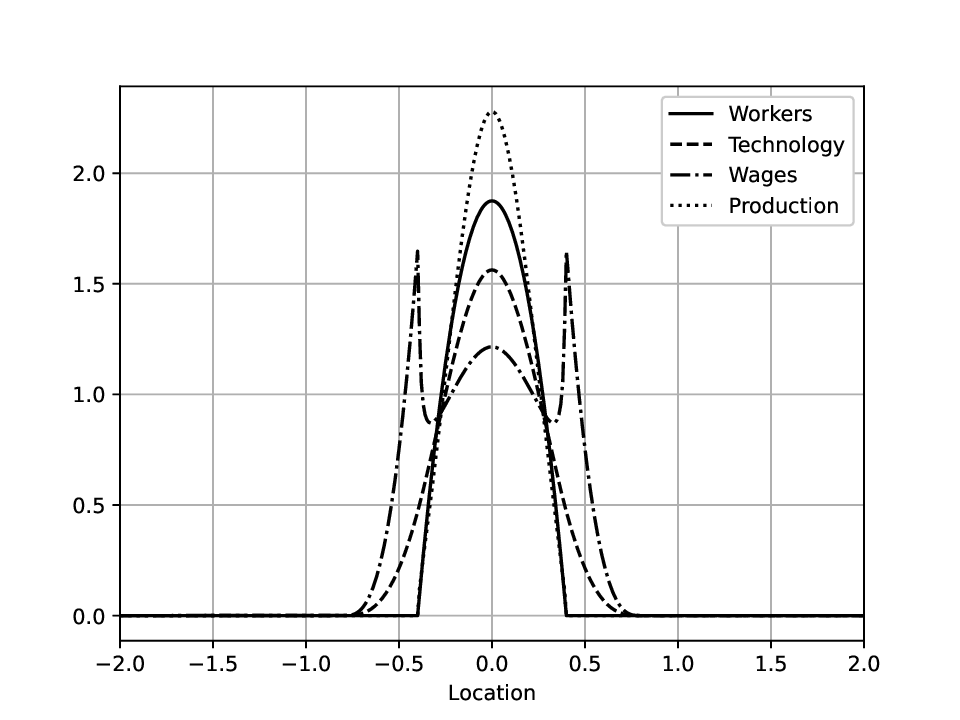}
\caption{Bandwidth of Epanechnikov distribution equals 0.4.}
\label{fig:distributionOverSpace_II}
\end{subfigure}
\caption{The spatial distribution of workers is given by an Epanechnikov distribution with different bandwidths; technology, wages and production are calculated according to Theorem \ref{teo:limitInfiniteAgentsII} with $G(x)=1 \; \forall x$, $\beta=0.6$ and $h=0.4$ (a correction for wages is introduced to avoid singularities around $l(x,t) \leq 0.1$).}
\label{fig:distributionOverSpace}
\end{figure}

\subsection{Stationary equilibrium distribution \label{sec:TheorylongRunEquilibrium}}

Definition \ref{def:stationaryEquilibrium} characterises a stationary equilibrium distribution of workers, i.e. $l(x,t)=l^{EQ}(x)$, which also implies the stability of all other endogenous variables, i.e. technology, income, wages and endogenous amenities.
\begin{defi}[Stationary Equilibrium Distribution]\label{def:stationaryEquilibrium}
A stationary equilibrium distribution (SED) is a function $l^{EQ}(x)$ such that the right-hand side of Eq. \eqref{eq:dynamicsLabourDistributionII} is equal to zero.
\end{defi}

Unfortunately, there not exists a general results of existence and convergence towards a SED for our class of PDE. However, we are able to prove under some nonrestrictive assumptions that: i) the \textit{barycentre} of the initial distribution is invariant in the SED (if it exists) (Theorem \ref{teo:centerofmass}); ii) the SED is \textit{radially symmetric} about the barycentre (Theorem \ref{teo:Equilibrium}); and, iii) the \textit{uniform} distribution is a SED but potentially unstable (Theorem \ref{teo:stabilityUniformDistribution}).  

Theorem \ref{teo:centerofmass} states that in the absence of labour market congestion, exogenous factors and endogenous amenities, the spatial reallocation of workers does not affect the barycentre of the workers' distribution.
\begin{teo}[Barycentre of the Workers' Distribution]\label{teo:centerofmass}
Assume $\beta = 1$, $A_{ES}(x,t)=\bar{A}_{ES}(t)$, $G(x) \equiv \overline{G}$, and $A_{0}=0$. 
Then the barycentre (the centre of mass) of the distribution of workers is invariant over time, i.e. defining the barycentre as
\begin{equation*}
CM_{l}(t) := \int_{\Omega} x l(x,t)\,dx,
\end{equation*}
then
\begin{equation*}
CM_{l}(t) = CM_{l}(0),\quad \forall t \geq  0.
\end{equation*}
\end{teo}
\begin{proof}
See Appendix \ref{app:proofCenterOfMass}.
\end{proof}
A key implication of Theorem \ref{teo:centerofmass} is a radical path dependence in the location of cities in the case of an uniform space \citep{allen2020persistence}. Our result is very similar to that in \citet{krugman1993first}, which states that ``...there are multiple equilibria (indeed a continuum) for metropolitan location.'', but derived in a dynamic framework. In this respect, the exogenous factors (the ``first nature'' in \citet{krugman1993first}'s words), that likely determine the initial distribution, are crucial in determining the barycentre of a SED.

All effects of the location of workers relative to one another, i.e. agglomeration and congestion (the ``second nature” in \citet{krugman1993first}'s words), jointly with diffusion, instead shape the SED as shown by Theorem \ref{teo:Equilibrium}.\footnote{See \citet{redding2024quantitative} for a quantitative evaluation of the importance of ``first nature” versus ``second nature” geography.} 

\begin{teo}[Shape of the Stationary Equilibrium Distribution]\label{teo:Equilibrium}
Assume $\beta = 1$, $A_{ES}(x,t)=\bar{A}_{ES}(t)$, $G(x) \equiv \overline{G}$, and $A_{0}=0$ and let $l^{EQ}(x)$ be a SED. Then $l^{EQ}(x)$ must be a radially symmetric and decreasing function with respect to the barycentre at time zero, i.e. 
\begin{equation}\label{eq:teoEquilibrium}
l^{EQ}(x) = L^{EQ}(\norm{x-CM_{l}(0)})
\end{equation}
for some decreasing function $L^{EQ}:\RR^{+} \to \RR^{+}$.
\end{teo}
\begin{proof}
For the proof that the SED is radially symmetric and decreasing, see \citet[Theorem 2.2]{carrillo2019nonlinear}. To see that the point of radial symmetry is the barycentre at time zero, it is enough to apply Theorem \ref{teo:centerofmass}. 
\end{proof}
Theorem \ref{teo:Equilibrium} thus rules out the possibility that a SED could be multi-peaked, regardless of the initial distribution. Although it does not give a specific shape, the density of workers in the SED must decrease with distance from the barycentre, regardless of direction. The empirical evidence of persistent several spatial agglomerations can be traced to i) the existence of exogenous amenities and/or ii) the diffusion, which drives the results in Theorem \ref{teo:Equilibrium}, that, even if present, is sufficiently weak to generate the phenomenon of metastability discussed in Section \ref{sec:TheoryMetastability}. 

On the other hand, Theorem \ref{teo:stabilityUniformDistribution} states that the uniform distribution is a SED but possibly unstable.\footnote{For a general introduction to stability in PDEs see, e.g., \cite{kapitula2013spectral}.}  
\begin{teo}[Stability of the uniform Stationary Equilibrium Distribution] \label{teo:stabilityUniformDistribution}
Assume $A_{ES}(x,t)=\bar{A}_{ES}(t)$ and $G(x) \equiv \overline{G}$. Then the uniform distribution $l^{EQ}(x)=\bar{l}$ is a SED and wages, technology, income, endogenous amenities, systematic utility and utility are equal across locations. 
Assume, in addition, that $\beta=1$ and $A_{0}=0$. Then there are two critical values $\overline{c_M} > \underline{c_M} > 0$ such that
\begin{enumerate}[i)]
\item if $0 < c_M < \underline{c_M}$ then the SED $l^{EQ}(x)= \bar{l}$ is linearly asymptotically stable;
\item if $c_M > \overline{c_M}$ then the SED $l^{EQ}(x)= \bar{l}$ is linearly unstable.
\end{enumerate}
\end{teo}
\begin{proof}
If $l^{EQ}(x)= \bar{l} \; \forall x \in \Omega $, then all the terms on the right hand side of Eq. \eqref{eq:dynamicsLabourDistributionII} are zero if $A_{ES}(x,t)=\bar{A}_{ES}(t)$ and $G(x) \equiv \overline{G}$.

\noindent If $l^{EQ}(x)=\bar{l}$, then $v^{EQ} = \bar{l}^{\beta} + \left[\bar{l}^{\beta+1}\right]^{\phi} - \mu \bar{l} + \bar{A}_{ES}(t)$, and $u^{EQ} = v^{EQ} - \left(\dfrac{\sigma^2}{2 c_M}\right)\log \bar{l}$, which are independent of $x$. The spatial uniformity of wages, technology, income and endogenous amenities can be deduced in a similar way.

\noindent The proof of Point $i)$ directly derive from \citet[Proposition 2.2]{celinski2014stability}, while of Point $ii)$ from \citet[Theorem 2.4]{celinski2014stability}.
\end{proof}
Theorem \ref{teo:stabilityUniformDistribution} states that the stability of the uniform distribution depends crucially on the relative magnitude of diffusion as opposed to agglomeration forces, as expected. The thresholds $\underline{c_M}$ and $\overline{c_M}$ depend on all the parameters of the model. In particular, an increase in $\sigma$ or $h$, which increases the diffusion or the extent of spatial externalities, leads to a decrease in $\overline{c_M}$ and an increase in $\underline{c_M}$. 
In conclusion, $c_M > \overline{c_M}$ seems to be a necessary condition for the existence of a SED with spatial agglomerations.

\subsection{Metastability\label{sec:TheoryMetastability}}

The class of PDEs reported in Theorem \ref{teo:limitInfiniteAgentsII} admits a phenomenon called \emph{metastability}, which in our economy means that $l(x,t)$ spends an extended period of time around a given spatial distribution before suddenly transitioning to another (apparently stable) spatial distribution \citep{carrillo2019aggregation}. However, to our knowledge, metastability has an elusive definition, and there are no general results in the literature for our case.\footnote{Numerically, metastability has been studied extensively (see, for example, \citealp{Evers2016}).}

As briefly discussed in Section \ref{sec:model}, the metastable behaviour is related to the presence of a noise. To fix the idea, consider the simplified setting where $\beta=1$, $G(x) = \overline{G}$, $A_{ES}(x,t)=\bar{A}_{ES}(t)$, $A_0=0$, and a very limited range of interaction (small $h$); first, consider the case of an economy without randomness, i.e. $\sigma = 0$: for the effect of aggregation over time, all workers tend to concentrate in a finite set of locations.\footnote{Technically, this means that the mass of workers concentrates into multiple \textit{Dirac} masses, i.e. a ``probability density'' that takes a value of infinity in a single location and zero elsewhere, or into a single Dirac mass if the initial distribution of workers is within the \textit{perception} range of $W_{h}$.} Then consider a $\sigma>0$, but very small: \citet{cozzi2017aggregation} show that at finite time interval $[0,T]$ the aggregate behaviour closely approximates that with $\sigma = 0$, i.e. a multi-peaked distribution. However, Theorem \ref{teo:Equilibrium} states that the SED is mono-peaked and, therefore, the distribution in period $[0,T]$ appears to be into the set of metastable equilibria.

A drawback of the above setting is that $T$ must be chosen sufficiently small to ensure that workers are not all concentrated in a finite set of locations as a result of aggregation. In our setting, this can be overcome by taking $\beta<1$, which introduces congestion. Given an initial distribution with sufficiently dispersed workers, the intuitive result is a multi-peaked distribution for \textit{any} $T$. Unfortunately, for this setting a rigorous result like the one by \citet{cozzi2017aggregation} is not available, which suggests exploring numerically this issue (we refer to Section \ref{sec:numericalExplorations} for a more sophisticated numerical example).\footnote{In this regard, the most promising approach appears the classical $\Gamma$-convergence techniques used, e.g., by \citet{ambrosio2005gradient} and \citet{lagoutiere2023vanishing}.} 
Figure \ref{fig:L2NormDifference} illustrates the difference between the workers' spatial distribution and their initial (multi-peaked) distribution, which is the SED for $\sigma = 0$, measured by the $L^{2}$-norm, i.e. $\norm{l_{\sigma=0}(\infty,\cdot)-l_{\sigma}(T,\cdot)}_{2}$, calculated at periods $T=0.5,1.0$ and $2.0$ for different values of $\sigma$.
\begin{figure}[!htpb]
\includegraphics[width=0.45\textwidth]{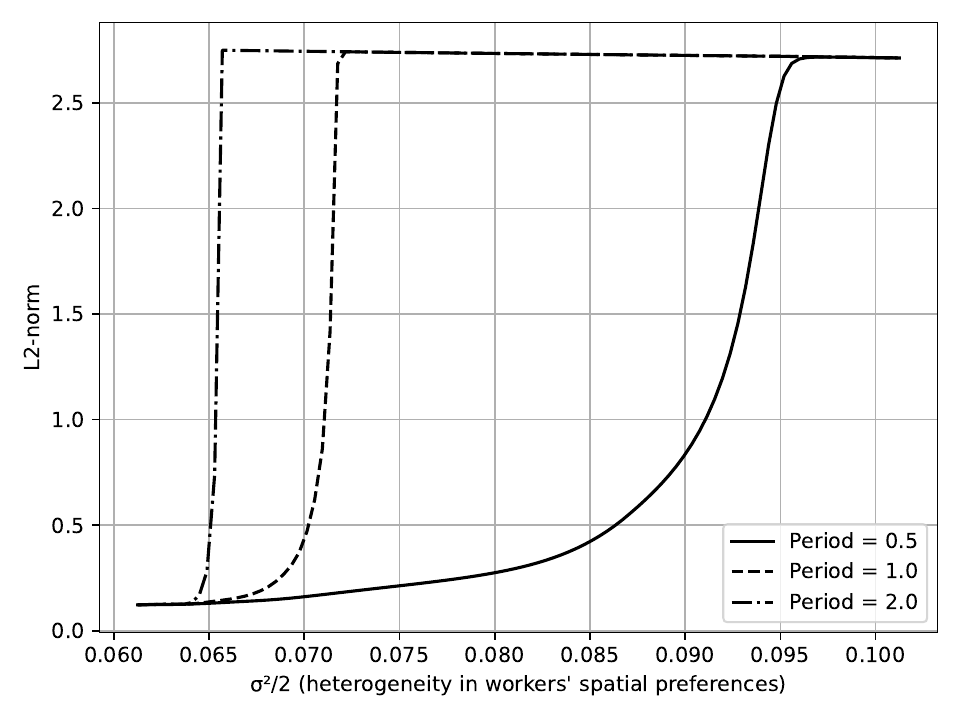}
\caption{The difference between the spatial distribution of workers and their initial (multi-peaked) distribution measured by the $L^{2}$-norm $\norm{l_{\sigma=0}(\infty,\cdot)-l_{\sigma}(T,\cdot)}_{2}$ as a function of the heterogeneity in workers' spatial preferences $\sigma$ calculated at period $T=0. 5,1.0$ and $2.0$ with $\beta=0.6$, $G(x) = 1$, $A_{ES}(x,t)=1$, $A_0=0$ and $h=0.15$.} 
\label{fig:L2NormDifference}
\end{figure}
We observe that for low values of $\sigma^2/2$ (i.e. below 0.065), the distribution remains very close to the initial one even for $T=2$; instead, in the range $(0.065,0.070)$, a sharp transition occurred before $T=2$; the timing of this transition, finally, seems to be inversely related to $\sigma$.\footnote{In statistical physics, this behaviour is called a \textit{first-order phase transition} \citep{bovier2016metastability}.}

From a more general perspective, metastability can be seen as a type of \textit{persistence} because: i) an economy can remain in a (metastable) distribution for a long time before it starts to converge to its SED; and ii) if an exogenous shock drives an economy out of its SED, the convergence path may take a long time if the economy is trapped in a metastable distribution. \citet{allen2020persistence} provides an insightful theoretical discussion and sound empirical evidence on the ubiquity of persistence in the spatial distribution of economic activity.

\subsection{Social utility\label{sec:TheorySocialUtility}}
A natural question that arises in our setting is whether the uncoordinated choices of workers, who maximise their utility locally, have some welfare properties at the aggregate level. Eq. \eqref{eq:teoTotalUtility} suggests as a measure of social utility:
\begin{multline}
SU(l(t)) := \int_{\Omega} l(x,t) u(x,t) \, dx = \int_{\Omega} l(x,t) v(x,t) \, dx - \dfrac{\sigma^2}{2c_M}\int_{\Omega} l(x,t) \log l(x,t) \,dx.
\label{def:socialUtility}
\end{multline}
Eq. (\ref{def:socialUtility}) defines social utility $SU$ as the sum of two components: the first component accounts for the aggregate/average systematic utility of workers, while the second component, directly related to diffusion forces, reaches its maximum for a uniform spatial distribution of workers and is proportional to the Theil index of spatial density inequality \citep{theil1967information}.\footnote{In the mathematical literature, social utility \eqref{def:socialUtility} is commonly referred to as ``free energy'', the Theil index of spatial inequality as ``entropy'', and aggregate/average systematic utility as ``interaction energy'' \citealp{carrillo2019aggregation}.} The cost of spatial density inequality is weighted by the ratio $\sigma^2/c_M$, reflecting that higher $\sigma$ implies, \textit{ceteris paribus}, a preference of workers to be more dispersed in space, while higher movement costs $c_M$ lead workers to prefer a more concentrated distribution. 

Theorem \ref{teo:SU} characterises the dynamics of $SU$ in Eq. \eqref{def:socialUtility}.\footnote{The \textit{Wasserstein distance} with index 2 on $\Omega$ is defined on the set of spatial probability distributions on $\Omega$, i.e. non-negative functions $f:\Omega \to \RR^{+}$ that integrate to one, and is defined as
\begin{equation*}
W_2(f, g) = \left( \inf_{T \, : \, T_\# f = g} \int_\Omega \norm{x-T(x)}^{2}f(x) \, dx \right)^{1/2},
\end{equation*}
where the infimum is taken over all functions, called \emph{transport maps} $T : \Omega \to \Omega$, that transform the distribution $f(x)$ into $g(x)$. The definition of \textit{Wasserstein Space} and Wasserstein distance is more general \citep{ambrosio2005gradient}. Distance $W_2(f, g)$ is closely related to the quadratic movement costs, see Online Appendix \ref{app:micro-foundationAgentsMovements}.}
\begin{teo}[Dynamics of Social Utility]\label{teo:SU}
Suppose $\beta=1$, $G(x) = \overline{G}$ and $A_0=0$. 
Then Eq. \eqref{eq:dynamicsLabourDistribution} can be written as:
\begin{equation}
\partial_t l(t) = \left(\dfrac{1}{c_M} \right) \nabla_l SU(l(t)),
\label{eq:gradientFlow}
\end{equation}
where $\nabla_l SU(l(t))$ denotes the gradient with respect to the Wasserstein 2 metric, and therefore 
\begin{equation}
\frac{d}{dt} SU(l(t)) \geq 0 \; \forall t.
\end{equation}
\end{teo}
\begin{proof}
See \cite{carrillo2003kinetic,carrillo2006contractions}
\end{proof}
Eq. \eqref{eq:gradientFlow} in Theorem \ref{teo:SU} states that the dynamics of our economy is ``locally'' maximising $SU$ over time, i.e. the sequence of local optimisations of workers guarantees a kind of (local) social efficiency in the aggregate dynamics. At the same time, it is not guaranteed that the utility of each worker does not decrease over time.
To understand where $SU$ converges over time, the dynamics of the spatial distribution can be expressed as a sequence of maximisation problems with a small time step of length $\tau$, i.e:
\begin{equation}
l(k+\tau) = \arg \max_{l} \left\{SU(l) - \frac{c_{M}}{2 \tau} W_2^2(l, l(k)) \right\},
\label{eq:JKO}
\end{equation}
which is a time discretisation of the gradient equation according to the Jordan–Kinderlehrer–Otto (JKO) scheme in \citet{jordan1998variational}.
Therefore, any SED to which an economy converges represents a local maximum for $SU$; unfortunately, we do not have a result that the SED is unique, and therefore convergence to a global maximum of $SU$ is not guaranteed. In any case, assuming that the economy has reached the global maximum of $SU$, the corresponding SED would also be the spatial distribution preferred by a social planner who could allocate workers at no cost. This finding contrasts with the inefficiency of the competitive spatial equilibrium established by \cite{fajgelbaum2020optimal}, but the latter result relies crucially on the heterogeneity of exogenous local productivity and amenities. In this setting, we conjecture that the same inefficiency would be present in our framework.
Theorem \ref{teo:SU} also suggests another perspective on metastability based on the shape of $SU$: if $SU$ is almost ``flat'' in a given range of the spatial distribution, its gradient would be small and therefore the time change in the spatial distribution will be very ``slow'' (Eq. \eqref{eq:gradientFlow}).

Finally, we conjecture that Theorem \ref{teo:SU} also holds under the more general conditions of Theorem \ref{teo:limitInfiniteAgentsII}. Setting the parameters' values as reported in Table \ref{tab:parameterValues} below (which represent the baseline setting of our numerical explorations of the model in Section \ref{sec:numericalExplorations}), Figure \ref{fig:socialUtility} confirms that $SU$ is always increasing over time also for an economy with $\beta \in \left(0,1\right)$ and endogenous amenities. We note also that aggregate income is not a good proxy for social welfare, in particular when high diffusion (and/or low movement costs): while for $\sigma=0.05$ (Figure \ref{fig:socialUtility_I}) $SU$ and aggregate income have the same increasing trend over time, for $\sigma=0.2$ (Figure \ref{fig:socialUtility_II}) aggregate income is decreasing.
\begin{figure}[!htbp].
\begin{subfigure}[b]{0.45\textwidth} 
\centering
\includegraphics[width=\textwidth]{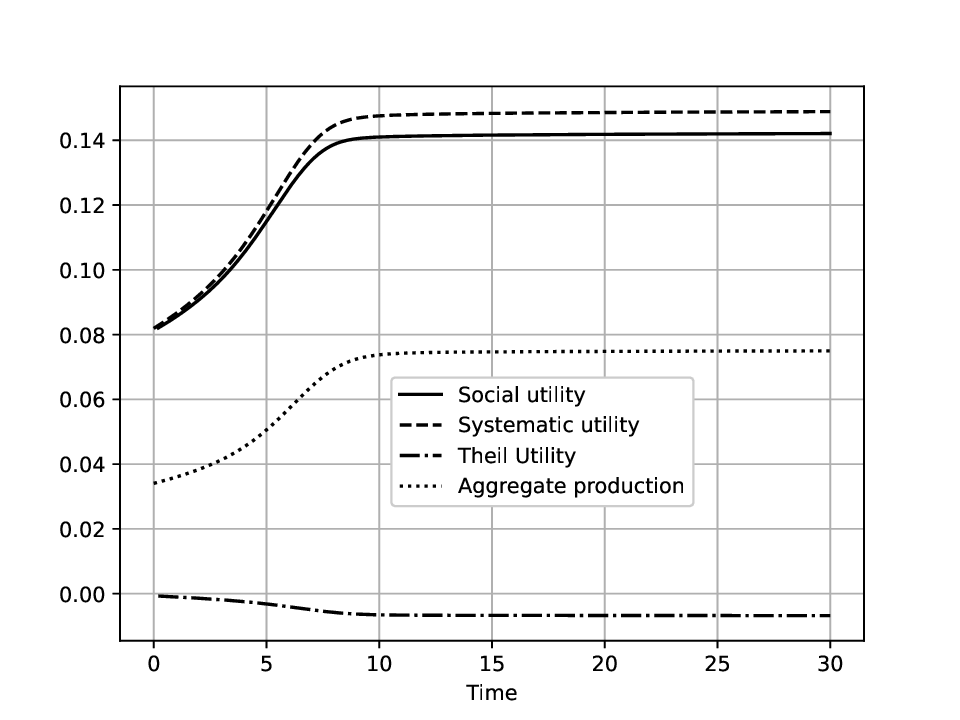}
\caption{Heterogeneity in workers' spatial preferences $\sigma=0.05$.}
\label{fig:socialUtility_I}
\end{subfigure}\hfill
\begin{subfigure}[b]{0.45\textwidth}
\centering
\includegraphics[width=\textwidth]{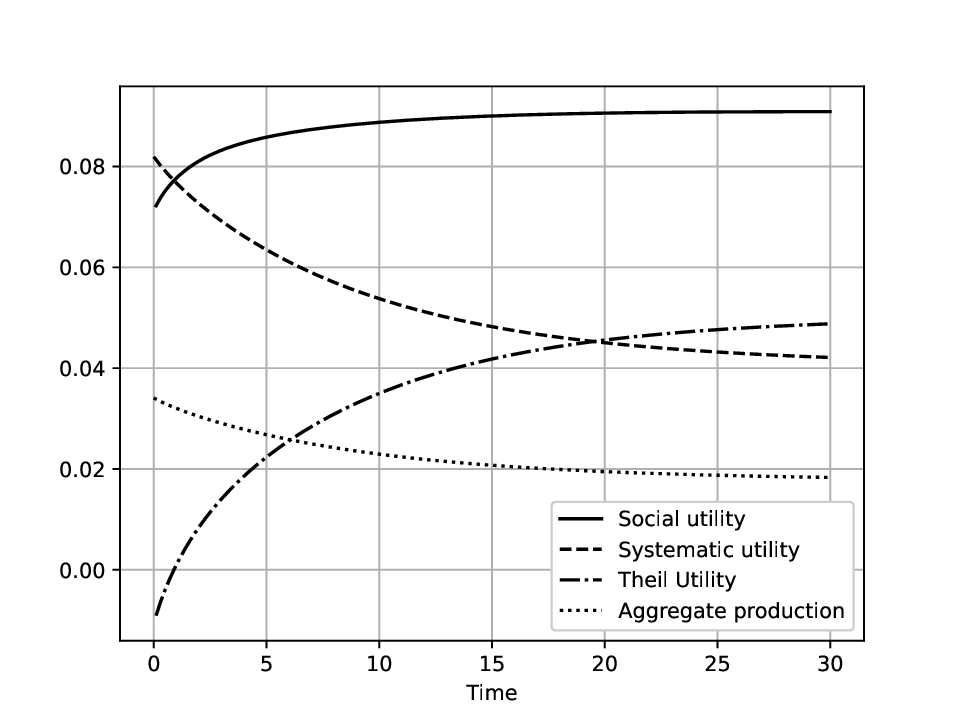}
\caption{Heterogeneity in workers' spatial preferences $\sigma=0.2$.}
\label{fig:socialUtility_II}
\end{subfigure}
\caption{The dynamics of social utility $SU$, aggregate systematic utility, Theil index and aggregate production for the setting reported in Table \ref{tab:parameterValues}, excluding heterogeneity in workers' spatial preferences $\sigma=0.05$ (left) and $\sigma=0.2$ (right) in period $[0,30]$.}
\label{fig:socialUtility}
\end{figure}

\section{Numerical explorations\label{sec:numericalExplorations}}
In this section we use numerical simulations to investigate the properties of our economy from Theorem \ref{teo:limitInfiniteAgentsII}, using as a baseline the setting reported in Table \ref{tab:parameterValues}. 
In particular, Section \ref{sec:oneMegaCity} analyses the spatial dynamics leading to a megacity; Section \ref{sec:viaEmilia} examines the emergence of cities of different size and shape, as well as the importance of history (initial conditions and path dependence); and finally, Section \ref{sec:metastability} examines the phenomenon of persistence and metastability.

\begin{table}[!htbp]
\begin{footnotesize}
\begin{tabular}{ccl}
\hline \hline
\\
\hline
Parameter & Value & Description \\	
\hline	
$c_M$ & 100 & The cost of movement\\
$\sigma$ & 0.05 & The standard deviation of the random component of worker utility\\
& & (the heterogeneity in workers' spatial preferences) \\
$\beta$ & 0.6 & The labour coefficient of the local production function\\
$h$ & 0.4 & The extent of spatial spillovers in local technology\\
$A_0$ & 2.68 & The scale parameter for endogenous amenities\\
$\phi$ & 0.5 & The elasticity of endogenous amenities to local income\\
$\mu_{A}$ & 0.45 & The parameter measuring the congestion of local amenities\\
\\
\hline
Distributions & Value & Description \\	
\hline
$W_1^P\left(z\right)$ & $\left(1- \norm{z}\right)\mathds{1}_{\norm{z} \leq 1}$ & The kernel function for the local technology (i.e. the tent kernel function)\\
$G(x)$ & 1 & The production potentials (i.e. uniform over space)\\
$A_{ES}(x,t)$ & 1 & The exogenous amenities (i.e. uniform over space)\\
\\
\hline \hline
\end{tabular}
\caption{The baseline setting used in the numerical simulations of Eq. \eqref{eq:dynamicsLabourDistributionII}.}
\label{tab:parameterValues}
\end{footnotesize}
\end{table}	

\subsection{One megacity \label{sec:oneMegaCity}}

Figure \ref{fig:numericalInvestigationh=0.4} shows the dynamics of the spatial distribution of workers, starting from an initial uniform distribution over the space $\Omega=[0,4]\times[0,4]$. In period 1, four agglomerations (cities) emerge; this pattern becomes more pronounced in period 5, with four distinct cities; however, in period 10, they tend to merge and, finally, in period 20, they merge into one large city, which also remains in the next periods (not shown in the figure). In a large part of the plane in period 2, workers are not present, including regions that experimented an increase in the local workforce in the first period, providing an example of \textit{non-linear out-of-equilibrium} dynamics.
\begin{figure}[!htbp]
\centering
\begin{subfigure}[b]{0.19\textwidth}
\centering
\includegraphics[width=\textwidth]{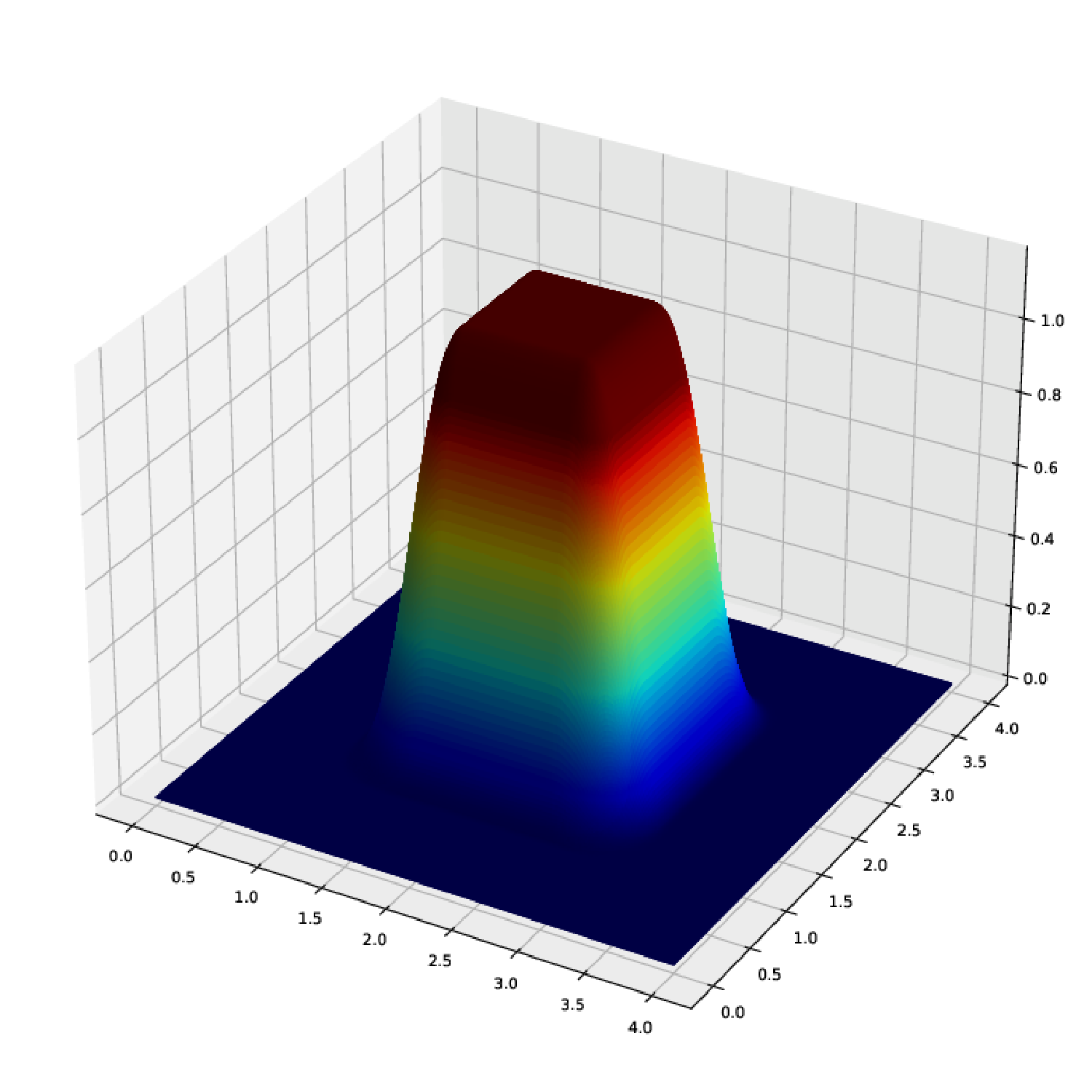}
\caption{$t=0$}
\end{subfigure}
\begin{subfigure}[b]{0.19\textwidth}
\centering
\includegraphics[width=\textwidth]{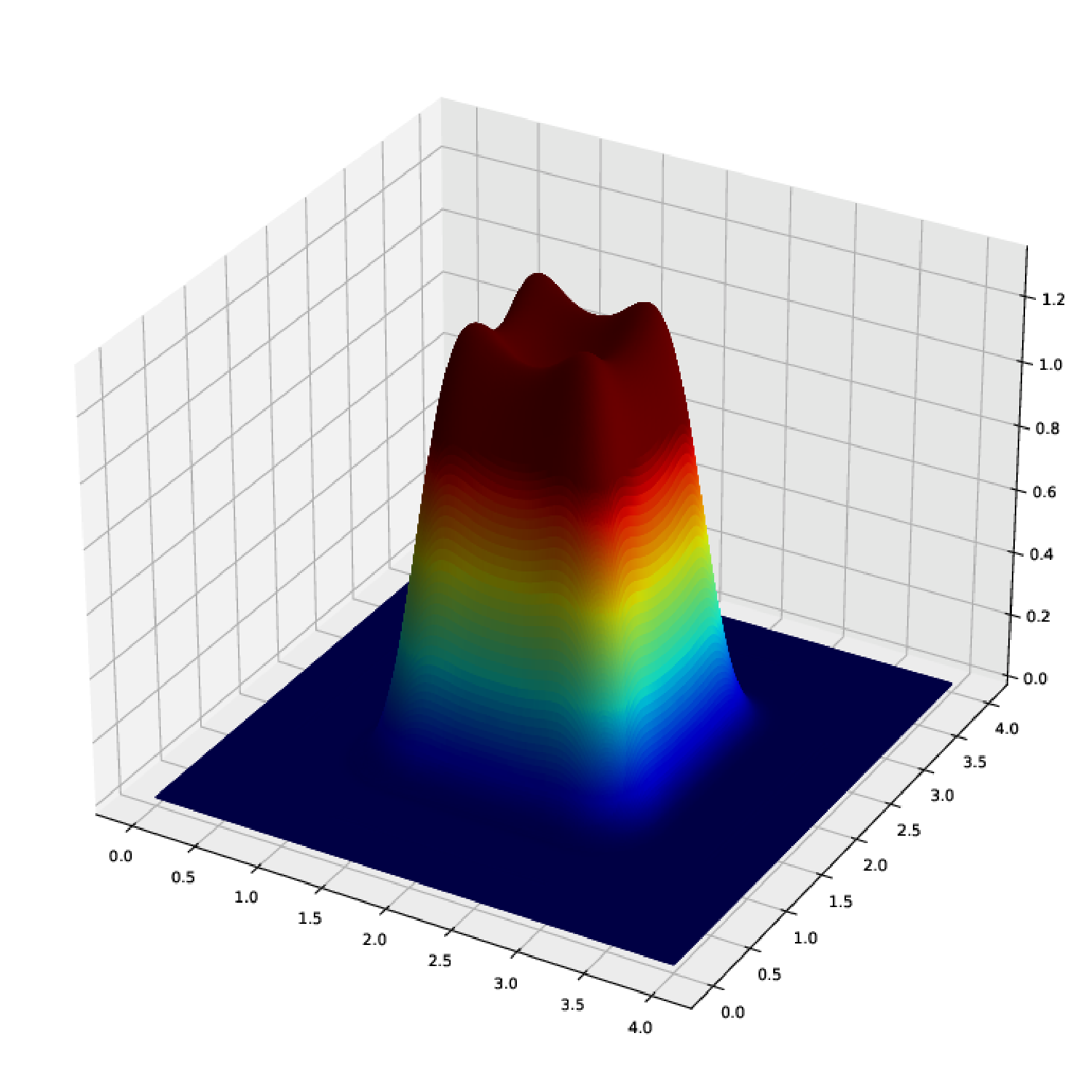}
\caption{$t=1$}
\end{subfigure}
\begin{subfigure}[b]{0.19\textwidth}
\centering
\includegraphics[width=\textwidth]{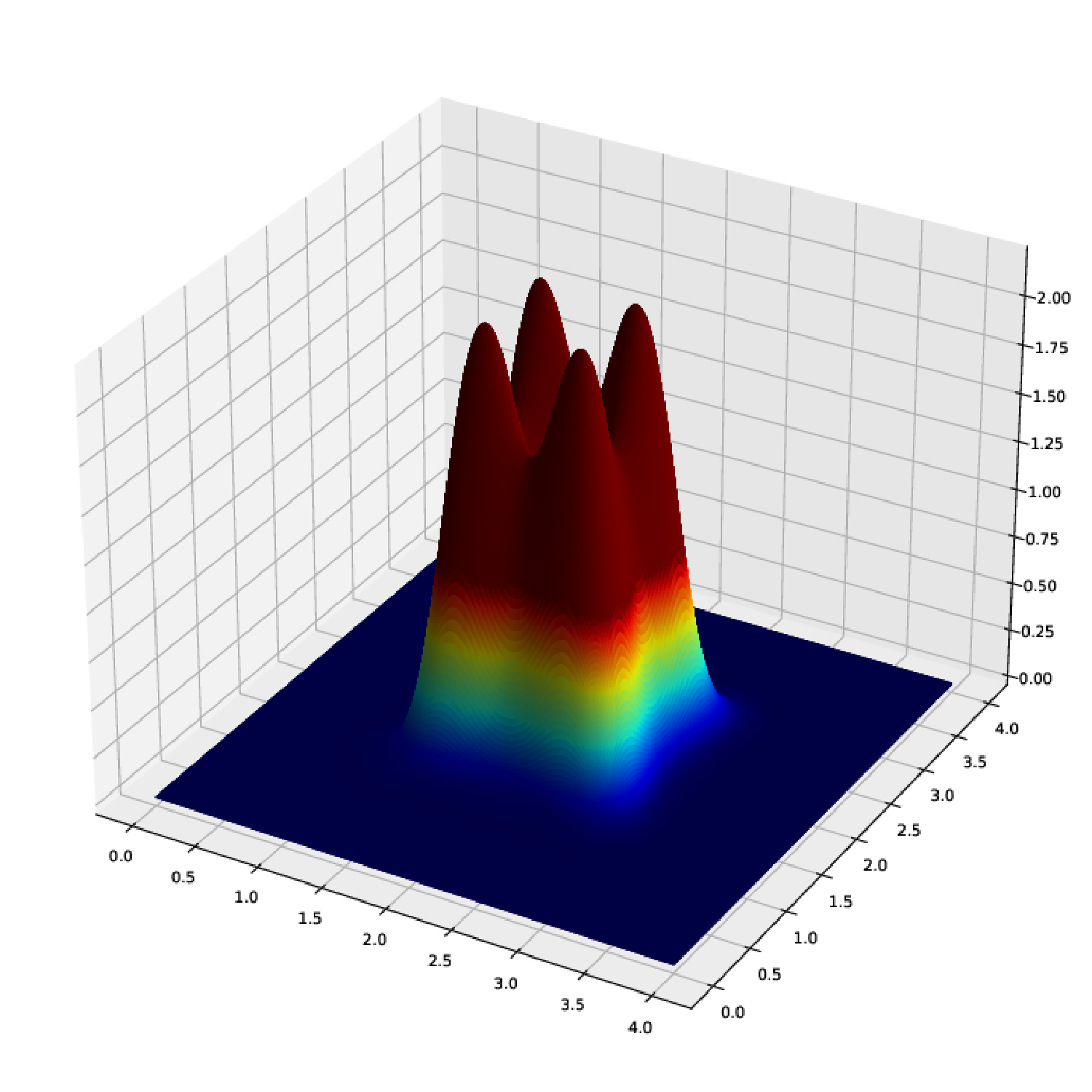}
\caption{$t=5$}
\end{subfigure}
\begin{subfigure}[b]{0.19\textwidth}
\centering
\includegraphics[width=\textwidth]{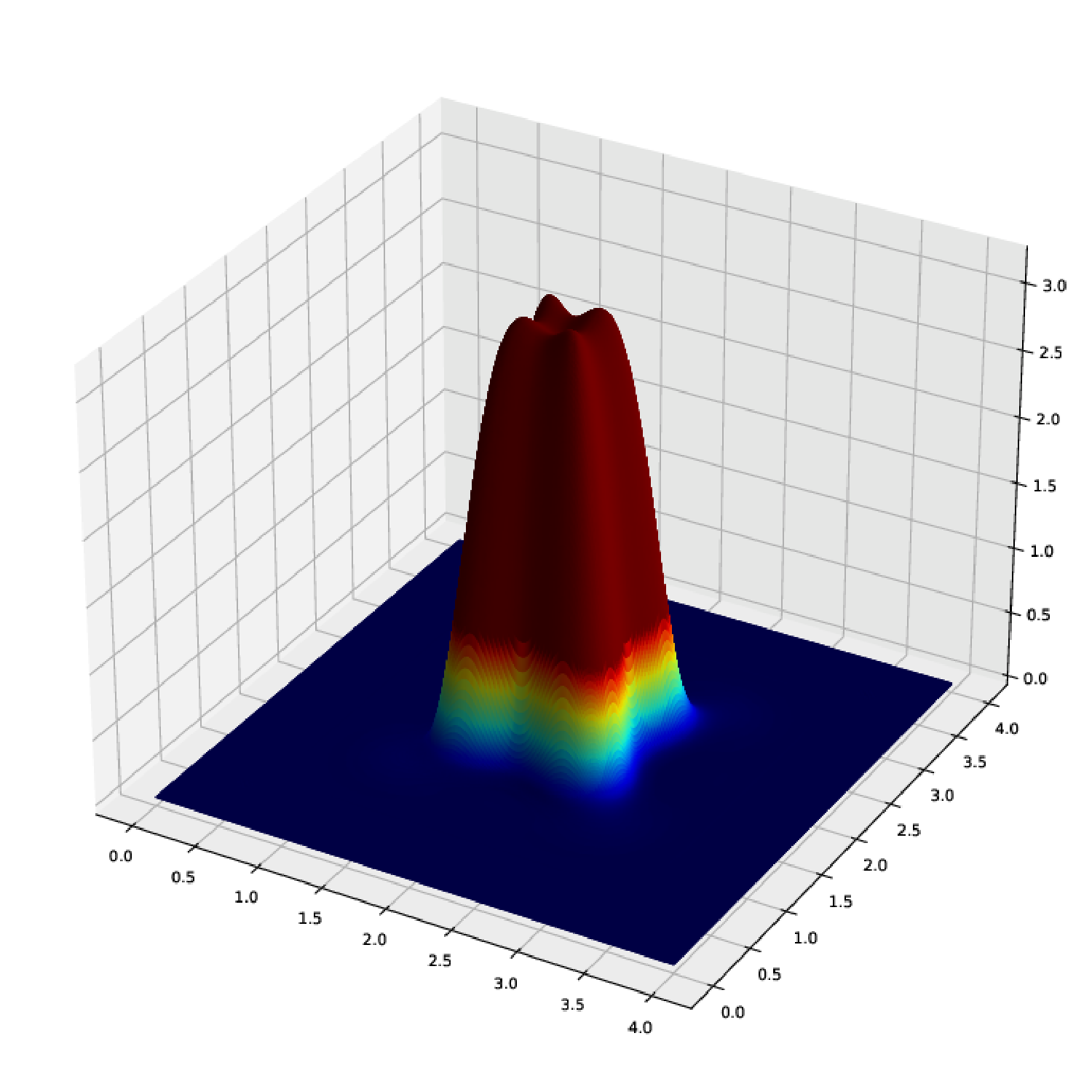}
\caption{$t=10$}
\end{subfigure}
\begin{subfigure}[b]{0.19\textwidth}
\centering
\includegraphics[width=\textwidth]{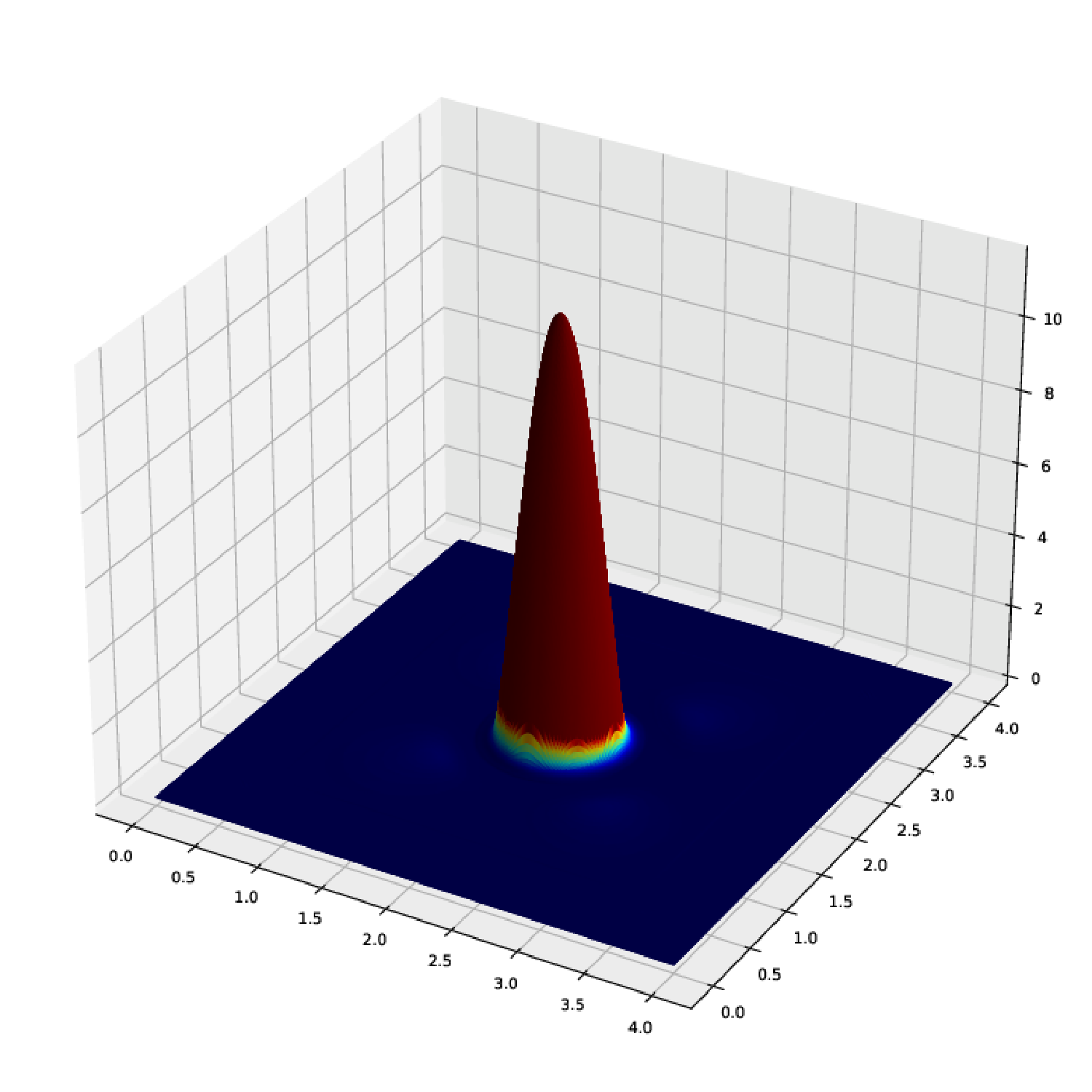}
\caption{$t=20$}
\end{subfigure}
\caption{The distribution of $l(x,t)$ (workers) over the space $\Omega=[0,4]\times[0,4]$ for the setting given in Table \ref{tab:parameterValues} at different times $t$ from 0 to 20.}
\label{fig:numericalInvestigationh=0.4}
\end{figure}

Figure \ref{fig:megacityAllVariables} shows the spatial distribution of the variables in Theorem \ref{teo:limitInfiniteAgentsII} at period 20. In this period, the economy seems to have reached a SED, as the systematic utility of workers is almost flat over space (at least for the populated locations).\footnote{The random part of utility, positively related to $\sigma$, makes the spatial distribution of systematic utility not perfectly flat in the SED.} The spatial distributions of workers, total income and technological progress are strongly correlated, while wages and endogenous amenities are less correlated as a joint effect of congestion and agglomeration. The spatial gradient of wages with respect to the city centre appears to be strongly non-linear, with a centre of medium wages, an intermediate ring of low wages and a peripheral ring of high wages (as in Figure \ref{fig:distributionOverSpace_II}). An increase in diffusion can radically alter the overall spatial characteristics of the SED. The Online Appendix \ref{app:megacityHighDiffusion} reports the case with $\sigma=0.2$; in the SED, the single-peaked distribution still persists, but with higher spatial dispersion, while wages instead reflect the density of workers (as in Figure \ref{fig:distributionOverSpace_I}).\footnote{See e.g. \citet{timothy2001intra} for some evidence on intra-urban wage variation.}
\begin{figure}[!htbp]
\centering
\includegraphics[width=0.7\textwidth]{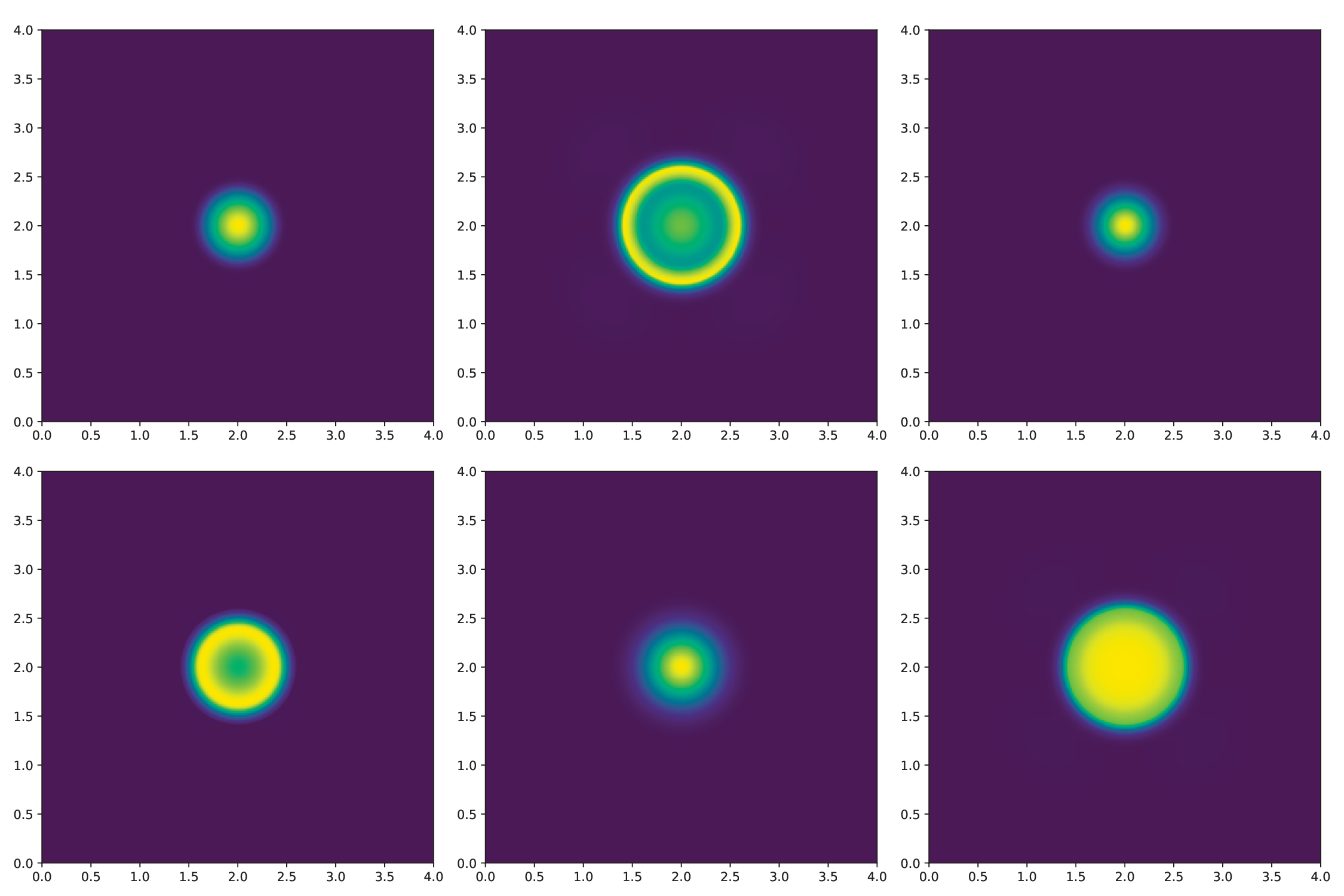}
\caption{The distribution of $l(x,t)$ (workers, top left), $w(x,t)$ (wages, top centre), $y(x,t)$ (income, top right), $A_{EN}(x,t)$ (endogenous amenities, bottom left), $A_l(x, t)$ (technological progress, bottom centre), $v(x,t)$ (individual systematic utility, bottom right) over the space $\Omega=[0,4]\times[0,4]$ at $t=20$ for the setting given in Table \ref{tab:parameterValues}. }
\label{fig:megacityAllVariables}
\end{figure}

\subsection{Stable city size distribution} \label{sec:viaEmilia}

Figure \ref{fig:viaEmilia} examines the effects of workers initially being located in a particular strip, specifically in the area $[1,2]\times[1,7]$. In the first period, workers tend to agglomerate at the two extremes of the strip; however, after $t=100$, other spatial agglomerations appear and finally a stable spatial distribution with six different cities of different sizes emerges. This is an extreme case of dependence on initial conditions and shows how a complex city size distribution can be generated even in the absence of any spatial heterogeneity. \cite{marsili1998interacting,gabaix1999zipf} reach a similar result but in a different theoretical framework, in particular providing a theoretical explanation of Zipf's law. A similar phenomenon can be found in \cite{allen2020persistence}.
\begin{figure}[!htbp]
\centering
\begin{subfigure}[b]{0.24\textwidth}
\centering
\includegraphics[width=\textwidth]{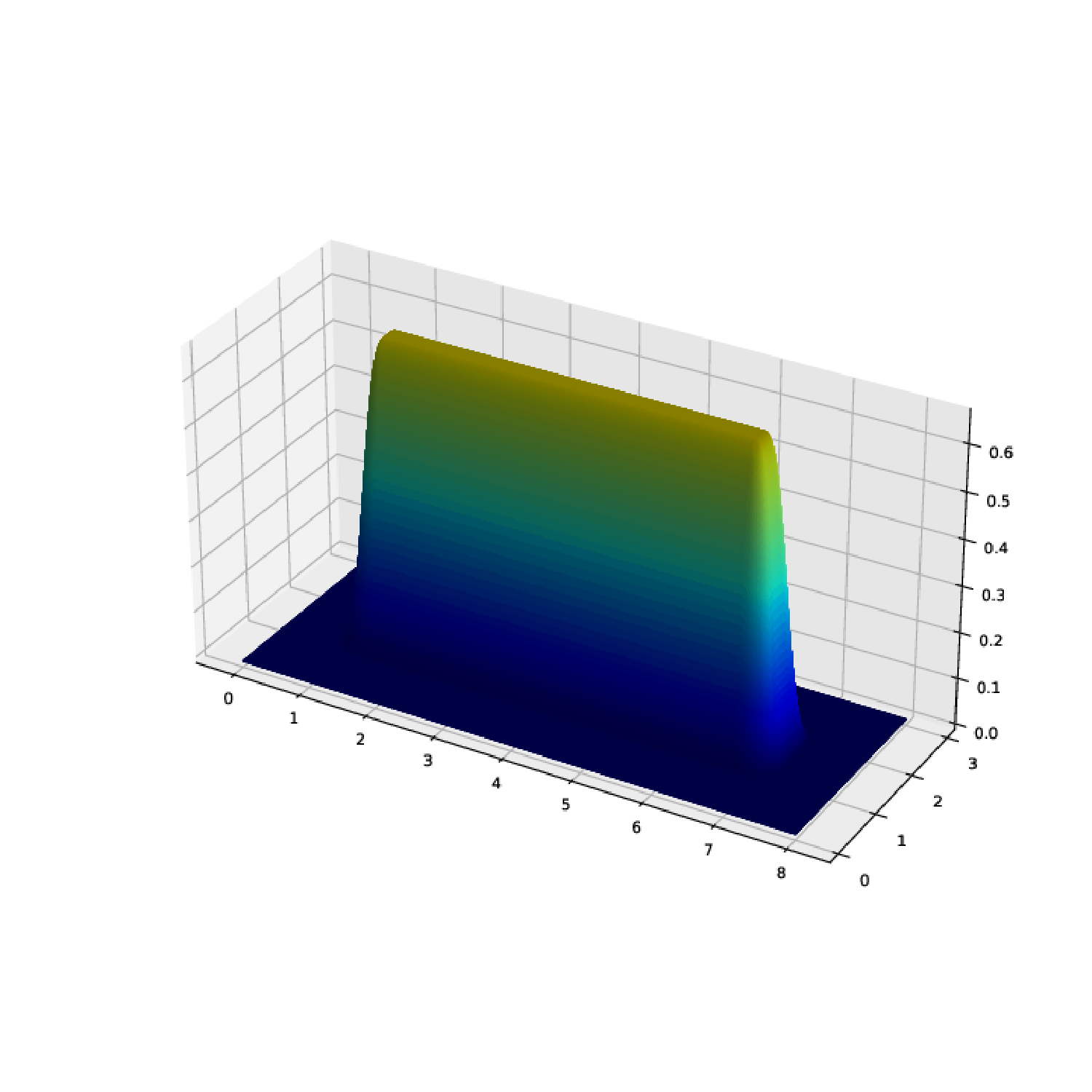}
\caption{$t=0$}
\end{subfigure}
\begin{subfigure}[b]{0.24\textwidth}
\centering
\includegraphics[width=\textwidth]{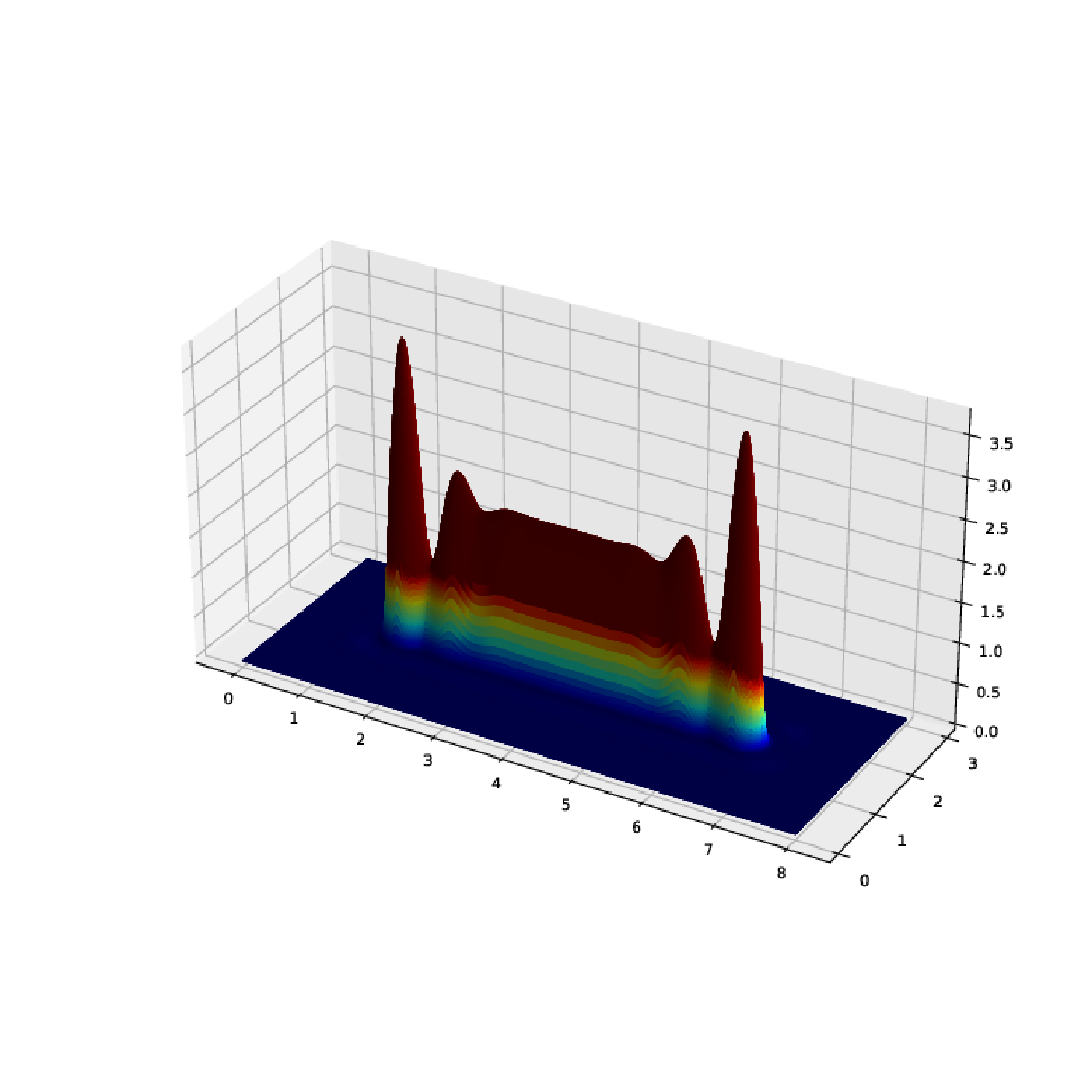}
\caption{$t=10$}
\end{subfigure}
\begin{subfigure}[b]{0.24\textwidth}
\centering
\includegraphics[width=\textwidth]{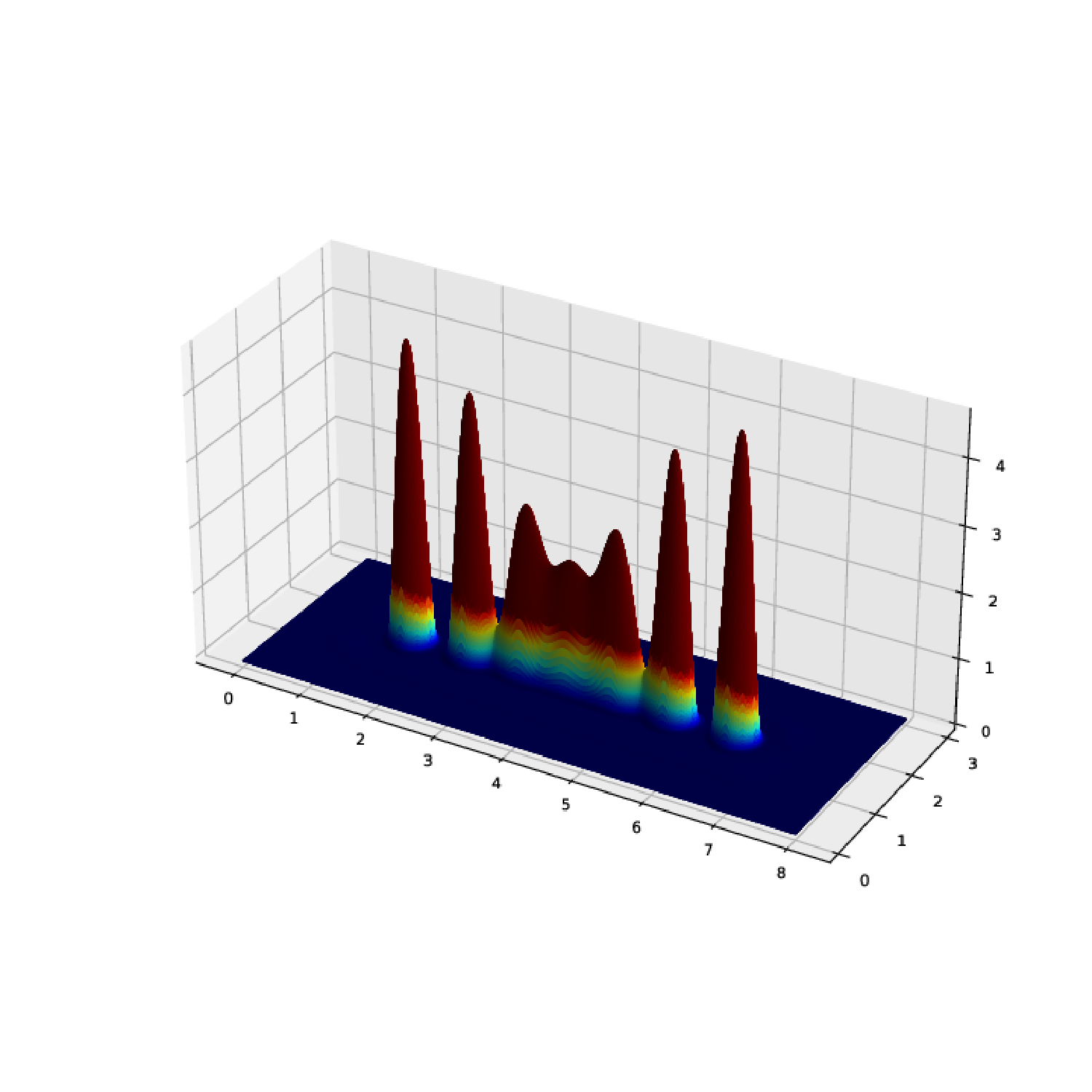}
\caption{$t=20$}
\end{subfigure}
\begin{subfigure}[b]{0.24\textwidth}
\centering
\includegraphics[width=\textwidth]{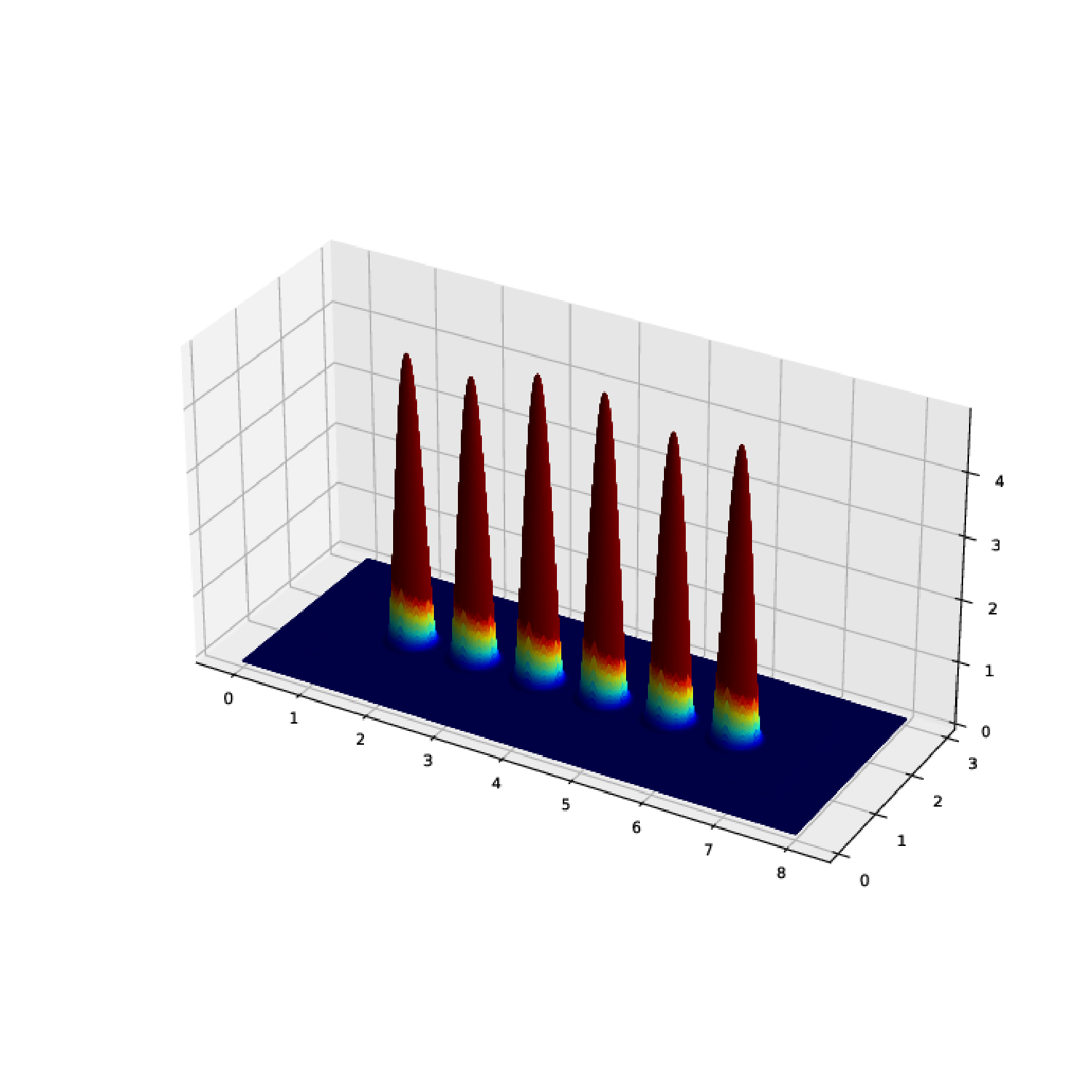}
\caption{$t=100$}
\end{subfigure}
\caption{The distribution of $l(x,t)$ (workers) over space $\Omega=[0,3]\times[0,8]$ at $t$ equal to 0, 10, 20 and 100 for the setting reported in Table \ref{tab:parameterValues}, but workers being initially located in the area $[1,2]\times[1,7]$.}
\label{fig:viaEmilia}
\end{figure}

\subsection{Persistence in the spatial distribution}\label{sec:metastability}

Figure \ref{fig:numericalInvestigationh=0.3} shows the dynamics of the worker distribution when the extent of spatial spillovers in technology $h$ is reduced from 0.4 to 0.3.
While the SED reached in period 185 is always a megacity, there is a wide interval from $t=10$ to $t=175$ where the spatial pattern appears stable with four different medium-sized cities. However, in just two periods, from 180 to 182, this pattern changed radically with the emergence of the SED megacity. This dynamic well illustrates the phenomenon of metastability \citep{Evers2016} and raises serious doubts about the possibility of using standard econometric techniques to study the dynamics of city size distribution. Metastability is directly related to the phenomenon of persistence of urban agglomerations, as discussed in \cite{allen2020persistence}.
\begin{figure}[!htbp]
\centering
\begin{subfigure}[b]{0.20\textwidth}
\centering
\includegraphics[width=\textwidth]{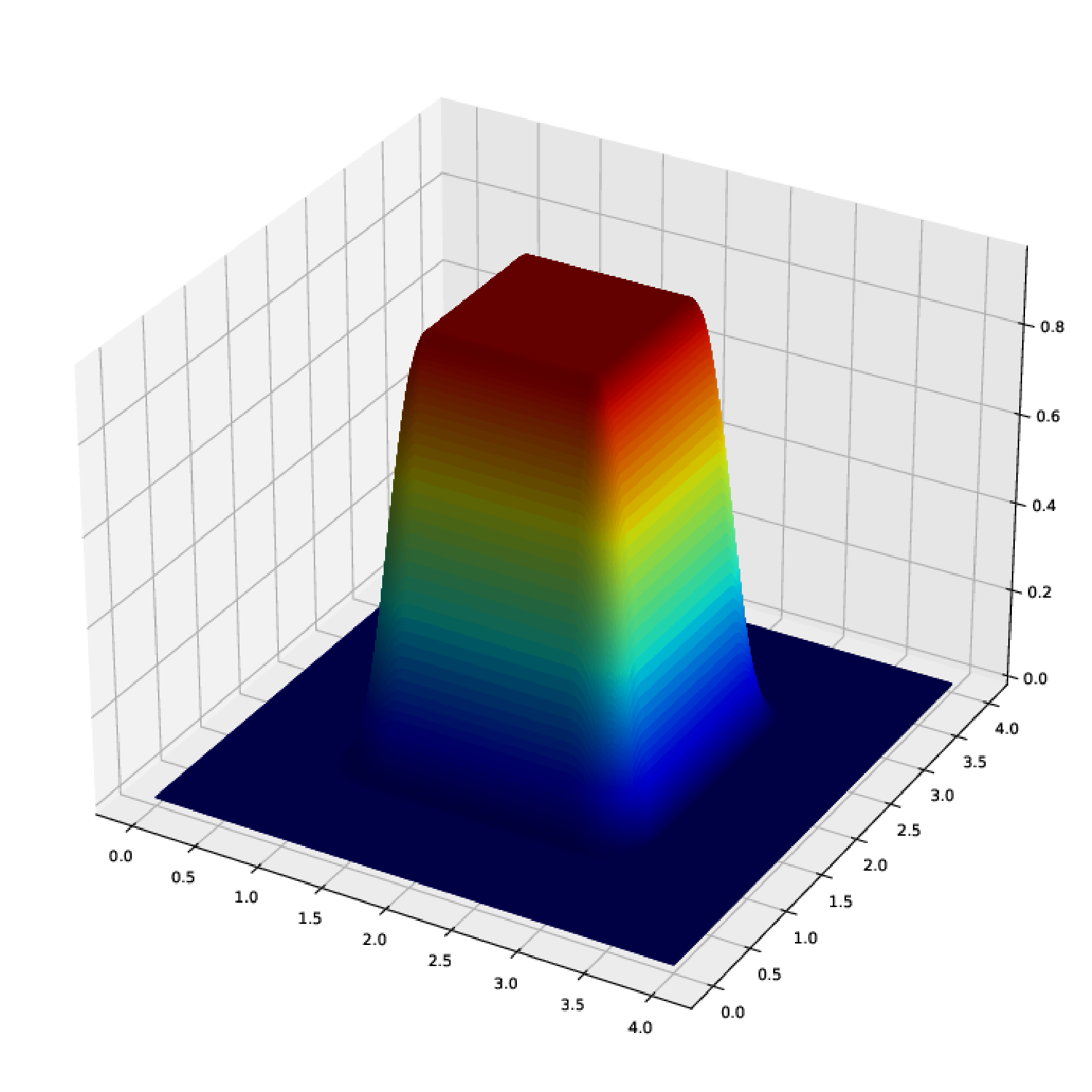}
\caption{$t=0$}
\end{subfigure}
\begin{subfigure}[b]{0.20\textwidth}
\centering
\includegraphics[width=\textwidth]{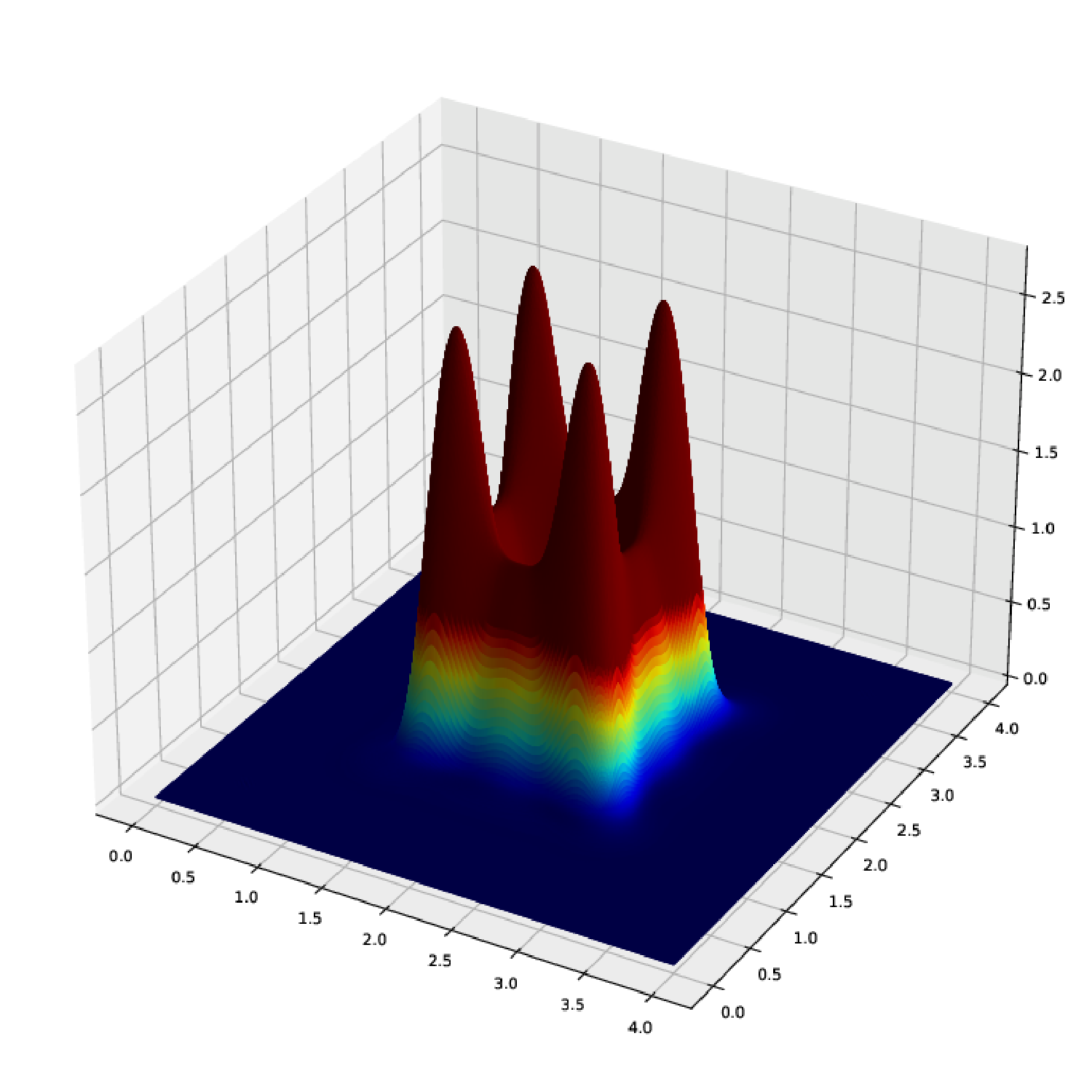}
\caption{$t=5$}
\end{subfigure}
\begin{subfigure}[b]{0.20\textwidth}
\centering
\includegraphics[width=\textwidth]{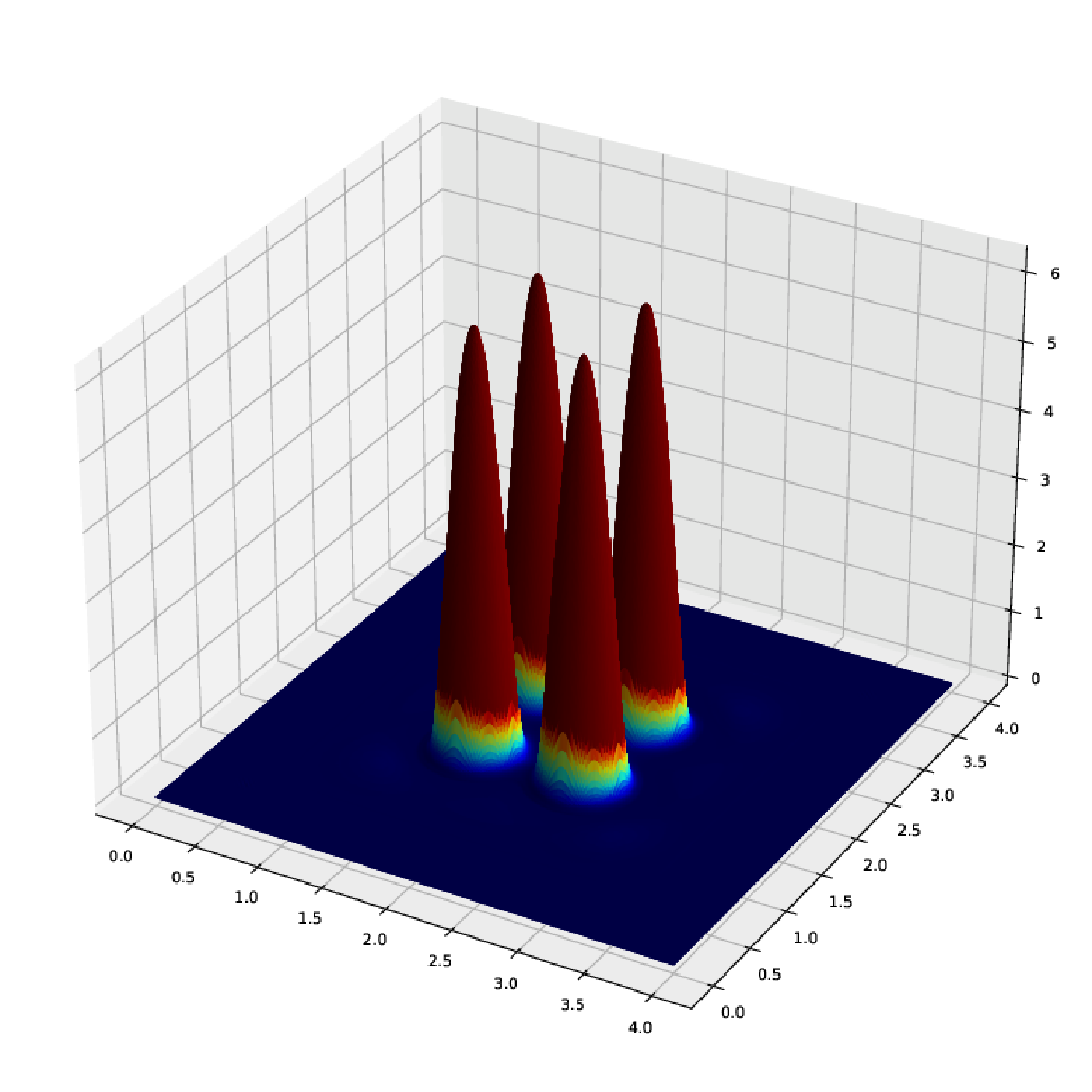}
\caption{$t=10$}
\end{subfigure}
\begin{subfigure}[b]{0.20\textwidth}
\centering
\includegraphics[width=\textwidth]{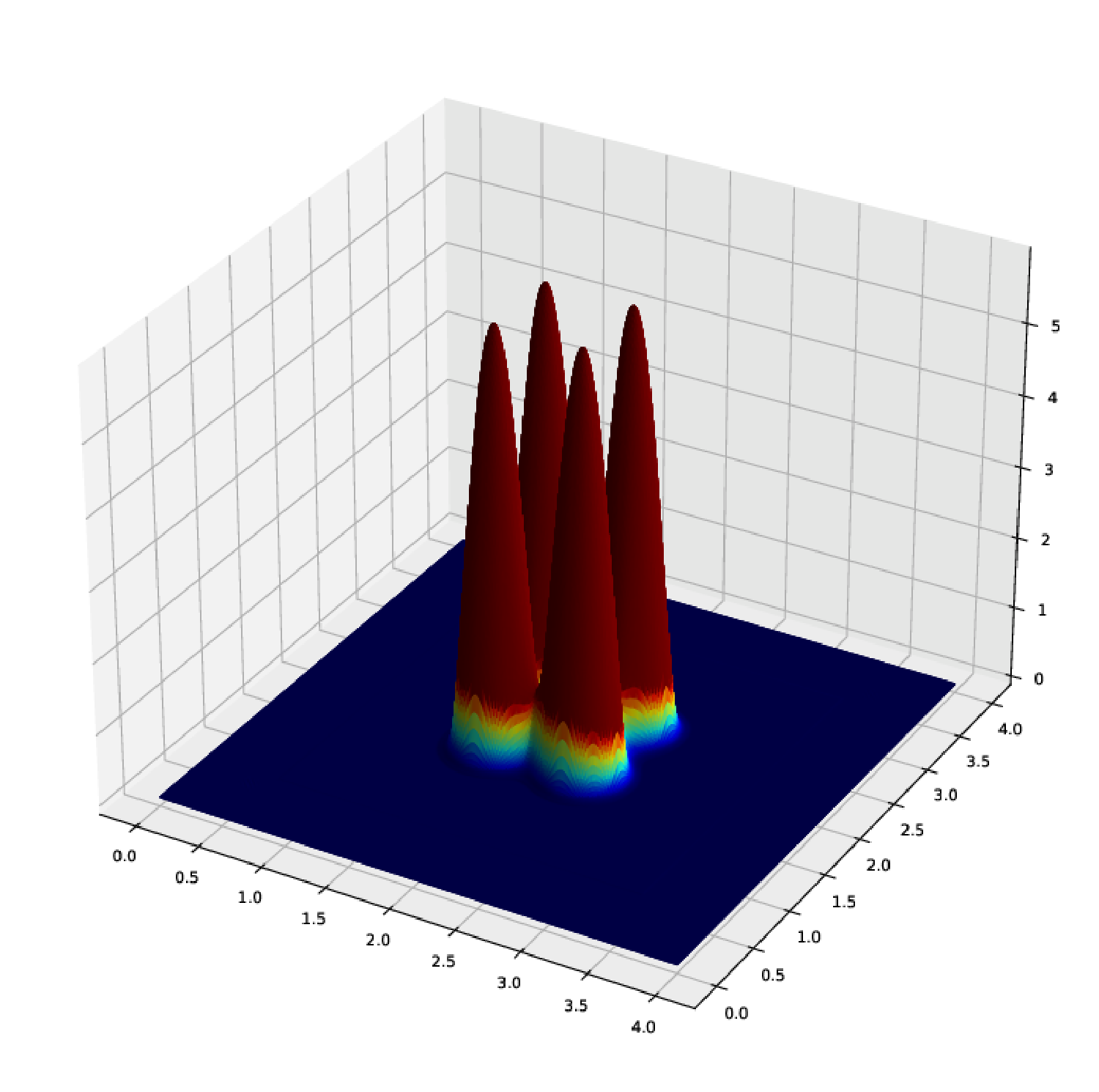}
\caption{$t=175$}
\end{subfigure}
\begin{subfigure}[b]{0.20\textwidth}
\centering
\includegraphics[width=\textwidth]{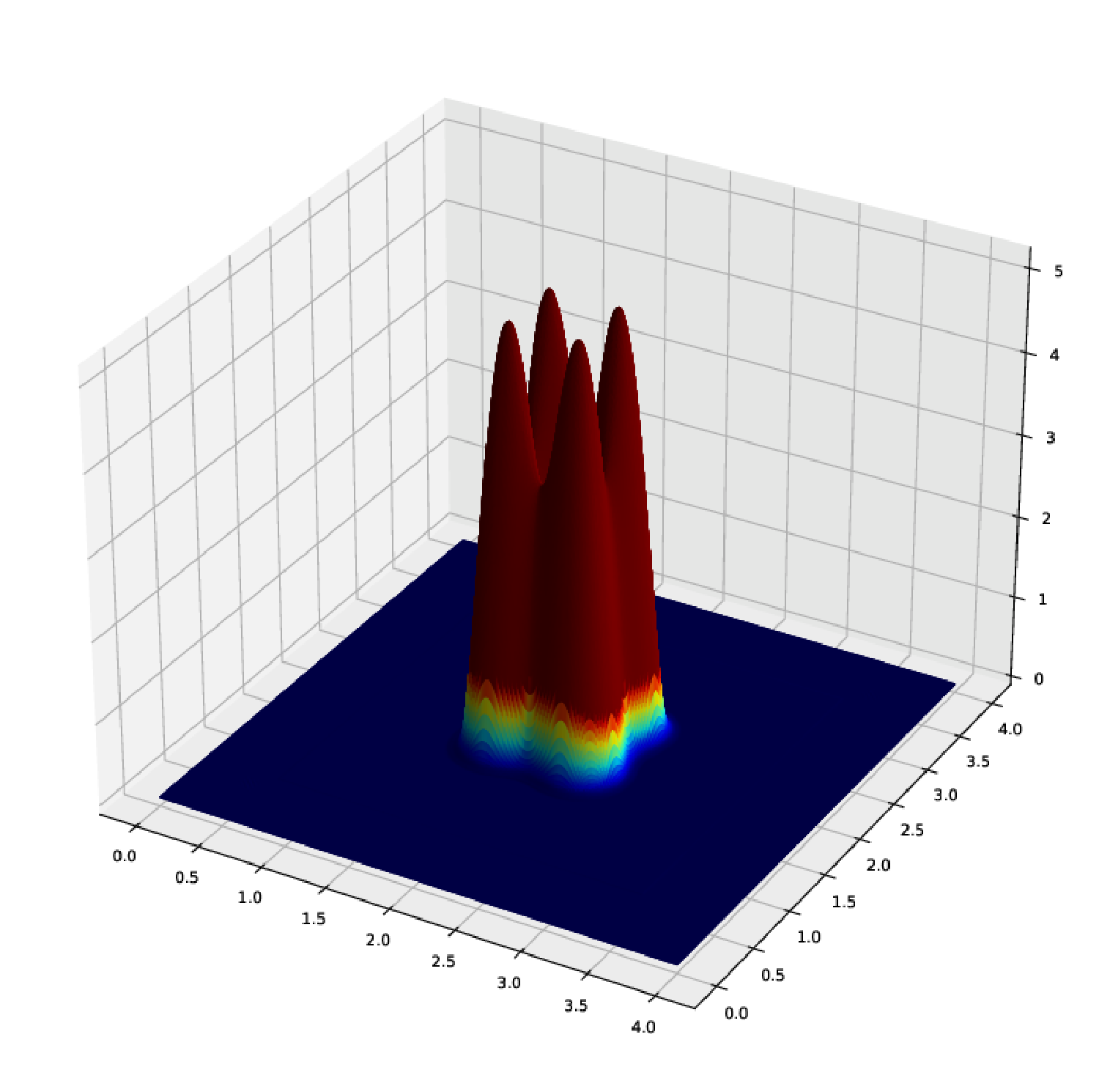}
\caption{$t=180$}
\end{subfigure}
\begin{subfigure}[b]{0.20\textwidth}
\centering
\includegraphics[width=\textwidth]{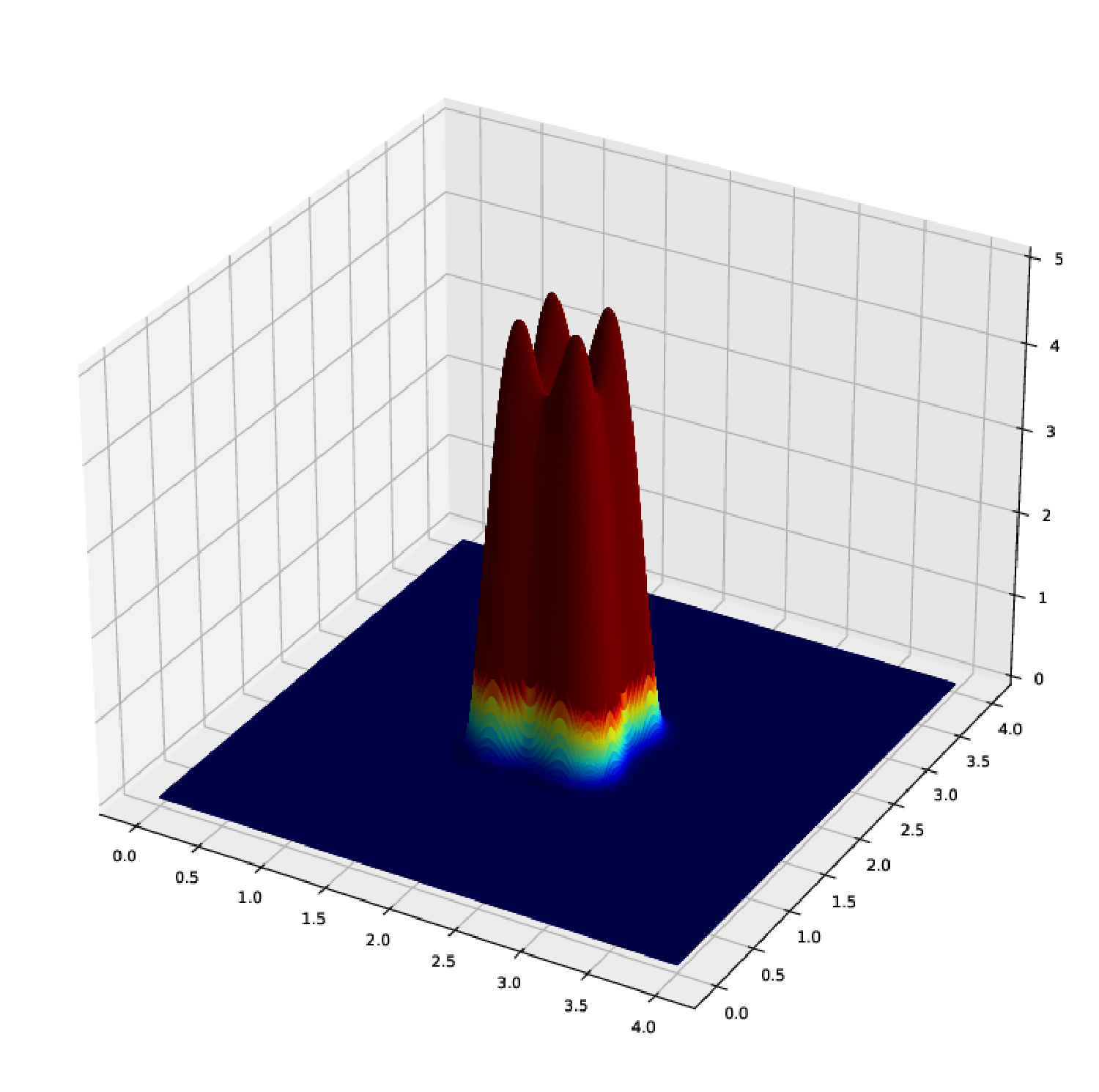}
\caption{$t=181$}
\end{subfigure}
\begin{subfigure}[b]{0.20\textwidth}
\centering
\includegraphics[width=\textwidth]{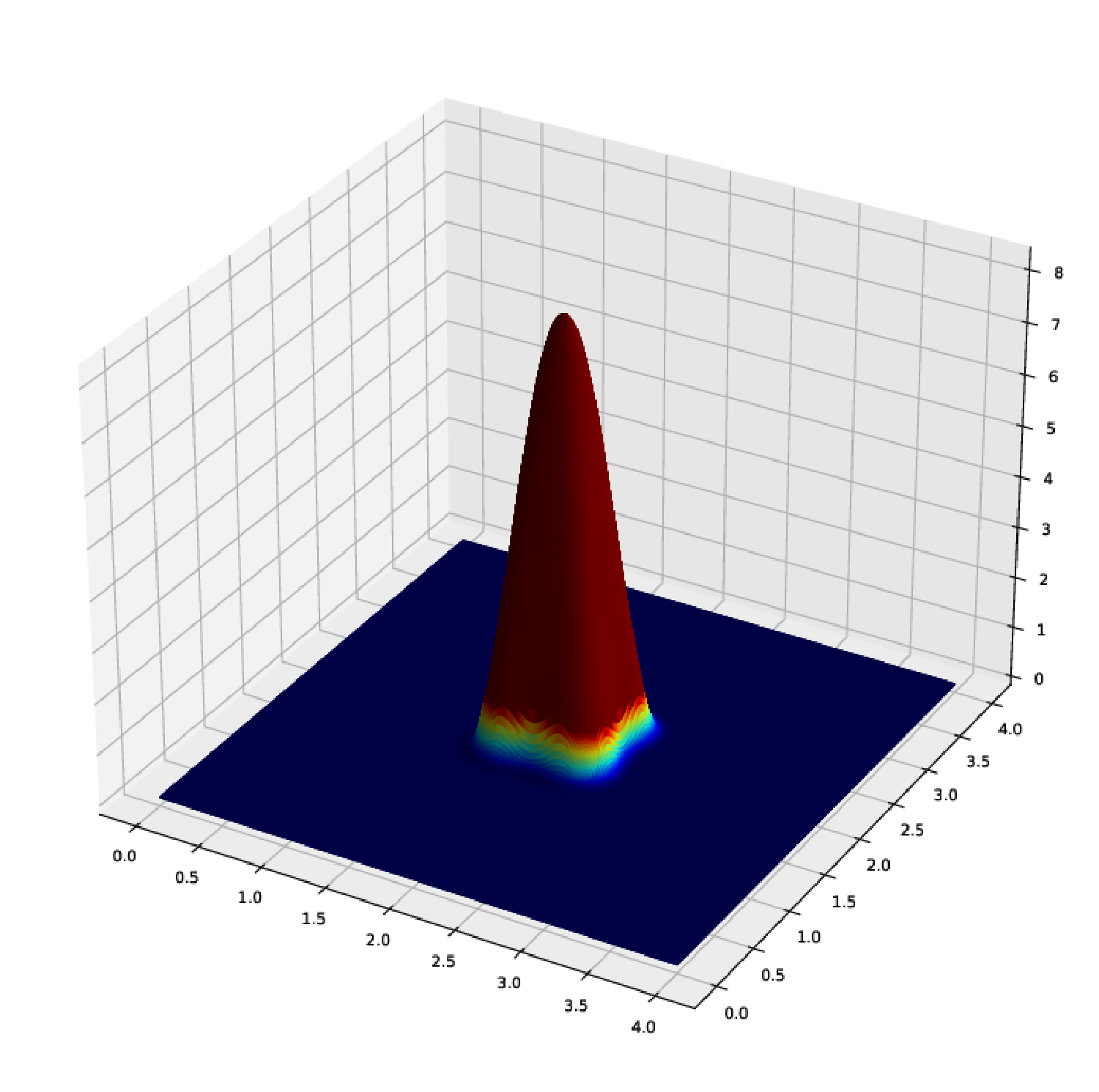}
\caption{$t=182$}
\end{subfigure}
\begin{subfigure}[b]{0.20\textwidth}
\centering
\includegraphics[width=\textwidth]{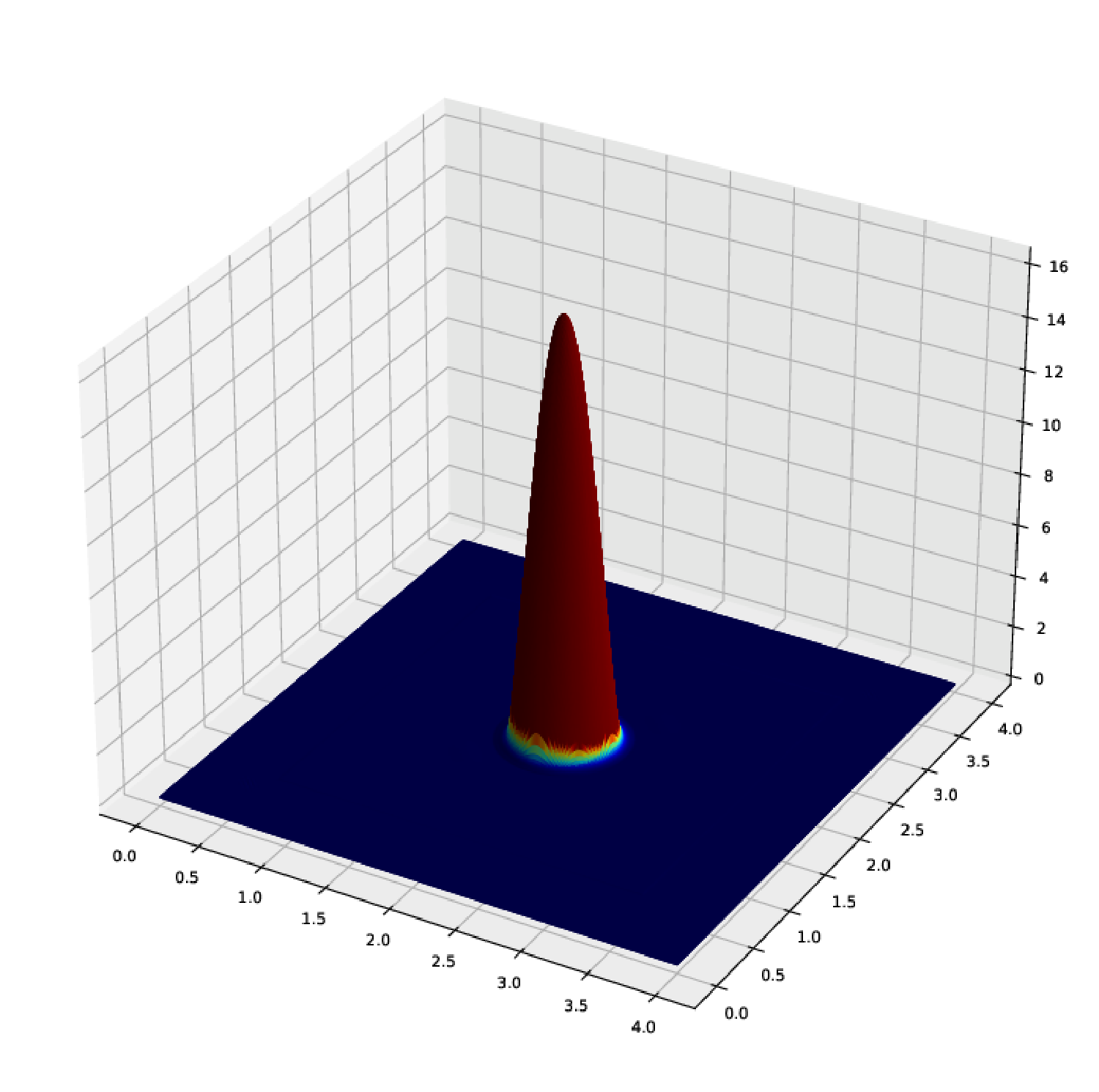}
\caption{$t=185$}
\end{subfigure}
\caption{The distribution of $l(x,t)$ over space $\Omega=[0,4]\times[0,4]$ for selected $t$ in the range 1-185 for the setting reported in Table \ref{tab:parameterValues}, but $h$ is 0.3.}
\label{fig:numericalInvestigationh=0.3}
\end{figure}

\section{Conclusions\label{sec:conclusions}}

We have shown how a complex spatial pattern can emerge from uncoordinated mobility decisions of workers. Even though workers are myopic in the sense that they maximise their instantaneous utility considering only the current distribution of other workers, the aggregate dynamics exhibit non-decreasing social utility, and in the stationary equilibrium distribution the efficient spatial allocation of workers seems guaranteed. Agglomeration forces, spatial congestion and diffusion driven by idiosyncratic preferences shape the spatial distribution of workers. In particular, when diffusion dominates agglomeration, the uniform distribution prevails, while in the opposite case a city (or more than one) can emerge, at least for a long period. The model can reproduce the emergence of cities of different sizes, shapes and spatial organisation, as well as the phenomenon known as metastability, in which a large city emerges from the abrupt agglomeration of several medium-sized cities that were apparently stable (persistent) for a long period.
The initial spatial distribution of workers turns out to be a strong determinant of the location and size of cities (history matters and path dependence is pervasive), as is the presence of exogenous spatial heterogeneity (such as the sea, rivers and mountains). 

The analysis can be extended in a number of ways.
The first is to introduce more realistic details of topography (coastline, rivers, mountains, etc.), natural resources (mines, fertile land, etc.) and infrastructure (roads, railways, etc.) for modelling $G(x)$ and $c_M$ as in \cite{allen2014trade} and \cite{redding2024quantitative}; in particular, the cost of moving should be a function of the location of the worker, its direction of movement and the location of the other workers, i.e. $c_{M}=c_{M}\left(X_t^{i,N},dX_t^{i,N},\mathbf{X}^{N}_{t}\right)$. While the dependence on the location of the worker and the location of the other workers does not pose any significant technical challenges, the inclusion of the direction requires a substantial modification of our framework \citep{flandoli2021navier}.
In the same respect, the study of the model in the case of a small sample, i.e. not in the limit of $N \rightarrow \infty$, should increase its explanatory power, especially for economies with a small population \citep{Evers2016}.
A further line of research is to analyse the ability of the model to reproduce Zipf's law and, more generally, the determinants of the city size distribution in a worker-based framework as in \cite{arshad2018zipf}.
Another relevant extension is the consideration of commuting, i.e. the possibility that a worker lives in a different place from his workplace \citealp{krugman1996self}. This extension is particularly important for the study of the internal structure of cities and the related phenomena of segregation and unequal spatial distribution of wages \citep{eberts1981empirical,albouy2015driving}.
Finally, the most difficult task is to extend the model to intertemporally optimising workers in continuous space. \cite{bardi2021convergence} describes a possible theoretical framework based on the \textit{Mean-Field Game} approach.

Returning to the initial question of the paper - why France, Germany and Italy exhibit such different spatial patterns of agglomeration - our model suggests that the answer is far from simple. While agglomeration, congestion, and diffusion forces play a role, the spatial equilibrium distribution of workers is also crucially shaped by the initial distribution of population, topography, and infrastructure.
Empirical research in this area faces an additional challenge: the continuous time-space formulation limits the applicability of standard spatial econometric techniques. To address this, \cite{fiaschiParentiRicci2023} introduce an innovative discretisation method for continuous time-space models that allows the separate identification of these determinants.
A key limitation of this approach is its reliance on high-resolution datasets, such as the \emph{Global Human Settlement Layer}\footnote{\url{https://human-settlement.emergency.copernicus.eu/ghs_pop2023.php}.} for residential population and NASA's \emph{Black Marble}\footnote{\url{https://blackmarble.gsfc.nasa.gov/}.} for local economic activity inferred from night lights. However, the increasing availability of high-resolution satellite imagery and geo-referenced census data is expected to alleviate this limitation.



\medskip
\noindent\textbf{Acknowledgements}: The authors have been supported by the Italian Ministry of University and Research (MIUR), in the framework of PRIN project 2017FKHBA8 001 (The Time-Space Evolution of Economic Activities: Mathematical Models and Empirical Applications).
Davide Fiaschi has been supported by the University of Pisa, in the framework of the PRA Project PRA\_2022\_86 (Mobilità di persone e merci nell'Europa).\\
\noindent\textbf{Declaration of interest}: The authors declare that they have no known competing financial interests or personal relationships that could have appeared to influence the work reported in this paper.\\
\noindent\textbf{Code availability}: All the codes are free to use and available at the link: \url{https://github.com/PRINNEtimeSpaceEconAct/randomUtilityModelPopulation.git}.
\clearpage

\bibliographystyle{chicago}
\bibliography{biblio}

\clearpage

\appendix
\counterwithin*{equation}{section}
\renewcommand\theequation{\thesection\arabic{equation}}
\textbf{ \Large Appendix}

\section{Sketch of the proof of Theorem \ref{teo:limitInfiniteAgentsII}  \label{app:proofMeanFieldLimit}}

In this appendix, we discuss a sketch of the proof of Theorem \ref{teo:limitInfiniteAgents} in the case where the utility is specified in Eq. \eqref{eq:movementAgent_IIEcon}, that is Theorem \ref{teo:limitInfiniteAgentsII}. As the proof is not particularly technical compared to the literature, in order to make the argumentation more effective for the reader, some details of the proof are skipped and we make some nonrestrictive assumptions \citep{bellomo2017active, bellomo2019active}. In particular, we assume in the proof that i) there is no change in the total number of workers, i.e. $n(x,t) = 0$; the case $n(x,t) \neq 0$ can be treated by introducing some technicalities related to Poisson processes, but the proof only becomes more involved without any additional significant conceptual difficulty (see, e.g., \citealp{catellier2021mean}).
Moreover, ii) we also assume that there are no endogenous amenities, i.e. $A_{0}=0$; their presence would only add some additional computations to the proof, since the main technical difficulties are already present in the expression for wages.
Finally, iii) we assume as domain $\Omega$ the \textit{two-dimensional torus} $\TT^{2} = \RR^{2}/\mathbb{Z}^2$ in order to neglect the boundary effects; this is not restrictive since the same technique of proof can be applied in any domain $\Omega \subseteq \RR^2$ by introducing a confining potential at the boundary of $\Omega$ to prevent workers from escaping the desired domain. 

Assume that at time $t = 0$ workers are independently randomly distributed in the domain $\Omega$, following a common probability density distribution on $\Omega$ called $l_0(x)$. For $t > 0$ the location of each worker evolves according to Eq. \eqref{eq:movementAgent_II}, i.e:
\begin{eqnarray}  \nonumber
dX^{i,N}_{t} &=& \left(\frac{1}{c_{M}}\right)\left[ \nabla_{x} w^R\left(\mathbf{X}_t^N\right)\left(X^{i,N}_{t}, t\right) + \nabla_{x}A_{ES}\left(X^{i,N}_{t},t\right)\right]dt + \left(\frac{1}{c_{M}}\right)\sigma dB^{i}_{t},
\end{eqnarray}
or by making each term more explicit,
\begin{eqnarray}  \nonumber
dX^{i,N}_{t} &=& \left(\frac{1}{c_{M}}\right) G(X^{i,N}_{t})\left[\sum_{j=1}^{N} \frac{1}{N}W_h^P(X^{i,N}_{t}-X^{j,N}_t) \right]l^{N}(X^{i,N}_{t},t)^{\beta-1} \\
&+& \left(\frac{1}{c_{M}}\right)\nabla_{x}A_{ES}\left(X^{i,N}_{t},t\right) dt \nonumber \\
&+& \left(\frac{\sigma}{c_{M}}\right) dB^{i}_{t}.
\label{eq:AppAgentMovement} 
\end{eqnarray}
Consider the empirical distribution of all worker locations:
\begin{equation}
E^{N}_t := \frac{1}{N}\sum_{i=1}^{N} \delta_{X^{i,{N}}_t},
\end{equation}
where $\delta_z$ is the random variable on $\TT^2$ with unitary mass at point $z$. $E^{N}_t$ is a continuous set of random variables on $\TT^2$ depending on time. For any given $N \in \NN$ and $t>0$, $E^{N}_t$ is singular, in the sense that it is a distribution over $\TT^2$ that does not admit a probability density function, since it has a positive probability only on a finite set (corresponding to the location of the $N$ workers). However, as $N$ tends to infinity, the family of random variables $E^{N}_t$ becomes diffuse and converges (in distribution) to a continuous family of random variables over $\TT^2$, denoted by $E_t$ for any $t > 0$. The distribution of $E_t$ is regular and admits a probability density function for each $t$, called $l(t,x)$. An explicit expression for $l(t,x)$ for each $t$ is not available. 
However, we can prove that the probability density function $l(t,x)$ is the solution to Eq.  \eqref{eq:dynamicsLabourDistribution} of Theorem \ref{teo:limitInfiniteAgents} under the assumption stated above. 

Recall that the stock of workers at location $x$ at time $t$ is defined by
\begin{equation*}
l^N(x,t) = \sum_{i = 1}^{N} \frac{1}{N}\cdot\theta_N(x-X^{i,N}_t),
\end{equation*}
where $\theta_{N}(x) = N^{2\lambda}\theta(N^{\lambda}x)$, $\theta:\TT^{2}\to \RR$ is a function that is not negative, $C^{\infty}$, and integrates to one. The function $l^{N}(x,t)$ can also be described as a \emph{mollified empirical measure} of the stock of workers at location $x$ at time $t$.
The parameter $\lambda$ can be taken in the interval $[0,1]$. However, in our framework some choices of $\lambda$ are not possible. In particular 
\begin{itemize}
\item $\lambda = 0$: this corresponds to the classical \textit{mean field} case, where the radius of interaction is fixed. This choice leads to non-local macroscopic dynamics, which is not suitable for our purposes;
\item $\lambda \in (0,1)$: this corresponds to the theory of \textit{moderate interactions} originally developed in \cite{oelschlager1985law} and extensively generalised in the literature (see e.g. \citealp{meleard1987propagation}). Our results fall into this case, i.e. we assume $\lambda \in (0,\overline{\lambda})$ with $\overline{\lambda} < 1$;
\item $\lambda = 1$: this corresponds to the case of \textit{local interactions}, which is still not well understood in the literature (see e.g. \citealp{varadhan1991scaling,uchiyama2000pressure,flandoli2020macroscopic}).
\end{itemize}

\subsection{Preliminaries}
In this section, we provide some preliminaries needed for the proof of Theorem \ref{teo:limitInfiniteAgentsII}. In particular, the definition of a \textit{weak} solution of Eq. \eqref{eq:dynamicsLabourDistribution} and the appropriate measure space with a metric corresponding to weak convergence of the probability measure.
\begin{defi}[Weak solution of Eq. \eqref{eq:dynamicsLabourDistribution}]
A function $l$ of suitable regularity is a \textit{weak solution} of the system of Eq. \eqref{eq:dynamicsLabourDistribution} if for each test function $\phi \in C^{\infty}(\TT^{2}; \RR)$ the following holds
\begin{eqnarray} \nonumber
\ang{l(t),\phi} &=&  \ang{l(0),\phi} + \frac{ \sigma^2}{2c_{M}^2}\int_{0}^{t}\ang{l(s),\Delta_{x}\phi}\,ds +\\ 
&+& \left(\frac{1}{c_{M}}\right) \int_{0}^{t}\ang{l(s),\nabla_{x}\phi \cdot \nabla_{x}\left( G\, (W_{h}^{P} * l)(s) \, l(s)^{\beta-1} \right)}\,ds +\nonumber \\
&+& \left(\frac{1}{c_{M}}\right) \int_{0}^{t}\ang{l(s),\nabla_{x}\phi \cdot \nabla_{x}A_{ES}(s)}\,ds .
\end{eqnarray}
\end{defi}
Introduce the \textit{empirical measure} of the workers' location:
\[
S^{N}_{t}(dx) = \frac{1}{N}\sum_{i=1}^{N}\delta_{X^{i,N}_{t}}(dx),
\]
which is a random element of the space $C([0,T];\PP_{1}(\TT^{2})$), where 
\begin{equation}\label{def:metricSpace}
\PP_{1}(\TT^{2}) := \left\{ \mu \text{ probability measure on } (\TT^{2},\mathcal{B}(\TT^{2})) \,\Bigg\vert\, \int_{\TT^{2}}\abs{x}\,\mu(dx) < \infty\right\}
\end{equation}
is the space of all probability measures on the Borel sets of $\TT^{2}$ with finite first moment. We endow this space with the \textit{Wasserstein$-1$ metric}, which can be defined equivalently as
\[
\mathcal{W}_{1}(\mu,\nu) := \sup_{ [\phi]_{Lip}\leq 1}\abs{ \int_{\TT^{2}} \phi \, d\mu - \int_{\TT^{2}} \phi \,d\nu },
\]
where $[\phi]_{Lip}$ is the usual Lipschitz seminorm. Note that this is not the usual definition of Wasserstein metrics, but rather a characterisation provided by a Theorem called the Kantorovich-Rubinstein characterisation \citep{KR:58}. Equipped with this metric, the space $\PP_{1}(\TT^{2})$ becomes a completely separable metric space, whose convergence implies the weak convergence of probability measures.

\subsection{Candidate for the limit of infinite \textit{N} }
For any test function $\phi$ it holds:
\[
\ang{S^{N}_{t},\phi} = \int_{\TT^{2}}\phi(x)S^{N}_{t}(dx) = \frac{1}{N}\sum_{i=1}^{N}\phi(X^{i,N}_{t});
\]
So, by taking the family of functions parameterised by $x$ $\phi(\cdot) = \theta_{N}(x-\cdot)$ we have 
\[
\ang{S^{N}_{t},\theta_{N}(x-\cdot)} = \frac{1}{N}\sum_{i=1}^{N}\theta_{N} (x-X^{i,N}_{t}) = l^{N}(x,t). 
\]

Taking the above expressions, we can calculate their change over time by \textit{It\^o Formula}.
\begin{lem}[It\^o Formula] 
For any test function $\phi$ it holds
\begin{eqnarray}\nonumber
d\phi(X^{i,N}_{t}) &=& \frac{\sigma^{2}}{2c_{M}^2}\Delta_{x} \phi(X^{i,N}_{t})\,dt \\
&+& \left(\frac{1}{c_{M}}\right)\nabla\phi(X^{i,N}_{t})\cdot \nabla w^{N}(X^{i,N}_{t},t) \, dt\nonumber \\
&+& \left(\frac{1}{c_{M}}\right)\nabla\phi(X^{i,N}_{t})\cdot \nabla A_{ES}(X^{i,N}_{t},t) \, dt \nonumber\\
&+& \left(\frac{\sigma}{c_{M}}\right) \nabla \phi(X^{i,N}_{t})\cdot dB^{i}_{t};
\label{eq:firstPart}
\end{eqnarray}
\begin{eqnarray}\nonumber
d\ang{S^{N}_{t},\phi} &=& \frac{\sigma^{2}}{2c_{M}^2} \ang{S^{N}_{t},\Delta_{x} \phi}\,dt \\
&+& \left(\frac{1}{c_{M}}\right)\ang{S^{N}_{t},\nabla_{x}\phi \cdot \nabla_{x}w^{N}(t,\cdot)}\,dt \nonumber \\ 
&+& \left(\frac{1}{c_{M}}\right)\ang{S^{N}_{t},\nabla_{x}\phi \cdot \nabla_{x}A_{ES}(t,\cdot)}\,dt \nonumber \\
&+& \left(\frac{\sigma}{c_{M}}\right){N}\sum_{i=1}^{N}\nabla_{x} \phi(X^{i,N}_{t})\cdot dB^{i}_{t},
\label{eq:secondPart} 
\end{eqnarray}
and
\begin{eqnarray}\label{eq:appItoEm}\nonumber
dl^{N}(x,t) &=& \frac{\sigma^{2}}{2c_{M}^2} \Delta_{x} l^{N}(x,t)\, dt \nonumber \\
&-& \left(\frac{1}{c_{M}}\right) \frac{1}{N} \sum_{i=1}^{N} \nabla\theta_{N}(x-X^{i,N}_{t})\cdot \nabla w^{N}(X^{i,N}_{t},t) \,dt \nonumber \\
&-& \left(\frac{1}{c_{M}}\right) \, \div_x \left( l^N(x,t) \nabla_{x}A_{ES}(x,t) \right)\,dt  \nonumber \\
&+& \frac{\sigma}{Nc_{M}}\sum_{i=1}^{N}\nabla_{x} \theta_{N}(x-X^{i,N}_{t})\cdot dB^{i}_{t}.
\end{eqnarray}
\end{lem}
\begin{proof}
Eq. \eqref{eq:firstPart} directly follows by applying It\^o formula to the function $\phi(X^{i,N}_{t})$ and using Eq \eqref{eq:AppAgentMovement}. Eq. \eqref{eq:secondPart} follows from Eq. \eqref{eq:firstPart} using the linearity of the sum, while Eq. \ref{eq:appItoEm} is obtained by taking the family of functions parameterised by $x$ $\theta_{N}(x-\cdot)$ and applying the It\^o formula, then again using the linearity of the sum. 
\end{proof} 

The expression for the drift in the second line on the right-hand side of Eq. \eqref{eq:appItoEm} can be written as
\begin{multline*}
-\frac{1}{N}\sum_{i=1}^{N}  \nabla_{x}\theta_{N}(x-X^{i,N}_{t})\cdot \nabla_{x} w^{N}(X^{i,N}_{t},t) =\\
=-\int_{\TT^{2}}\nabla_{x}\theta_{N}(x-x')\cdot \nabla_{x} w^{N}(x',t)\,S^{N}_{t}(dx') = \\
-\div_{x} \int_{\TT^{2}}\theta_{N}(x-x')\nabla_{x} w^{N}(x',t)\,S^{N}_{t}(dx').
\end{multline*}
However, this is not the correct expression to obtain the desired limit of infinite $N$ in Eq. \eqref{eq:dynamicsLabourDistribution}; instead, the correct expression is:
\[
-\div_{x} \int_{\TT^{2}}\theta_{N}(x-x')\nabla_{x} w^{N}(x,t)\,S^{N}_{t}(dx') = -\div_{x} (l^{N}(x,t)\nabla_{x} w^{N}(x,t)).
\]
Hence, we rewrite Eq. \eqref{eq:appItoEm} as the proper expression, plus a remainder that we will have to show vanishes in the limit of infinite $N$:
\begin{eqnarray}\label{eq:appEmWithCommutator}\nonumber
dl^{N}(x,t) &=& \frac{\sigma^{2}}{2c_{M}^2} \Delta_{x} l^{N}(t)\, dt \nonumber \\
&-& \left(\frac{1}{c_{M}}\right) \, \div_{x} (l^{N}(x,t)\nabla_{x} w^{N}(x,t))\,dt  \nonumber \\
&-& \left(\frac{1}{c_{M}}\right) \, \div_x \left( l^N(t) \nabla_{x}A_{ES}(x,t) \right)\,dt  \nonumber \\
&+& R^{N}_{t}(x) \, dt + dM^{N,\theta_{N}}_{t}(x),
\end{eqnarray}
where 
\[
M^{N,\theta_{N}}_{t}(x) = \frac{\sigma}{Nc_{M}}\sum_{i=1}^{N}\nabla_{x} \theta_{N}(x-X^{i,N}_{t})\cdot B^{i}_{t}
\] 
is a martingale, and 
\[
R^{N}_{t}(x) = \left(\frac{1}{c_{M}}\right)\div_{x} \int_{\TT^{2}}\theta_{N}(x-x')\left[\nabla_{x} w^{N}(x,t)-\nabla_{x} w^{N}(x',t)\right]\,S^{N}_{t}(dx').
\]
Eq. \ref{eq:appEmWithCommutator} gives a candidate for the limit behaviour of the system of Eq. \eqref{eq:AppAgentMovement}, with the intuition that the last two terms should converge to zero as $N$ goes to infinity.

\subsection{Steps for the proof of Theorem \ref{teo:limitInfiniteAgentsII} }

The proof of Theorem \ref{teo:limitInfiniteAgentsII} needs some intermediate lemmas that one has to prove first. The general strategy is that of \emph{compactness} and consists of three main steps. 

First, one has to prove that the sequence $(S^N_{t})_{N \in \NN}$ lies uniformly in some compact subset of $\PP_{1}(\TT^{2})$.\footnote{The notion of \textit{compactness} here is the classical one of compact sets in metric spaces.} This leads to the existence of limits of each subsequence. In our case, this follows from a uniform estimate in $L^1(\Omega)$ for each of the workers' locations. Since the convergence is induced by the metric used on $\PP_{1}(\TT^{2})$ (see Eq. \eqref{def:metricSpace}), this only holds in the distribution. However, it is well known that convergence in the distribution of random variables to a deterministic limit implies convergence in probability.

$\bullet$ (\textit{Characterisation}) Second, one has to give a characterisation of the limit points as solutions of Eq. \eqref{eq:dynamicsLabourDistribution}. This usually follows from the convergence of each term in Eq. \eqref{eq:appItoEm} to its corresponding term in the limit equation for infinite $N$.

$\bullet$ (\textit{Uniqueness}) Finally, one has to prove a uniqueness result for Eq. \eqref{eq:dynamicsLabourDistribution}. This last step is necessary to go from the convergence of each subsequence of the first step to the convergence of the entire sequence. The reason is as follows: since each subsequence converges to the solutions of Eq. \eqref{eq:dynamicsLabourDistribution}, if the solution is unique, then every subsequence converges to the same limit. One then just has to realise that in any metric space, if every subsequence converges to the same limit, the same holds for the whole sequence.

The key points for each of the three steps are as follows

$\bullet$ (\textit{Compactness}) One must prove that there exists a constant $C$ which does not depend on $N$, such that 
\begin{equation}\label{eq:LnNablaL2}
\EE{\norm{\nabla_{x} l^{N}(t)}^2_{L^{2}(\TT^{2})}}\leq C.
\end{equation}
Again, we need to make additional assumptions about the production function. In particular, we will assume that it is equal to $l^\beta$ when $l$ is greater than some small threshold $\underline{l}$, and instead has the shape of a parabola for values of $l$ between zero and $\underline{l}$. A possible expression is given by
\begin{equation*}
f(l) = 
\begin{cases}
l^\beta + (-\frac{1}{2} \beta^2 + \frac{3}{2} \beta - 1) \underline{l}^\beta & \text{ if } l > \underline{l}, \\
\frac{1}{2} \beta (\beta-1) \underline{l}^{\beta-2} l^2 + \beta (2-\beta) \underline{l}^{\beta-1} l & \text{ if } 0 \leq l \leq \underline{l}.
\end{cases}
\end{equation*}
In this way we keep the usual assumption on the production function, i.e. i) positive first derivative and ii) negative second derivative, but the first derivative is now globally bounded for all values of $l$. 
This first step is necessary to complete the first step on (\textit{Compactness}) of the sequence $(S^{N}_{t})_{N \in \NN}$ in $\PP^{1}(\TT^{2})$. 

$\bullet$ (\textit{Characterisation}) Suppose $S^{N}_{t} \stackrel{N\to\infty}{\to}  \mu_{t}$ in $\PP_{1}(\TT^{2})$, where $\mu_{t}$ is a probability measure with density $l(x,t)$ (or analogously for a subsequence $S^{N_{k}}_{t}$ from the first step). Then we need to check that the expression for wages converges to the analogous limit expression for wages
\begin{multline*}
w^{N}(x,t) =\\ = G(x) \frac{1}{N}\sum_{i=1}^{N}W_{h}^{P}(x-X^{i,N}_{t}) \left(\frac{1}{N}\sum_{i=1}^{N}\theta_{N}(x-X^{i,N}_{t}) \right)^{\beta-1}\stackrel{N\to\infty}{\longrightarrow} G(x)(W_{h}^{P} * l)(x,t) l(x,t)^{\beta-1} =\\ = w(x,t).
\end{multline*}	
This part is necessary to give a characterisation of the limit points (\textit{Characterisation}). 
We also need to show that the martingale term $M^{N,\theta_{N}}_{t}(x) $ goes to zero in probability as $N$ goes to infinity. However, this is easily done since we have an explicit expression for the quadratic variation
\begin{equation*}
\left[M^{N,\theta_{N}}\right]_{t} = \frac{\sigma^{2}}{N^{2}c_{M}^2}\sum_{i=1}^{N} \abs{\nabla_{l}\theta_{N}(x-X^{i,N}_{t})}^{2}
\end{equation*}
since the Brownian motions $B^{i}_{t}$ are independent, and by using classical martingale inequalities. In this step we will have to assume $\lambda < 1/4$ to prove such an inequality, and use Eq. \eqref{eq:LnNablaL2}.
Finally, we need to show that the remainder $R^{N}_{t}(x)$ also goes to zero as $N$ goes to infinity. This again follows from the estimate in Eq. \eqref{eq:LnNablaL2}, using the fact that the difference $\left[\nabla_{x} w^{N}(x)-\nabla_{x} w^{N}(x')\right]$ is non-zero only when $x$ and $x'$ are close, since it is multiplied by $\theta_{N}(x-x')$, whose support shrinks with $N$. \footnote{See for example the proof of \cite[Lemma 5.5]{flandoli2021navier}} 

$\bullet$ (\textit{Uniqueness}) To show the uniqueness of the solutions of Eq. \eqref{eq:dynamicsLabourDistribution} one follows the classical strategy of assuming the existence of two different solutions $l^1(x,t)$ and $l^2(x,t)$ with the same initial condition at time $t = 0$. Then consider the difference between the two and study the evolution in the $L^{2}$-norm in $\TT^{2}$, i.e:
\begin{equation*}
\norm{l^1(t)-l^2(t)}_{L^{2}(\TT^{2})}^{2} = \int_{\TT^{2}} \abs{l^1(x,t)-l^2(x,t)}^{2}\,dx 
\end{equation*}
and show that it satisfies
\begin{equation*}
\norm{l^1(t)-l^2(t)}_{L^{2}(\TT^{2})}^{2} \leq C_{T}\int_{0}^t\norm{l^1(s)-l^2(s)}_{L^{2}(\TT^{2})}^{2} \,ds
\end{equation*}
for some constant $C_{T}$. Then \textit{Gronwall Lemma} implies that if $\norm{l^1(0)-l^2(0)}_{L^{2}(\TT^{2})} = 0 $ then $\norm{l^1(t)-l^2(t)}_{L^{2}(\TT^{2})} = 0$ for all $t > 0$, so solutions are unique. This is a proof of (\textit{Uniqueness}).

\section{Proof of Theorem \ref{teo:centerofmass}\label{app:proofCenterOfMass}}
\begin{proof}
In our hypothesis one can rewrite Eq. \eqref{eq:dynamicsLabourDistributionII} as
\begin{multline}\label{eq:proofCenterOfMass}
\partial_t l(x,t) = \dfrac{\sigma^2}{2c_M^2} \Delta_x l\left(x,t\right) - \left(\frac{1}{c_M}\right) \, \div_x \left( l(x,t) \nabla_x w\left(x,t\right) \right) = \\ 
= \dfrac{\sigma^2}{2c_M^2} \Delta_x l\left(x,t\right) - \left(\frac{1}{c_M}\right) \, \div_x \left( l(x,t) \nabla_x (W_{h} * l)\left(x,t\right) \right).
\end{multline}
Then calculate the evolution of the position of the centre of mass:
\begin{multline}\label{eq:proofCenterOfMass2}
\frac{d}{dt} \int_{\Omega} x l(x,t)\,dx = \int_{\Omega} x \partial_{t}l(x,t)\,dx = 
\dfrac{\sigma^2}{2c_M^2}\int_{\Omega}x \Delta_x l\left(x,t\right)\,dx - \\ 
- \frac{1}{c_M}\int_{\Omega}x \, \div_x \left( l(x,t) \nabla_x (W_{h} * l)\left(x,t\right) \right)\, dx.
\end{multline}
We will now show that both terms of Eq. \eqref{eq:proofCenterOfMass2} vanish. By integration by parts, the first term becomes
\begin{multline*}
\dfrac{\sigma^2}{2c_M^2}\int_{\Omega}x \Delta_x l\left(x,t\right)\, dx = -\dfrac{\sigma^2}{2c_M^2}\int_{\Omega}\div_{x} x \cdot \nabla_x l\left(x,t\right)\, dx = \\
= -2\dfrac{\sigma^2}{2c_M^2}\int_{\Omega} 1 \cdot \nabla_x l\left(x,t\right)\, dx = 2\dfrac{\sigma^2}{2c_M^2}\int_{\Omega} \nabla_{x} 1 \cdot l\left(x,t\right)\, dx = 0
\end{multline*}
since $\nabla_{x} 1 = 0$.
For the second term of Eq. \eqref{eq:proofCenterOfMass2} we have, again using integration by parts
\begin{multline*}
- \frac{1}{c_M}\int_{\Omega}x \, \div_x \left( l(x,t) \nabla_x (W_{h} * l)\left(x,t\right) \right)\, dx = \frac{1}{c_M}\int_{\Omega}\div_{x}x \, \left( l(x,t) \nabla_x (W_{h} * l)\left(x,t\right) \right)\, dx = \\
= 2\frac{1}{c_M}\int_{\Omega} l(x,t) \nabla_x (W_{h} * l)\left(x,t\right) \, dx = 2\frac{1}{c_M}\int_{\Omega} l(x,t) \int_{\Omega}\nabla_x W_{h}(x-y)l(t,y)\,dy \, dx = \\ 
= 2\frac{1}{c_M}\int_{\Omega} l(x,t) \int_{\Omega}\nabla_x W_{h}(x-y)l(y,t)\,dy \, dx = 2\frac{1}{c_M}\int_{\Omega} \int_{\Omega}\nabla_x W_{h}(x-y)l(y,t)l(x,t)\,dy \, dx.
\end{multline*}
Furthermore, by changing the variable
\begin{multline*}
2\frac{1}{c_M}\int_{\Omega} \int_{\Omega}\nabla_x W_{h}(x-y)l(y,t)l(x,t)\,dy \, dx = 2\frac{1}{c_M}\int_{\Omega} \int_{\Omega}\nabla_x W_{h}(y-x)l(y,t)l(x,t)\,dy \, dx = \\
= -2\frac{1}{c_M}\int_{\Omega} \int_{\Omega}\nabla_x W_{h}(x-y)l(y,t)l(x,t)\,dy \, dx
\end{multline*}
by using the symmetry about the origin of $W_{h}$. Therefore this last term, being equal to it's opposite, must be equal to zero. So we have shown that
\begin{equation*}
\frac{d}{dt} \int_{\Omega} x l(x,t)\,dx = 0
\end{equation*} 
which ends the proof. 
\end{proof}

\clearpage
\clearpage
\setcounter{section}{0}
\renewcommand\sectionname{Online Appendix}

\begin{center}
\textbf{\huge  Online Appendix}
\end{center}

\section{Micro-foundation of workers' movements \label{app:micro-foundationAgentsMovements}}

In this section, we illustrate the micro-foundations of workers' movement driven by spatial differential utilities and show how this can be derived as a Nash equilibrium. For the sake of simplicity, assume that the space dimension is one, i.e. $X^{i,N}_{t}\in \RR$; all the arguments made in the section can be repeated in the same manner, but with a more cumbersome notation, in any dimension $d > 1$. 

Assume that there is a utility function, common to  all workers, defined as:
\begin{eqnarray}
u:&\RR^{N}& \to \big[ \RR \times \RR^{+} \to \RR\big]\nonumber \\
&\mathbf{X}&\mapsto \big[(x,t) \mapsto v(\mathbf{X})(x,t)\big],
\end{eqnarray}
that is for any $\mathbf{X} := ({X}_{1},\dots,{X}_{N})\in \RR^{N}$, representing the location of all the workers, there exists a \textit{field of utility} $u(\mathbf{X}):\RR \times \RR^{+} \to \RR$ which, in any location $x$ and $t$, associates the corresponding utility.

Denote by
\begin{equation} 
\frac{\partial u}{\partial {X}_{k}}(\mathbf{X})(x,t), \quad	\frac{\partial u}{\partial t}(\mathbf{X})(x,t) \quad \text{ and } \quad \frac{\partial u}{\partial x}(\mathbf{X})(x,t)
\end{equation}
the partial derivative with respect to the $k$-th component of the vector $\mathbf{X}=(\mathbf{X}_{1},\dots, \mathbf{X}_{N})$, the variable $t$ and $x$, respectively. 
The partial derivative of $u$ with respect to the variable $X^{i,N}_{t}$ in the first component, evaluated in $\left(\mathbf{X}^{N}_{t}\right)(X^{i,N}_{t},t)$ is assumed to be equal to zero, i.e.
\begin{equation}\label{eq:utilityNoEffectOnItself}
\frac{\partial u}{\partial {X}_{i}}\left(\mathbf{X}^{N}_{t}\right)(X^{i,N}_{t},t) = 0;
\end{equation}
this reflects the fact that a change in the location of worker $i$ does not affect the spatial shape of utility in $X^{i,N}_{t}$. Therefore,
worker $i$ 's utility varies as it changes its position with respect to other workers, but the specific contribution of the change in its location to its overall utility is zero. In other words, the worker maintains the same distance by itself along all its movement over space. This assumption is common in the Mean Field Game literature \citep{carmona2018probabilistic}.

The dynamic over space of worker $i$ is the result of a maximization of utility over space, subject to moving costs. In particular, 
consider the \emph{instantaneous} variation of utility of worker $i$, $du^{i}_{t}$, as the result of moving from location $X^{i,N}_{t}$ to $X^{i,N}_{t} +  dt\, dX^{i,N}_{t}$, i.e:
\begin{equation}
du^{i}_{t} = \frac{\partial u}{\partial x}\left(\mathbf{X}^{N}_{t}\right)(X^{i,N}_{t},t)dX^{i,N}_{t} + \sum_{j=1}^{N}  \frac{\partial u}{\partial {X}_{j} }\left(\mathbf{X}^{N}_{t}\right)(X^{i,N}_{t},t)dX^{j,N}_{t},
\label{eq:changeIndividualUtility}
\end{equation}
where the first term of Eq. \eqref{eq:changeIndividualUtility} is the effect of movement of worker $i$, while the second term refers to the effects of other workers' movement.
Assuming quadratic adjustment cost for moving, at each time $t$ worker $i$ solves the following \emph{static} optimization problem:
\begin{equation}\label{eq:gameMaxInstantaneousUtilityImperfectInfo}
\max_{\{dX^{i,N}_{t}\in \RR\}} \frac{\partial u}{\partial x}\left(\mathbf{X}^{N}_{t}\right)(X^{i,N}_{t},t)dX^{i,N}_{t} + \sum_{j=1}^{N}  \frac{\partial u}{\partial {X}_{j}}\left(\mathbf{X}^{N}_{t}\right)(X^{i,N}_{t},t)dX^{j,N,e_{i}}_{t} -  \frac{c_{M}}{2}\norm{dX^{i,N}_{t}}^2,
\end{equation}
where $c_{M}$ is a parameter measuring the intensity of the moving costs,\footnote{In a more general setting, where the infrastructures, roads and their congestion are taken into account, the cost of moving should be a function of the location of worker, its direction of movement and the location of other workers, i.e. $c_{M}=c_{M}\left(X_t^{i,N},dX_t^{i,N},\mathbf{X}^{N}_{t}\right)$.} $\norm{\cdot}$ is the norm, and $dX^{j,N,e_{i}}_{t}$ is the \textit{expected movement} of worker $j$' formulated by worker $i$. At the time of the worker $i$'s decision, this information is not available to the worker and it is taken as given. 

Given Eq. \eqref{eq:utilityNoEffectOnItself}, the first order condition of Problem \ref{eq:gameMaxInstantaneousUtilityImperfectInfo} reduces to:
\begin{equation}\label{eq:optimalMovement}
dX^{i,N}_{t} = \frac{1}{c_{M}}\left[\frac{\partial u}{\partial x}\left(\mathbf{X}^{N}_{t}\right)(X^{i,N}_{t},t) + \frac{\partial u}{\partial {X}_{i}}\left(\mathbf{X}^{N}_{t}\right)(X^{i,N}_{t},t)dX^{i,N}_{t}\right]=\frac{1}{c_{M}}\frac{\partial u}{\partial x}\left(\mathbf{X}^{N}_{t}\right)(X^{i,N}_{t},t).
\end{equation}
Therefore, Eq. \eqref{eq:optimalMovement} represents the optimal movement of worker $i$ in the limit when its time horizon in the optimization is going to zero; the latter explains the absence of the other workers' movements in Eq. \eqref{eq:optimalMovement}. From a temporal perspective, worker $i$ is myopic, also because mobility costs mainly appear as sunk costs. However, it is not myopic with respect to the spatial distribution of workers, since the latter are included in $u$. Since all workers are optimizing their movements at the same time taking the other workers' movements as given, the aggregate outcome of Eq. \eqref{eq:optimalMovement} appears as a \textit{Nash Equilibrium} of a non-cooperative game. In the next Section \ref{app:NashEquilibriumAgentsMovement}, we will discuss as this intuition is correct for a quadratic utility function.

\subsection{Workers' movements as a Nash equilibrium \label{app:NashEquilibriumAgentsMovement}}

In this section, we explicitly derive Eq. \eqref{eq:optimalMovement} for a quadratic utility function as the optimal movement of worker $i$ when it takes also into account the movement of other workers, i.e. Eq. \ref{eq:optimalMovement} is actually a Nash equilibrium for an economy where workers are maximizing step by step in discrete time, when time step of decisions is going to zero.

\subsubsection{Optimal movement in discrete time}

Assume now that, to maximize its personal utility at each time $t$, worker $i$ makes a step $dt\,V^{i,N}_{t}$ over a time interval of duration $dt$, subject to \textit{quadratic movement costs}. However, since its utility also depends on the locations of the other workers, in this maximization procedure it has to keep into account the \textit{expected behaviour} of all the other workers.
Therefore, worker $i$ at every time $t = 0,dt,2dt,\dots$ solves the following optimization problem:
\begin{equation}\label{eq:optDiscreteTimeExpectedMovement}
\max_{\{V^{i,N}_{t}\in \RR\}} u(\mathbf{X}^{N}_{t} +  dt\,\mathbf{V}^{N,e_{i}}_{t})(X^{i,N}_{t}+ dt\,V^{i,N}_{t},t+dt) - dt\frac{c_{M}}{2}\norm{V^{i,N}_{t}}^{2},
\end{equation}
where $\mathbf{V}^{N,e_{i}}_{t} := (V^{1,N,e_{i}}_{t},\dots, V^{N,N,e_{i}}_{t})$ and $V^{i,N,e_{i}}_{t} := V^{i,N}_{t}$. Roughly speaking, at each time $t$ worker $i$ has a certain belief of where other workers will go after $dt$, represented by the vector $\mathbf{V}^{N,e_{i}}_{t}$.

The first order condition for Problem \eqref{eq:optDiscreteTimeExpectedMovement} is:
\begin{eqnarray}\label{eq:FOCoptDiscreteTimeExpectedMovement}
0 &=& \sum_{j=1,j \neq i}^{N}\frac{\partial u}{\partial {X}_{j}}(\mathbf{X}^{N}_{t} +  dt\,\mathbf{V}^{N,e_{i}}_{t})(X^{i,N}_{t}+ dt\,V^{i,N}_{t},t+dt) \frac{dV^{j,N,e_{i}}}{dV^{i,N}_{t}}dt + \nonumber \\ 
&+& \frac{\partial u}{\partial {X}_{i}}(\mathbf{X}^{N}_{t} +  dt\,\mathbf{V}^{N,e_{i}}_{t})(X^{i,N}_{t}+ dt\,V^{i,N}_{t},t+dt)dt+ \nonumber \\ 
&+& \frac{\partial u}{\partial x}(\mathbf{X}^{N}_{t} +  dt\,\mathbf{V}^{N,e_{i}}_{t})(X^{i,N}_{t}+ dt\,V^{i,N}_{t},t+dt)dt - dt c_{M}V^{i,N}_{t}.
\end{eqnarray}
Assume that the expected behaviour of worker $j$, $V^{j,N,e_{i}}_{t}$, does not depend on the optimal behaviour of worker $i$, $V^{i,N}_{t}$; then all terms in the first summation vanish. Moreover, since Eq.  \eqref{eq:utilityNoEffectOnItself} holds, any change in the location of worker $i$ has no direct effect on the field of utility evaluated in its location since the only thing that matters is its relative location with respect to others. Therefore, Eq. \eqref{eq:FOCoptDiscreteTimeExpectedMovement} reduces to:
\begin{equation}\label{eq:FOCoptDiscreteTimeExpectedMovementSimplified}
\frac{\partial u}{\partial x}(\mathbf{X}^{N}_{t} +  dt\,\mathbf{V}^{N,e_{i}}_{t})(X^{i,N}_{t}+ dt\,V^{i,N}_{t},t+dt) -  c_{M}V^{i,N}_{t}  = 0.
\end{equation}
To explicitly solve Eq. \eqref{eq:FOCoptDiscreteTimeExpectedMovementSimplified}, assume that worker's utility has the following shape:
\begin{equation}\label{eq:utilitySimplest}
u(\mathbf{X})(x,t) = \frac{1}{N}\sum_{j=1}^{N} W(x-\mathbf{X}_{j}),
\end{equation}
where the \textit{kernel function} $W$ is defined as
\begin{equation}
W(z) = \frac{1}{2}\left(1- z^2\mathds{1}_{\abs{z}\leq 1}-\mathds{1}_{\abs{z}>1} \right)
\label{eq:kernelFunction}
\end{equation}
($\mathds{1}_{\abs{z}\leq 1} $ means that the value is equal to 1 when $\abs{z}\leq 1$ and 0 otherwise).
\begin{figure}[htbp]
\begin{center}
\includegraphics[width=0.48\textwidth]{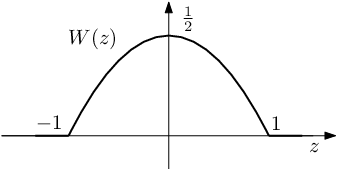}
\end{center}
\caption{The shape of kernel function $W(z)$ in Eq. \eqref{eq:kernelFunction}, the base of individual utility of Eq. \eqref{eq:utilitySimplest}.}
\label{fig:Wx}
\end{figure}
Figure \ref{fig:Wx} reports the shape of $W(z)$; in particular, $W^\prime(0) = 0$ in agreement with the Eq. \eqref{eq:utilityNoEffectOnItself}, while the partial derivative of $u$ with respect to $\mathbf{X}^i$ is always zero when evaluated in $\mathbf{X}^i$.
The utility function \eqref{eq:utilitySimplest} is higher when all the workers are close to each other; in particular, it reaches its maximum value of $1/2$ when all workers are in the same location, i.e. when the level of agglomeration is maximum, which in Figure \ref{fig:Wx} corresponds to $z=0$. 

We can now rewrite Eq. \eqref{eq:FOCoptDiscreteTimeExpectedMovementSimplified} as:
\begin{equation}\label{eq:FOCoptDiscreteTimeExpectedMovementSimplifiedStep1}
V^{i,N}_{t} = -\frac{1}{c_{M}}\frac{1}{N}\sum_{j=1}^{N}\left[(X^{i,N}_{t}+dtV^{i,N}_{t})-(X^{j,N}_{t}+dtV^{j,N,e_{i}}_{t})\right]\mathds{1}_{\left\{\abs{(X^{i,N}_{t}+dtV^{i,N}_{t})-(X^{j,N}_{t}+dtV^{j,N,e_{i}}_{t})}\leq 1\right\}}.
\end{equation}
Assume now that workers at time $t$ are all within distance 1 from each other, and that the same holds at time $t+dt$, so that all the indicator functions in the previous sum are equal to 1, which is not restrictive if $dt$ is very small, so that the movement that each worker will perform is very small.
Therefore, from Eq. \eqref{eq:FOCoptDiscreteTimeExpectedMovementSimplifiedStep1} we have 
\begin{equation}\label{eq:FOCoptDiscreteTimeExpectedMovementSimplifiedStep2}
V^{i,N}_{t} = -\frac{1}{c_{M}}\left[X^{i,N}_{t} - \frac{1}{N}\sum_{j=1}^{N}X^{j,N}_{t}\right] - \frac{dt}{c_{M}} \left[V^{i,N}_{t} - \frac{1}{N}\sum_{j=1}^{N}V^{j,N,e_{i}}_{t}\right],
\end{equation}
from which
\begin{equation}\label{eq:FOCoptDiscreteTimeExpectedMovementSimplifiedStep3}
V^{i,N}_{t} = -\frac{1}{c_{M}+dt-dt/N}\left[X^{i,N}_{t} - \frac{1}{N}\sum_{j=1}^{N}X^{j,N}_{t} - \frac{dt}{N}\sum_{j=1,j\neq i}^{N}V^{j,N,e_{i}}_{t}\right].
\end{equation}
By considering a continuous time process, instead of a discrete one, i.e. letting $dt$ go to zero, the expected choices of other workers do not affect the optimal choice of worker $i$ (as seen from worker $i$). The intuition is the following: in continuous time the choice of worker $i$ is repeated infinitely many times, and each of these choices, taken singularly (and not as a continuous choice process) has no effect on the dynamic. The same effect holds for the choice of all other workers.
Therefore, either let $dt$ go to zero, or assume that worker $i$ has zero-movement expectation on other workers, i.e. $V^{j,N,e_{i}} = 0$ for $j \neq i$, the optimal strategy to maximize utility \eqref{eq:utilitySimplest} is to move towards the barycentre of spatial distribution of workers, which is the mean location of workers. This is coherent with the intuition that the particular shape of utility \eqref{eq:utilitySimplest} should tend to aggregate workers over space.\footnote{See Online Appendix \ref{app:increasingSpatialAgglomeration} for a proof.} Moreover, the barycentre of the location of the workers remains unchanged after each forward step, i.e. the barycentre is an \textit{invariant} for the dynamic.

\subsubsection{Nash equilibrium}

In each time $t$, the equilibrium where any worker $i$ solves Problem \eqref{eq:optDiscreteTimeExpectedMovement} and expectation are realized, i.e. $V^{j,N,e_{i}}_{t} = V^{j,N}_{t}$, is a Nash equilibrium of a non-cooperative game.
We have therefore to solve a system of $N$ linear equations derived by Eq. \eqref{eq:FOCoptDiscreteTimeExpectedMovementSimplifiedStep2}, i.e.: \begin{equation}\label{eq:FOCoptDiscreteTimeExpectedMovementSimplifiedNash1}
V^{i,N}_{t} = -\frac{1}{c_{M}}\left[X^{i,N}_{t} - \frac{1}{N}\sum_{j=1}^{N}X^{j,N}_{t}\right] - \frac{dt}{c_{M}} \left[V^{i,N}_{t} - \frac{1}{N}\sum_{j=1}^{N}V^{j,N}_{t}\right],\quad i=1,\dots, N,
\end{equation}
that is:
\begin{equation}\label{eq:FOCoptDiscreteTimeExpectedMovementSimplifiedNash2}
V^{i,N}_{t}(1+(N-1)b) - b\sum_{j=1,j\neq i}^{N}V^{j,N}_{t} = -\frac{1}{c_{M}}\left[X^{i,N}_{t} - \frac{1}{N}\sum_{j=1}^{N}X^{j,N}_{t}\right], \quad i=1,\dots, N,
\end{equation}
where $b = dt/\left(c_{M}N\right)$.
System of Eqq. \eqref{eq:FOCoptDiscreteTimeExpectedMovementSimplifiedNash2} can be rewritten in matrix-vector form as $A \mathbf{V}^{N}_{t} = \tilde{\mathbf{X}}^{N}_{t}$, where: 
\begin{equation}
A := \begin{bmatrix}
1+(N-1)b & -b &  & \cdots & -b \\
-b & 1+(N-1)b & -b & \cdots & -b \\
\vdots &  & \ddots &  & \vdots \\
\vdots &  &  & \ddots & -b \\
-b & \cdots & \cdots & -b & 1+(N-1)b 
\end{bmatrix},\quad
\mathbf{V}^{N}_{t} := \begin{bmatrix}
V^{1,N}_{t} \\
V^{2,N}_{t} \\
\vdots \\
\vdots \\
V^{N,N}_{t} 
\end{bmatrix},
\end{equation}
\begin{equation}
\tilde{\mathbf{X}}^{N}_{t} = \begin{bmatrix}
-\frac{1}{c_{M}}\left(X^{1,N}_{t} - \frac{1}{N}\sum_{j=1}^{N}X^{j,N}_{t}\right) \\
-\frac{1}{c_{M}}\left(X^{2,N}_{t} - \frac{1}{N}\sum_{j=1}^{N}X^{j,N}_{t}\right) \\
\vdots \\
\vdots \\
-\frac{1}{c_{M}}\left(X^{N,N}_{t} - \frac{1}{N}\sum_{j=1}^{N}X^{j,N}_{t}\right)
\end{bmatrix}.
\end{equation}
One can verify that the inverse of the matrix $A$ is given by
\begin{equation}
A^{-1} = \frac{1}{Nb+1}\begin{bmatrix}
b+1 & b &  & \cdots & b \\
b & b+1 & b & \cdots & b \\
\vdots &  & \ddots &  & \vdots \\
\vdots &  &  & \ddots & b \\
b & \cdots & \cdots & b & b+1 
\end{bmatrix},
\end{equation}
so that the solution of System of Eqq. \eqref{eq:FOCoptDiscreteTimeExpectedMovementSimplifiedNash2} is given by:
\begin{equation} \label{eq:finalMovementNashEquilibrium}
V^{i,N}_{t} = -\frac{1}{c_{M}+dt} \left[X^{i,N}_{t} - \frac{1}{N}\sum_{j=1}^{N}X^{j,N}_{t}\right].
\end{equation}
The Nash equilibrium is analogous to the case where worker $i$ has zero-movement belief on the other workers' movement (compare Eqq. \eqref{eq:finalMovementNashEquilibrium} and \eqref{eq:FOCoptDiscreteTimeExpectedMovementSimplifiedStep3}). 
Finally, letting $dt$ go to zero, we find that the optimal movement of Eq. \ref{eq:finalMovementNashEquilibrium} reduces to follow the gradient of the utility function $u$ of Eq. \eqref{eq:optimalMovement} in the case utility function is given by Eq. \eqref{eq:utilitySimplest}.

\subsection{Increasing spatial agglomeration in the sequence of Nash equilibria \label{app:increasingSpatialAgglomeration}}

Below, we show that in the sequence of Nash equilibria in discrete time described by Eq. \eqref{eq:finalMovementNashEquilibrium} (or in its continuous time limit), the spatial agglomeration of workers increases over time, when as measure of spatial agglomeration is used the \textit{Gini Index} in a spatial framework defined as
\begin{equation}\label{eq:indicatorAggregation}
G^S(\mathbf{X}) := \frac{1}{2 \mathbf{\overline{X}}}\left(\frac{1}{N}\sum_{i=1}^{N} \sum_{j=1}^{N} \abs{\mathbf{X}_{i}-\mathbf{X}_{j}}\right),
\end{equation}
where 
\begin{equation}
\mathbf{\overline{X}} := \frac{1}{N}\sum_{i=1}^{N}\mathbf{X}^i
\end{equation}
is the empirical mean of the location of the workers. $G^S(\mathbf{X})$ has the remarkable property to be bounded between 0 and 1, with \textit{0 is the maximum level of agglomeration} (all workers in the same location) while \textit{1 is the minimum level of agglomeration} (no worker shares the same location with another worker). 

We first consider the level of aggregation at time $t$, $G^S(\mathbf{X}^{N}_{t})$ and show that each term of the sum in Eq. \eqref{eq:indicatorAggregation} gets smaller after a time step. Since the value of the mean location remains unchanged with each time step, the term $\mathbf{\overline{X}}^N_t$ is irrelevant for the calculation.
Considering \eqref{eq:finalMovementNashEquilibrium} we have:
\begin{multline}
\abs{X^{i,N}_{t+dt}-X^{j,N}_{t+dt}} = \abs{\left(X^{i,N}_{t} + dt V^{i,N}_{t}\right) - \left(X^{j,N}_{t} + dt V^{j,N}_{t}\right)} = \\
= \abs{\left(X^{i,N}_{t} - X^{j,N}_{t}\right) + dt\left(V^{i,N}_{t}-V^{j,N}_{t}\right)} = \abs{\left(X^{i,N}_{t} - X^{j,N}_{t}\right) -\frac{dt}{c_{M}+dt}\left(X^{i,N}_{t} - X^{j,N}_{t}\right)} = \\
= \abs{\left(1-\frac{dt}{c_{M}+dt}\right)\left(X^{i,N}_{t} - X^{j,N}_{t}\right)} = \left(1-\frac{dt}{c_{M}+dt}\right)\abs{\left(X^{i,N}_{t} - X^{j,N}_{t}\right)} \leq \abs{\left(X^{i,N}_{t} - X^{j,N}_{t}\right)}
\end{multline}
since the term $\left(1-\frac{dt}{c_{M}+dt}\right)$ is always positive and smaller than one. 
Since $\abs{X^{i,N}_{t+dt}-X^{j,N}_{t+dt}} \leq \abs{\left(X^{i,N}_{t} - X^{j,N}_{t}\right)}$, the spatial agglomeration after one step is higher, in the sense that $G^S(\mathbf{X})$ is smaller than that at the initial time.

\section{Micro-foundation of endogenous amenities \label{app:endoAmenities}}
Eq. \eqref{eq:endoAmenities} can be derived by assuming that individuals devote a constant share $\tau \in (0,1)$ of their income to finance the consumption of endogenous amenities; the total amount of expenditure is given by $\tau \cdot y^N(x,t)$. The production of amenities is subject to decreasing marginal productivity; hence the total amount of supplied amenities at location $x$ at time $t$ is:
\begin{equation}
PA^N(x,t) = \left[\tau \cdot y^N(x,t)\right]^\phi,
\end{equation}
where $\phi \in \left(0,1\right)$.
An alternative/complementary explanation for $PA(x,t)$ is that local endogenous amenities are financed by a flat tax rate $\tau \in (0,1)$ on income; in this case, $\tau \cdot y^N(x,t)$ would represent the collected revenues to finance the provision of amenities.
Taking the same model for congestion, an opportune rescaling of parameters leads to Eq. \eqref{eq:endoAmenities}.

\section{Increase in the number of workers \label{app:increasingNumberAgents}}
In this section we describe how to treat the case of the increasing number of workers.
Let $N^N_t$ be the number of workers that are present at time $t$, with $N^N_0 = N$, that is at time $t = 0$ $N$ workers are present.
We allow that the aggregate stock of labour is increasing over time with the total number of workers $N^N_t$, i.e. at period $t$ the aggregate stock of labour is $N^N_t/N \geq 1$.
Assume that in each location there exists a \textit{birth rate} $n\left(x,t\right) \geq 0$ specific of location $x$ at time $t$; a possible death rate has a similar characteristic and it is omitted in the analysis. New workers are introduced by means of a Poisson process with a \textit{time-space non-homogeneous} birth rate $n\left(X^{i,N}_t,t\right)$ for worker $i$ at period $t$, where the first jump of the process corresponds to the duplication of worker $i$. The offspring is created in the same location of its ``mother'', but with an independent noise. 

\begin{figure}[tbph]
\centering
\includegraphics[width=0.6\linewidth]{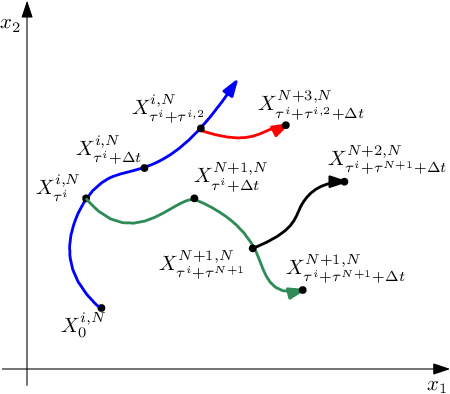}
\caption{The path of duplication for worker $i$ and its offspring in the space $\Omega$.}
\label{fig:movementagentiii}
\end{figure}

Figure \ref{fig:movementagentiii} reports an example of two duplications for the offspring of worker $i$. The time span between two duplications are denoted by $\tau^i$ for the first duplication of worker $i$, by $\tau^{i,2}$ for the second duplication of worker $i$, and so on. In Figure \ref{fig:movementagentiii} three new workers are generated in a period of length $\max\{\tau^i + \tau^{i,2},\tau^{i} + \tau^{N+1}\}$. The length in space that separates two consecutive duplications cannot be easily mapped into the time between the two duplications, depending on the speed of motion of worker $i$ and the local birth rate of the locations visited by worker $i$. 
We don't consider the presence of an exogenous death rate, which however can be treated within the same framework; hence, in the presence of a local mortality rate, $n\left(x,t\right)$ should be meant as the local workforce growth rate.

\section{Representative worker\label{app:representativeWorker}}

From Theorem \ref{teo:limitInfiniteAgents} is possible to go back to the dynamics of workers; in particular, Theorem \ref{teo:representativeMovement} describes the dynamics of a \textit{representative}  worker, i.e. a worker whose dynamics over space and time follows the \textit{average} dynamics of all workers present in the economy.
\begin{teo}[Dynamics of the Representative Worker] \label{teo:representativeMovement}
The dynamics of the representative worker is given by:
\begin{equation}\label{eq:representativeAgent}
d\bar{X}_t = \left( \frac{1}{c_{M}} \right) \nabla_{x} \bar{u}\left(\bar{X}_{t},t\right) dt,
\end{equation} 
where
\begin{equation}
\label{eq:utilityReppresentativeWorkers}
\bar{u}(x,t) = v(x,t) +\sigma x \cdot \frac{dB_{t}}{dt}
\end{equation}
is the random utility of the representative worker at location $x$, where $v(x,t)$ is defined in Eq. \eqref{eq:systematicUtility}. In particular, the process $\bar{X}_t$ has probability density function $l(t,x)$, i.e. the solution of Eq. \eqref{eq:dynamicsLabourDistribution}. Therefore, with a little abuse of notation, the expected movement of the representative worker is given by:
\begin{equation}\label{eq:expectedMovementRepresentative}
\EE{d\bar{X}_t } =  \left( \frac{2}{c_{M}} \right)\int_{\Omega} l(t,x) \nabla_x u(t,x), \,dx ,
\end{equation}
where $u(x,t)$ is defined in Eq. \eqref{eq:teoTotalUtility}. 
\end{teo} 
Eq. \eqref{eq:representativeAgent}, called \textit{McKean-Vlasov SDE}, is characterized by the presence of the probability density distribution of the process into the equation itself. Existence and uniqueness of the solutions to the McKean-Vlasov equation can be obtained in the limit of a large number of workers as for Theorem \ref{teo:limitInfiniteAgents}. In particular, \cite{mishura2020existence} contain a general result of the existence and uniqueness (both in the weak and strong sense) with minimal assumption on the drift and under regularity assumption on the diffusion (satisfied in our case of constant diffusion). In the case of $\beta = 1$ and $A_0=0$, \cite{sznitman1991topics}, instead, frames the analysis within the classical theory of \textit{Propagation of Chaos}.

Eq. \eqref{eq:expectedMovementRepresentative} can be made rigorous by Eq. \eqref{eq:dynamicsLabourDistributionUtility}, as it follows: 
\begin{multline}\label{eq:expectedMovementRepresentativeProof}
\EE{d\bar{X}_t } := \lim_{\tau \to 0}\frac{1}{\tau}\left(\EE{\bar{X}_{t+\tau}} - \EE{\bar{X}_t}\right) = \lim_{\tau \to 0}\frac{1}{\tau}\left[ \int_\Omega x l(x,t+\tau)\, dx - \int_\Omega x l(x,t)\, dx\right] = \\ 
= \int_\Omega x \partial_t l(x,t) \,dx  =-\left(\frac{1}{c_M}\right) \int_\Omega x  \, \div_x \left( l(x,t) \nabla_x u(x,t)\right) 
= \left(\frac{2}{c_M}\right) \int_\Omega  l(x,t) \nabla_x u(x,t) \,dx.
\end{multline}
Eq. \eqref{eq:expectedMovementRepresentativeProof} clarifies that our definition of the expected movement of representative worker can be traced to the difference over time of the average position of workers, which, in turn, is directly related to the differentials in utility over space ($\nabla_x u(x,t)$) weighted by the spatial distribution of workers. 
The SED is characterized by the equalization of the \textit{expected} utility of the representative worker as stated in Proposition \ref{prop:longRunEquilibrium}.
\begin{prop}[Pareto Optimality of Stationary Equilibrium Distribution] \label{prop:longRunEquilibrium}
In the SED, the following equality must hold:
\begin{equation}\label{eq:utilityRapresentativeWorkerLongRunDistribution}
\EE{d\bar{X}_t} = \EE{\nabla_{x} \bar{u}^{EQ}\left(\bar{X}_{t}\right)}=\int_\Omega  l^{EQ}(x) \nabla_x u^{EQ}(x) \,dx=0,
\end{equation}
where $\bar{u}^{EQ}$ is the utility of the representative worker in the SED.
\end{prop}
\begin{proof}
For the proof, see Eq. \eqref{eq:expectedMovementRepresentativeProof} imposing $\partial_t l(x,t)=0 \; \forall x$ and Eq. \eqref{eq:representativeAgent}.  	
\end{proof}
Proposition \ref{prop:longRunEquilibrium} suggests a form of \textit{Pareto optimality} for the SED in the case of risk-neutral workers, because no representative worker can increase its expected utility given the spatial allocation of the other workers.

\section{One megacity with high diffusion \label{app:megacityHighDiffusion}}

Figure \ref{fig:megacityBignoise} reports the SED for $\sigma=0.2$. Workers distribution is still single-peaked but with higher spatial dispersion. Wages mirror the workers density, as well as all the other variables of interest (Figure \ref{fig:distributionOverSpace_I}). Systematic utility presents a non negligible gradient over space, reflecting the strength of the random component in the total utility, i.e. of unexplained heterogeneity in workers' preferences.
\begin{figure}[!htbp]
\centering
\includegraphics[width=0.7\textwidth]{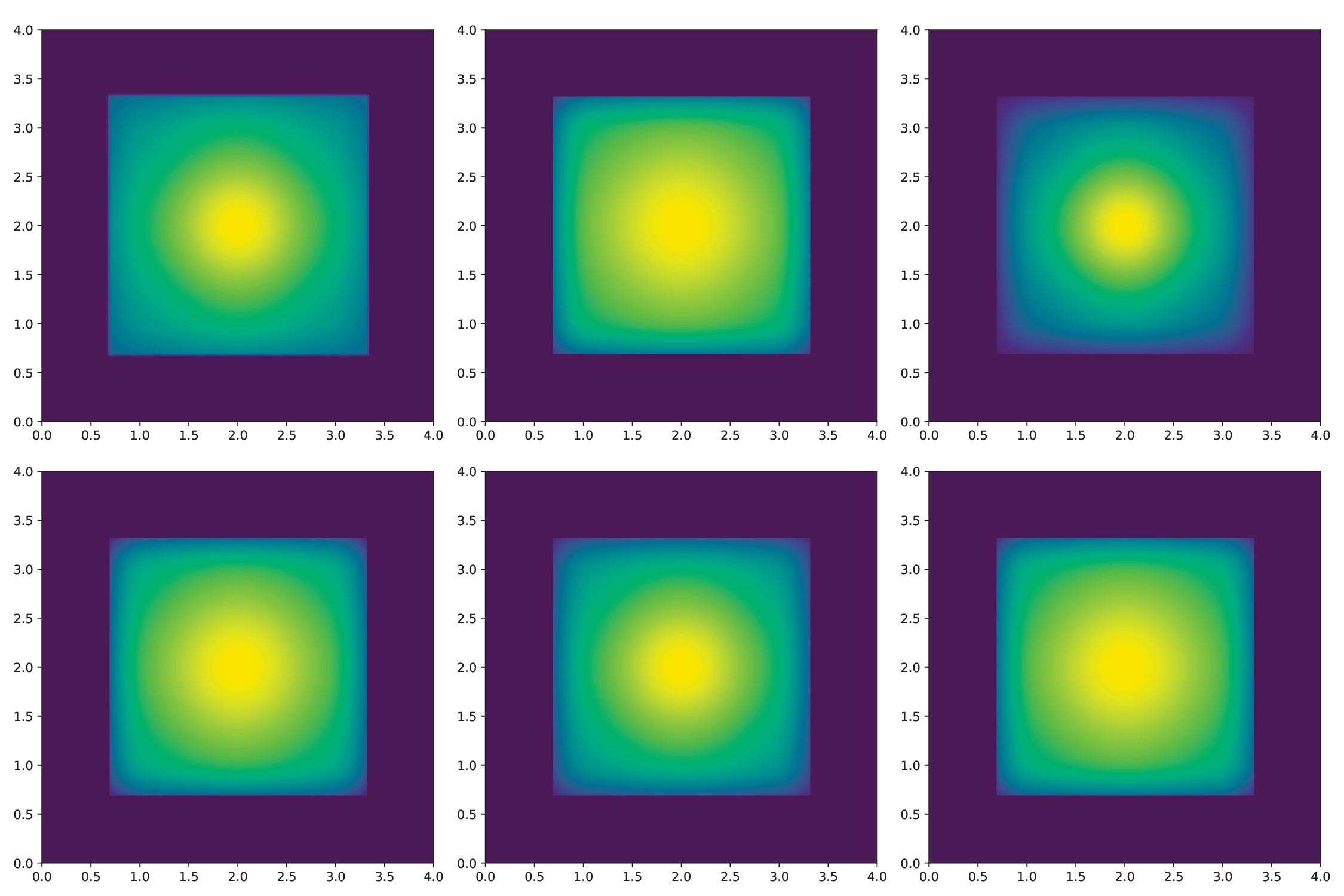}
\caption{The distribution of $l(x,t)$ (workers, top left), $w(x,t)$ (wages, top center), $y(x,t)$ (income, top right), $A_{EN}(x,t)$ (endogenous amenities, bottom left), $A_l(x,t)$ (technological progress, bottom center), $v(x,t)$ (individual systematic utility, bottom right) over space  $\Omega=[0,4]\times[0,4]$ at $t=20$ for setting reported in Table \ref{tab:parameterValues} but with $\sigma=0.2$.}
\label{fig:megacityBignoise}
\end{figure}

\end{document}